%% file: main.tex
\documentclass[journal, two column]{IEEEtran}
\usepackage{amsmath, amssymb, amsthm, bm}
\usepackage{amssymb}
\usepackage{graphicx}
\usepackage{color, colortbl} 
\usepackage[dvipsnames,svgnames,x11names]{xcolor}
\makeatother
\usepackage{tabularx, boldline, multirow}
\usepackage[boxed,ruled,vlined,linesnumbered]{algorithm2e}
\usepackage{enumerate}
\usepackage[list=true]{subcaption}

\usepackage[font={footnotesize}]{caption}
\usepackage{cite}
\usepackage[normalem]{ulem}
\usepackage{enumitem}
\usepackage{pgfplots}
\usepgfplotslibrary{fillbetween}
\pgfplotsset{compat=1.17}
\usetikzlibrary{plotmarks}
\usepackage{tikz}
\usetikzlibrary {arrows.meta} 
\usepackage{pgfgantt}
\usepackage{mathtools}
\usepackage{accents}
\usepackage{balance}
\usepackage{acro}
\usepackage{arydshln}
\usepackage{scalerel}

\usepackage{titlesec}
\titlespacing*{\section}{0pt}{10pt plus 2pt minus 2pt}{2pt plus 2pt minus 2pt}
\titlespacing*{\subsection}{0pt}{8pt plus 1pt minus 1pt}{2pt plus 1pt minus 1pt}
\titlespacing*{\subsubsection}{0pt}{5pt plus 1pt minus 1pt}{3pt plus 1pt minus 1pt}

\newtheorem{theorem}{Theorem}[]
\newtheorem{lemma}[theorem]{Lemma}
\newtheorem{proposition}{Proposition}
\newtheorem{remark}{Remark}

\newcommand{\rev}[1]{\textcolor{black}{#1}}

\usepackage[left=0.61 in,right=0.61 in,top=0.7 in, bottom=1 in]{geometry}
\include{Acronym}

\begin{document}

\title{Uplink Cell-Free Massive MIMO OFDM with Phase Noise-Aware Channel Estimation: Separate and Shared \rev{Local Oscillators}
\thanks{
A limited subset of initial results was presented at  IEEE GLOBECOM 2023, Kuala Lumpur, Malaysia~\cite{wu2023phasenoise}.

Yibo~Wu is with Ericsson Research, Ericsson, and the Department of Electrical Engineering, Chalmers University of  Technology, Gothenburg, Sweden (email: yibo.wu@ericsson.com).

Luca~Sanguinetti is with Dipartimento di Ingegneria dell’Informazione, University of Pisa, Pisa, Italy (email: luca.sanguinetti@unipi.it)

Ulf~Gustavsson is with Ericsson Research, Ericsson, Gothenburg, Sweden (e-mail: ulf.gustavsson@ericsson.com).

Musa~Furkan~Keskin, Alexandre~Graell~i~Amat, and Henk~Wymeersch are with the Department of Electrical Engineering, Chalmers University of Technology, Gothenburg, Sweden (emails: \{furkan, alexandre.graell, henkw\}@chalmers.se).

This work was supported by the Swedish Foundation for Strategic Research (SSF grant ID19-0021),  the Swedish Research Council (VR grants 2022-03007 \rev{and 2024-04390}), and scholarships from Chalmersska forskningsfonden and S.o.KG. Eliassons minnes- och tilläggsfonder. The computations were enabled by resources provided by the National Academic Infrastructure for Supercomputing in Sweden (NAISS), partially funded by the Swedish Research Council through grant agreement no. 2022-06725.  L. Sanguinetti was partially supported by the Italian Ministry of Education and Research (MUR) in the framework of the FoReLab project (Departments of Excellence), and by the Resilience Plan (NRRP) of NextGenerationEU, partnership on “Telecommunications of the Future” (PE00000001 – Program “RESTART”, Structural Project 6GWINET, Cascade Call SPARKS).
  }
}
\author{Yibo~Wu,~\IEEEmembership{Graduate~Student~Member,~IEEE},
Luca~Sanguinetti,~\IEEEmembership{Fellow,~IEEE},
Musa~Furkan~Keskin,~\IEEEmembership{Member,~IEEE},
    Ulf~Gustavsson,
Alexandre~Graell~i~Amat,~\IEEEmembership{Senior~Member,~IEEE} and
    Henk~Wymeersch~\IEEEmembership{Fellow,~IEEE}
        }
    \maketitle
\thispagestyle{empty}

\begin{abstract}
Cell-free \ac{mmimo} networks enhance coverage and spectral efficiency (SE) by distributing antennas across access points (APs) with phase coherence between APs. However, the use of cost-efficient local oscillators (LOs)  introduces phase noise (PN) that compromises phase coherence, even with centralized processing. Sharing an LO across APs can reduce costs in specific configurations  but cause correlated PN between APs, leading to correlated interference that affects centralized combining. This can be improved by exploiting the PN correlation in channel estimation.  This paper presents an uplink \ac{ofdm} signal model for PN-impaired cell-free \ac{mmimo}, addressing gaps in single-carrier signal models. We evaluate mismatches from applying single-carrier methods to OFDM systems, showing how they underestimate the impact of PN and produce over-optimistic achievable SE predictions. Based on our OFDM signal model, we propose two PN-aware channel and \acl{cpe} estimators: a distributed estimator for uncorrelated PN with separate LOs and a centralized estimator with shared LOs. We introduce a deep learning-based channel estimator to enhance the performance and reduce the number of iterations of the centralized estimator. The simulation results show that the distributed estimator outperforms mismatched estimators with separate LOs, whereas the centralized estimator enhances distributed estimators with shared LOs.
\end{abstract}

\section{Introduction}
\Acf{mmimo} networks enhance network capacity and reliability by coherently transmitting signals across multiple antennas, which improves \acf{se}~\cite{bjornson2016massive}. In cell-free \ac{mmimo} networks, antennas are distributed across multiple \acp{ap}, which collaborate to serve multiple \acp{ue}, effectively eliminating traditional cell boundaries. This coordination reduces inter-cell interference, enhances coverage, and significantly improves \ac{se}~\cite{cell_free_book}. These gains require coherent transmission, demanding time, frequency, and phase synchronization between distributed APs, which poses challenges like high overhead and limited scalability~\cite{bjornson2020scalable}. In particular, \ac{pn} introduced by imperfect \acp{lo} degrades system performance if not properly managed.

To address the synchronization challenge in cell-free \ac{mmimo} networks, groups of APs can be connected via cables, as seen in the radio stripes scenario~\cite{interdonato2019ubiquitous,shaik2021mmse}. This setup allows groups of APs to share a common oscillator, thereby simplifying the synchronization process. However, a common \ac{lo} introduces correlated \ac{pn} across APs, leading to correlated PN interference during centralized uplink combining or downlink precoding. Studies have shown that such correlated PN significantly degrades system performance, even with distributed mitigation techniques~\cite{bjornson2015massive,papazafeiropoulos2021scalable}. Despite this correlation challenge, the correlation of PN also opens the opportunity for centralized mitigation in the \ac{cpu}, a potential not yet fully explored in cell-free \ac{mmimo} studies.

\subsection{Motivation}
This paper seeks to investigate the impact of PN on cell-free \ac{mmimo} \ac{ofdm} networks, considering scenarios with both separate and shared LOs among APs. While PN effects in both cellular and cell-free \ac{mmimo} systems have been widely studied~\cite{bjornson2015massive, pitarokoilis2016performance, pitarokoilis2014uplink, papazafeiropoulos2021scalable, zheng2023asynchronous, jin2020spectral, ozdogan2019performance}, they primarily use single-carrier models. The achievable rates derived from these models tend to be overly optimistic when applied to \ac{ofdm} systems, due to underestimating the actual impact of \ac{pn}. Additionally, applying their methods, such as channel estimators, uplink combiners, and downlink precoders, to \ac{ofdm} models leads to inaccuracies and results in performance degradation. The effect of these inaccuracies on network design—including the number of \acp{ap} and \acp{ue}, \ac{lo} quality, and power allocation—is unclear and requires thorough investigation for practical OFDM models.

Specifically, among the works with the single-carrier PN model, \cite{bjornson2015massive} and~\cite{papazafeiropoulos2021scalable} analyze the effects of PN for both co-located and distributed antenna scenarios, highlighting the issue of correlated PN interference with shared LOs in co-located antenna scenarios. However, their works do not offer solutions to the correlated PN interference issue. \rev{This correlated PN phenomenon can also arise in distributed massive MIMO systems, such as radio stripe networks~\cite{interdonato2019ubiquitous,shaik2021mmse}, where APs are connected by cables and can share one or a small set of LOs to reduce hardware expenses. Although correlated PN leads to correlated interference, the shared LOs structure can be leveraged in a centralized way for enhanced channel and PN estimation, helping to alleviate the resulting interference.} Although~\cite{krishnan2014impact,pitarokoilis2016performance} study the impact of \ac{pn} in OFDM systems, they assume perfect channel knowledge, limiting practical use.  Other studies have limited scope or employ suboptimal techniques, such as overlooking the statistical characteristics of correlated PN across the APs~\cite{jin2020spectral,zheng2023asynchronous,ozdogan2019performance}, or using only maximum likelihood channel estimation~\cite{pitarokoilis2014uplink} and \ac{mr} combining~\cite{pitarokoilis2014uplink,jin2020spectral}, which perform poorly compared to the \ac{mmse} combiner in cell-free \ac{mmimo}~\cite{bjornson2019making}. Additionally, \cite{zheng2023asynchronous} uses rate-splitting techniques to mitigate PN effects, but their techniques become infeasible with a large number of UEs and APs. The work~\cite{ozdogan2019performance} only considers deterministic phase shifts rather than stochastic PN. 

Although PN mitigation in \ac{ofdm} systems is not new~\cite{petrovic2007effects}, no systematic study focuses on accurately modeling and mitigating PN in cell-free \ac{mmimo} OFDM systems, considering both separate and shared LO scenarios. In response to this research gap, this paper examines the impact of \ac{pn} on uplink cell-free \ac{mmimo} OFDM systems. Specifically, we aim to investigate the effects of mismatched PN-aware solutions from single-carrier models applied to OFDM models in cell-free \ac{mmimo} networks, quantifying the impact on achievable rate predictions and network parameters. Moreover, we propose new PN-aware channel estimation algorithms for both separate and common LO scenarios to mitigate uncorrelated and correlated \ac{pn} interference, where the latter is missing in the cell-free \ac{mmimo} literature.

\subsection{Contributions}
Our specific contributions include:\footnote{This paper significantly expands upon our previous conference work~\cite{wu2023phasenoise} by incorporating the shared common LO scenario, proposing a novel centralized estimator to address the PN correlation issue, and providing extensive comparisons with the mismatched single-carrier estimator, demonstrating their limitations and overly optimistic performance predictions in OFDM systems.}
\begin{itemize}
    \item \textbf{Uplink OFDM signal model with both correlated and uncorrelated PN in cell-free \ac{mmimo}:} We develop an uplink cell-free \ac{mmimo} OFDM signal model that incorporates both uncorrelated and correlated PN, by extending the uncorrelated PN model from \ac{siso} systems~\cite{petrovic2007effects}. This OFDM model with PN fills a gap in existing single-carrier signal model \ac{mmimo} studies with PN~\cite{papazafeiropoulos2021scalable,bjornson2015massive}. Our model accurately captures the impact of PN on channel aging.

    \item \textbf{Uplink achievable SE expression:}
     We derive a novel uplink achievable SE expression for cell-free \ac{mmimo} OFDM networks under PN, which accurately reflectes the actual impact of PN in OFDM models compared to the single-carrier achievable SE expressions  in~\cite{papazafeiropoulos2021scalable}. This provides realistic performance predictions that guide network design decisions about the number of APs and UEs, LO quality, and power allocation.
     
    \item \textbf{Centralized Channel and \acf{cpe} Estimator for Shared LO:}
    We propose a centralized channel and \ac{cpe} estimator to address the correlated PN problem in the case of a shared common LO among APs. This estimator alternates between distributed channel estimation and centralized LMMSE CPE estimation, improving the accuracy by exploiting the PN correlation.
    \item \textbf{Two Distributed PN-Aware Channel 
 and CPE Estimators for Separate LOs:}
    We propose two distributed channel estimators for scenarios with separate LOs among APs. The first estimator, based on LMMSE, jointly estimates the channel and \ac{cpe}, demonstrating enhanced performance compared to the single-carrier estimator~\cite{papazafeiropoulos2021scalable,bjornson2015massive} across various OFDM scenarios.  The second is a \ac{dl}-based estimator, serving as an initializer for the centralized estimator. This estimator improves the accuracy of channel estimation and decreases the number of iterations needed for the centralized estimator.
\end{itemize}

\subsection{{Paper Outline and Notation}}
The remainder of this paper is organized as follows. Section II covers preliminaries, including the channel model and pilot assignment in the cell-free \ac{mmimo} OFDM network. Section III details the system model, encompassing the PN model and both accurate and mismatched signal models under PN. Section IV describes the proposed distributed PN-aware LMMSE joint channel and CPE estimator and the deep learning channel estimator. Section V introduces the centralized alternative channel and CPE estimator. Section VI derives a novel achievable uplink SE expression. Section VII provides simulation results, comparing the proposed estimators with existing methods, highlighting improvements in SE and channel estimation accuracy, and exposing the optimistic bias of the mismatched signal model. Finally, Section VIII concludes with key findings and future research directions.

Lowercase and uppercase boldface letters, e.g., $\boldsymbol{x}$ and $\boldsymbol{X}$, denote column vectors and matrices. The superscripts $^{\mathsf{T}}$, $^*$, $^{\mathsf{H}}$, and $^{\dagger}$ denote transpose, conjugate, conjugate transpose, and pseudo-inverse, respectively. Variables with a check mark $\check{}$, e.g., $\check{x}$, are in the time-domain. $\check{x}_n$ is the $n$-th time sample in the time-domain, and ${x}_n$ is the $n$-th subcarrier in the frequency-domain. The $N \times N$ identity matrix is $\mathbf{I}_N$, and $\text{diag}(\boldsymbol{x})$ is a diagonal matrix with $\boldsymbol{x}$ on the diagonal. The expected value of $\boldsymbol{x}$ is $\mathbb{E}\{\boldsymbol{x}\}$.

\section{Preliminaries}
We consider a cell-free  \ac{ofdm} network comprising $L$ randomly distributed single-antenna \acp{ap}, connected to a \ac{cpu} via a fronthaul network and serving $K$ single-antenna \acp{ue}.\footnote{\rev{A single-antenna configuration for both APs and UEs is valid for cell-free mMIMO, as analyzed in~\cite{bjornson2019making}. In this setup, the many distributed antennas at the APs collectively form a large virtual array, providing array gains comparable to those in conventional cellular \ac{mmimo} systems, e.g.~\cite{bjornson2019making}.} } Each \ac{ofdm} symbol consists of $N$ subcarriers with spacing $\Delta_f$. A \ac{cp} length of $N_{\text{CP}}$ is considered. The signal bandwidth is $W= N\Delta_f$ so that the sampling time is $T_s= 1/W$. The \ac{ofdm} symbol time is $T=(N+N_{\text{CP}})T_s$.

\begin{figure}[t]
    \centering
    \input{Coherence_block_diagram}
    \caption{Illustration of coherence block $r$ on the time-frequency plane with $\tau_c N_c$ channel uses, where $\tau_p=20$ are for pilot transmission using pilot sequence $\boldsymbol{s}^\text{p}_{t_k} \in \mathbb{C}^{20}$ for UE $k$, following pilot pattern PP1 and PP2 with a pilot sample $s^{(\tau)}_{t_k,n}$ at subcarrier $n$ and OFDM symbol $\tau$. }
    \label{fig:coherent_block}
\end{figure}
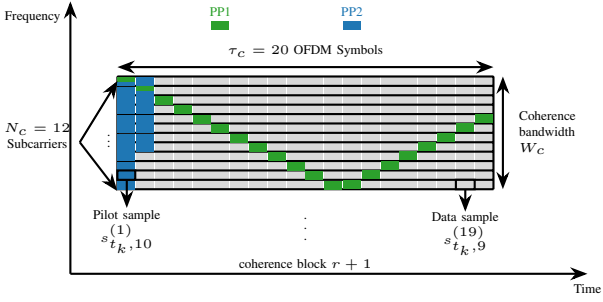

\nopagebreak
\subsection{Block Fading Channel Model}
The time-domain channel $\check{\boldsymbol{h}}_{k,l} = [\check{h}_{k,l,0}, \cdots, \check{h}_{k,l,Q-1}]^{\mathsf{T}}$ between \ac{ue} $k$ and \ac{ap} $l$ is modeled as a $Q$-tap \acl{fir} filter.
The frequency-domain channel is obtained by an $N$-point \ac{dft} on $\check{\boldsymbol{h}}_{k,l}$. We adopt the \ac{mmimo} TDD protocol from~\cite{bjornson2017massive}, which assumes that the time-frequency resources are divided into coherence blocks with time-invariant and frequency-flat channels within each block. This is shown in Fig.~\ref{fig:coherent_block} for an arbitrary coherence block $r$ with two different pilot patterns. Each coherence block has coherence time $T_c = \tau_c T$ and coherence bandwidth $W_c=N_{c}\Delta f$, \rev{spanning over $\tau_c$ OFDM symbol time and $N_c$ subcarriers.  The total number of samples in a coherence block is $(\tau_c N_c)$, including $\tau_p$ pilot samples and  $(\tau_c N_c- \tau_p)$ data samples.} Subcarriers in each \ac{ofdm} symbol are divided into $R=\lceil N/N_c \rceil$ coherence blocks, with the subcarrier set in block $r \in \{1,\cdots, R\}$ given by $\mathcal{R}_r=\{(r-1)N_{c},\cdots, rN_{c}-1\}$. The frequency-domain channel between UE $k$ and AP $l$ over subcarrier $n \in \mathcal{R}_r$ is denoted by $h_{k,l,n} \sim \mathcal{N}_{\mathbb{C}}(0,\beta_{k,l})$, where $\beta_{k,l}$ is the large-scale fading coefficient.

Under the coherence block fading channel model, we have $h_{k,l,n_1} = h_{k,l,n_2}$ for ${n_1, n_2} \in \mathcal{R}_r$ and all $\tau_c$ OFDM symbols. Channels of different \acp{ue}, APs, and coherence blocks are uncorrelated~\cite{bjornson2017massive,bjornson2020scalable}, i.e.,
\begin{align}
    \mathbb{E}\{h_{k_1,l,n} h_{k_2,l,n}\}= \mathbb{E}\{h_{k,l_1,n} h_{k,l_2,n}\} =0
    \label{eq:channel_uncorr}
\end{align}
for  $k_1\neq k_2$, or $l_1\neq l_2$, and $\mathbb{E}\{h_{k,l,n_1} h_{k,l,n_2}\} = 0$
for $n_1, n_2 \notin \mathcal{R}_r$. 



\subsection{Uplink Pilot Assignment} \label{sec:Uplink_pilot}

We use a widely adopted cell-free \ac{mmimo} pilot book setup as in~\cite{bjornson2020scalable}, which includes $\tau_p$ orthogonal pilot sequences of length $\tau_p$ in the frequency domain. These pilot sequences are grouped into a pilot set $\mathcal{S}^\text{p} = \{{\boldsymbol{s}}_{1}^\text{p}, {\boldsymbol{s}}_{2}^\text{p}, \ldots, {\boldsymbol{s}}_{\tau_p}^\text{p}\}$, where ${\boldsymbol{s}}_{t}^\text{p}\in \mathbb{C}^{\tau_p}$, $\|{\boldsymbol{s}}_{t}^\text{p}\|^2=\tau_p$, and ${\boldsymbol{s}}_{t_1}^{^\text{p}\mathsf{H}} {\boldsymbol{s}}_{t_2}^\text{p}=0$ for $t_1 \neq t_2$. Each \ac{ue} distributes its pilot sequence using the same pilot pattern in a coherence block as shown in Fig.~\ref{fig:coherent_block}. When $K>\tau_p$, UEs have to share pilots, which leads to pilot contamination~\cite{bjornson2017massive}.  For an arbitrary coherence block $r$, the set of subcarrier and OFDM symbol indices used for pilot transmission is denoted by $\mathcal{N}_{p}  =\{n_1,\cdots,n_{\tau_p}\} \subseteq \mathcal{R}_r$, and $\mathcal{T}_{p}=\{\tau_1,\cdots,\tau_{\tau_p}\} \subseteq  \{1,\cdots,\tau_c\}$, respectively. The set of remaining subcarrier indices for data transmission is denoted by $\mathcal{N}_{d} = \mathcal{R}_r \backslash \mathcal{N}_{p}$.  \ac{ue} $k$ is assigned with pilot sequence ${\boldsymbol{s}}_{t_k}^\text{p} \in \mathcal{S}_p$, which is distributed in a coherence block over pilot subcarriers $ n\in \mathcal{N}_{p}$ and OFDM symbols $\tau \in \mathcal{T}_p$, i.e., $\boldsymbol{s}_{t_k}^\text{p} =[{s}_{t_k,n_1}^{(\tau_1)},\cdots,{s}_{t_k,n_{\tau_p}}^{(\tau_p)}]^{\mathsf{T}}$. The $\tau$-th OFDM symbol of UE $k$ is ${\boldsymbol{s}}_{t_k}^{(\tau)} \in \mathbb{C}^{N}$, which consists of both pilot and data samples.  

\section{Phase Noise-Impaired Cell-free \ac{mmimo} OFDM}
\subsection{Phase Noise Model}
Imperfect hardware in \acp{ap} and \acp{ue} prevents phase synchronization between APs. Even with perfect synchronization, \ac{pn} would persist, varying over time due to \acp{lo} disparities. We model this PN in both time and frequency domains, considering separate and common LOs across APs.

\subsubsection{Phase Noise in the Time-Domain}
The \ac{pn} $\check{\phi}_{l,n}^{(\tau)}$ and $\check{\varphi}_{k,n}^{(\tau)}$ from the free-running \acp{lo} of AP $l$ and UE $k$, respectively, of time sample $n$ and OFDM symbol $\tau$, can be modeled as discrete-time Wiener processes~\cite[Eq. (7)]{petrovic2007effects}:
\begin{align}
\check{\phi}_{l, n}^{(\tau)}  &= \check{\phi}_{l, n-1}^{(\tau)}+ \check{\delta}_n^{\phi_l}, \\ 
\check{\varphi}_{k, n}^{(\tau)}  &=\check{\varphi}_{k, n-1}^{(\tau)}+\check{\delta}_n^{\varphi_k},
\end{align}
where $\check{\delta}_n^{\phi_l}\sim \mathcal{N}(0,\sigma_{\phi_l}^2)$ and $\check{\delta}_n^{\varphi_k}\sim \mathcal{N}(0,\sigma_{\varphi_k}^2)$ with variance
\begin{align}
    \sigma_i^2=4 \pi^2 f_{\mathrm{c}}^2 \gamma_i T_{\mathrm{s}}, \text{ for } i\in \{\phi_l, \varphi_k\}. \label{eq:LO_coe_define}
\end{align}
Here, $f_{\mathrm{c}}$ and $\gamma_i$ denote the carrier frequency and a constant describing the oscillator quality. Note that different APs and UEs may have different qualities. For simplicity, from now on we assume that all APs have the same quality LOs, and all UEs also have the same quality LOs.  The uplink received \ac{pn} from UE $k$ and AP $l$ at time-domain sample $m$ of OFDM symbol $\tau$ are combined as
\begin{align}
    \check{\theta}_{k,l,n}^{(\tau)} = \check{\varphi}_{k,n}^{(\tau)} + \check{\phi}_{l,n}^{(\tau)}.
\end{align}
The vector form for the whole OFDM symbol $\tau$ is denoted by $\check{\boldsymbol{\theta}}_{k,l}^{(\tau)}=[\check{\theta}_{k,l,0}^{(\tau)},\cdots,\check{\theta}_{k,l,N-1}^{(\tau)}]^{\mathsf{T}}$. Considering the \ac{cp}, the combined time-domain PN of UE $k$ and AP $l$ can be modeled as $\check{\theta}_{k,l,n_1}^{(\tau_1)} =  \check{\theta}_{k,l,n_2}^{(\tau_2)} + \mathcal{N}(0,|(\tau_1-\tau_2)(N+N_{\text{CP}})+n_1 -n_2| (\sigma_{\varphi_k}^2+\sigma_{\phi_l}^2))$. The time-domain multiplicative PN that affects the carrier signal is $\exp({j\check{\theta}^{(\tau)}_{k,l,n}})$, and its vector form for OFDM symbol $\tau$ is denoted by  $\exp{(j\check{\boldsymbol{\theta}}^{(\tau)}_{k,l})}$.

\begin{proposition}
    The auto-correlation of the \ac{pn} from UE $k$ and AP $l$ is
\begin{align}
    \mathbb{E}\{ e^{j\check{\theta}_{k,l,n_1}^{(\tau_1)}} e^{ -j\check{\theta}_{k,l,n_2}^{(\tau_2)}}\} = e^{ -\frac{\sigma_{\varphi_k}^2+\sigma_{\phi_l}^2}{2}\left(|(\tau_1-\tau_2)(N+N_{\text{CP}}) + n_1- n_2 |\right)}.
    \label{eq:PN_correlation_TD_same}
\end{align} 
\end{proposition}
\begin{proof}
    It follows by generalizing equation (19) in~\cite{petrovic2007effects} to the case with multiple OFDM symbols.
\end{proof}

\begin{proposition}
    The correlation of PN between different UEs ($k_1\neq k_2$) and \acp{ap} ($l_1\neq l_2$), for APs with \textit{separate} \acp{lo}, is
\begin{align}
     \mathbb{E}\{ e^{j\check{\theta}_{k_1,l_1,n_1}^{(\tau_1)}} e^{-j\check{\theta}_{k_2,l_2,n_2}^{(\tau_2)}}\} = \mathbb{E}\{ e^{j\check{\theta}_{k_1,l_1,n_1}^{(\tau_1)}}\} \mathbb{E} \{e^{-j\check{\theta}_{k_2,l_2,n_2}^{(\tau_2)}}\},
\label{eq:PN_correlation_TD_diff}
\end{align}
and for APs with a \textit{shared common} \ac{lo} is
\begin{align}
     \mathbb{E}\{ e^{j\check{\theta}_{k_1,l_1,n_1}^{(\tau_1)}} & e^{-j\check{\theta}_{k_2,l_2,n_2}^{(\tau_2)}}\} = e^{ - \frac{\sigma_{\rev{\phi_l}}^2}{2} \big( |(\tau_1-\tau_2)N + n_1   - n_2|\big)} \nonumber \\
&  \times e^{ -\frac{\sigma_{\varphi_{k_1}}^2}{2}( \tau_1 N + n_1)  -\frac{\sigma_{\varphi_{k_2}}^2}{2}(\tau_2 N + n_2)},  \label{eq:PN_correlation_TD_colocated}
\end{align}
where $\mathbb{E}\{ e^{j\check{\theta}_{k_1,l_1,n_1}^{(\tau_1)}}\}$ and $\mathbb{E}\{ e^{j\check{\theta}_{k_2,l_2,n_2}^{(\tau_2)}}\}$ in~\eqref{eq:PN_correlation_TD_diff} are given by~\cite{chorti2006spectral}
\begin{align}
    \mathbb{E}\{e^{j\check{\theta}^{(\tau)}_{k,l,n}}\} =  e^{-\frac{\sigma_{\varphi_k}^2+\sigma_{\phi_l}^2}{2} (\tau (N+N_{\text{CP}})+n)}, \label{eq:E_e_jtheta}
\end{align}
and the subscripts $l_1$ and $l_2$ in~\eqref{eq:PN_correlation_TD_colocated} are omitted \rev{since AP $l_1$ and $l_2$ share the same PN process when they share a common LO}.
\end{proposition}
\begin{proof}
   Equation~\eqref{eq:PN_correlation_TD_diff} follows since the PN between APs and UEs are uncorrelated. Equation~\eqref{eq:PN_correlation_TD_colocated} arises from the correlated PN at APs with a shared LO.
\end{proof}

\subsubsection{Phase Noise in the Frequency-Domain}
For any OFDM symbol $\tau$, the multiplicative \ac{pn} exp(${j\check{\boldsymbol{\theta}}_{k,l}^{(\tau)}}) \in \mathbb{C}^N$ in the time-domain is equivalent to the convolutional frequency-domain phase-drift vector
$\boldsymbol{J}_{k,l}^{(\tau)}\in \mathbb{C}^{N}$, whose $i$-th entry $J_{k,l,i}^{(\tau)}$, for $i=-N/2,\cdots,N/2-1$, is given by~\cite[Eq. (6)]{petrovic2007effects}
\begin{align}
	J_{k,l,i}^{(\tau)} = \frac{1}{N}\sum\limits_{n=0}^{N-1}e^{j\check{\theta}^{(\tau)}_{k,l,n}} e^{-j2\pi ni/N}. \label{eq:J_k_l_i_definition}
\end{align} 
For $i=0$, the phase-drift 
\begin{align}
J_{k,l,0}^{(\tau)}=\frac{1}{N}\sum_{n}e^{j\check{\theta}^{(\tau)}_{k,l,n}}
\end{align} 
is known as the \ac{cpe}~\cite{petrovic2007effects} since it is common to all subcarriers. The other non-zero phase-drifts $J_{k,l,i}^{(\tau)}$ for $ i\neq 0$ lead to \ac{ici}. Using~\eqref{eq:E_e_jtheta}, the expected value of $J_{k,l,i}^{(\tau)}$ can be calculated as
\begin{align}
     \mathbb{E}\{J_{k,l,i}^{(\tau)}\} \triangleq \bar{J}_{k,l,i}^{(\tau)} = \frac{1}{N}\sum\limits_{n=0}^{N-1}e^{-\frac{\sigma_{\varphi_k}^2+\sigma_{\phi_l}^2}{2} (\tau N+n)} e^{\frac{-j2\pi ni}{N}}.
     \label{eq:E_J}
\end{align}
For the same UE $k$ and AP $l$, the correlation between $J_{k,l,i_1}^{(\tau_1)}$ and $J_{k,l,i_2}^{(\tau_2)}$ is calculated as
\begin{align}
&\mathbb{E}\{J_{k,l,i_1}^{(\tau_1)} J_{k,l,i_2}^{*(\tau_2)}\} \triangleq {B}_{k,l,i_1,i_2}^{(\tau_1 -\tau_2)}= \nonumber\\
& \frac{1}{N^2} \sum_{n_1=0}^{N-1}\sum_{n_2=0}^{N-1}  \mathbb{E}\{ e^{j\check{\theta}_{k,l,n_1}^{(\tau_1)}} e^{-j\check{\theta}_{k,l,n_2}^{(\tau_2)}}\}  e^{\frac{-j2\pi(n_1i_1- n_2i_2)}{N}},\label{eq:B_matrix_compute}
\end{align}
where $\mathbb{E}\{ e^{j\check{\theta}_{k,l,n_1}^{(\tau_1)}} e^{j\check{\theta}_{k,l,n_2}^{(\tau_2)}}\}$ is obtained from~\eqref{eq:PN_correlation_TD_same}.
For different UEs ($k_1\neq k_2$) and APs ($l_1\neq l_2$), we can calculate the correlation between $J_{k_1,l_1,i_1}^{(\tau_1)}$ and $J_{k_2,l_2,i_2}^{*(\tau_2)}$ as
\begin{align}
    &\mathbb{E}\{J_{k_1,l_1,i_1}^{(\tau_1)} J_{k_2,l_2,i_2}^{*(\tau_2)}\} \triangleq \breve{B}_{k_1,k_2,l_1,l_2,i_1,i_2}^{(\tau_1 -\tau_2)}
   \nonumber \\& \frac{1}{N^2} \sum_{n_1=0}^{N-1}\sum_{n_2=0}^{N-1}  \mathbb{E}\{ e^{j\check{\theta}_{k_1,l_1,n_1}^{(\tau_1)}} e^{-j\check{\theta}_{k_2,l_2,n_2}^{(\tau_2)}}\}  
    e^{\frac{-j2\pi(n_1i_1- n_2i_2)}{N}}, \label{eq:B_matrix_compute_diff}
\end{align}
where $\mathbb{E}\{ \check{\theta}_{k_1,l_1,n_1}^{(\tau_1)} \check{\theta}_{k_2,l_2,n_2}^{(\tau_2)}\}$ can be obtained from~\eqref{eq:PN_correlation_TD_diff} for separate LOs and from~\eqref{eq:PN_correlation_TD_colocated} for a shared LO, respectively.
\subsection{Signal Model with Phase Noise}
\label{sec:signal_model}
\rev{We begin by presenting the PN-impaired uplink OFDM signal model for a single symbol in Section~\ref{sec:sig_model_OFDM_symbol} and then extend it to a vector form in Section~\ref{sec:sig_model_all_pilots} that covers pilot symbols across multiple OFDM symbols.}

\subsubsection{Signal Model for an OFDM Symbol}\label{sec:sig_model_OFDM_symbol}
For UE $k$, the $m$-th time-domain sample of $\tau$-th OFDM symbol is obtained by an IDFT on ${\boldsymbol{s}}_{t_k}^{(\tau)}$ as
		$\check{s}_{t_k,n}^{(\tau)} = \frac{1}{\sqrt{N}}\sum_{n=0}^{N-1} {s}_{t_k,n}^{(\tau)} \exp{({j2\pi nm}/{N})}$,
and its vector form for the whole OFDM symbol $\tau$ is denoted by $\check{\boldsymbol{s}}^{(\tau)}_{t_k} \in \mathbb{C}^{N}$.  
With the time-domain multiplicative \ac{pn} $\text{diag}(\exp{(j\check{\boldsymbol{\theta}}_{k,l}^{(\tau)}}))$, the time-domain received  signal $\check{\boldsymbol{y}}_{l}^{(\tau)} \in \mathbb{C}^{N}$ at AP $l$ for the $\tau$-th OFDM symbol is
\begin{align}
	\check{\boldsymbol{y}}_{l}^{(\tau)} = \sum\limits_{k=1}^{K} \sqrt{p_k} \text{diag}(e^{j\check{\boldsymbol{\theta}}_{k,l}^{(\tau)}})(\check{\boldsymbol{h}}_{k,l} \circledast \check{\boldsymbol{s}}^{(\tau)}_{t_k} ) + \check{\boldsymbol{w}}_{l}^{(\tau)},\label{eq:y_kl_TD}
\end{align}
where $p_k \geq 0$ is the transmit power of UE $k$, $\circledast$ denotes the circular convolution, and $\check{\boldsymbol{w}}_{k,l}^{(\tau)} \sim \mathcal{N}_{\mathbb{C}}(\mathbf{0},\sigma^2\boldsymbol{I}_{N}) $ denotes the thermal noise. By applying a $N$-point DFT to both sides of \eqref{eq:y_kl_TD}, \rev{the frequency-domain received signal at AP $l$,~${\boldsymbol{y}}_{l}^{(\tau)}$, is}
\begin{align}
		{\boldsymbol{y}}_{l}^{(\tau)} = \sum\limits_{k=1}^{K} \sqrt{p_k} \boldsymbol{J}_{k,l}^{(\tau)} \circledast ( {\boldsymbol{h}}_{k,l} \odot {\boldsymbol{s}}^{(\tau)}_{t_k} ) + {\boldsymbol{w}}_{l}^{(\tau)},\label{eq:y_kl_FD}
\end{align}
where $\odot$ denotes the Hadamard product, and the thermal noise follows the same distribution, ${\boldsymbol{w}}_{l}^{(\tau)} \sim \mathcal{N}_{\mathbb{C}}(\mathbf{0},\sigma^2\boldsymbol{I}_{N})$. The elements of the phase-drift  $\boldsymbol{J}_{k,l}^{(\tau)}$ are defined in~\eqref{eq:J_k_l_i_definition}.
\rev{The $n$-th element, ${y}_{l,n}^{(\tau)}$, of ${\boldsymbol{y}}_{l}^{(\tau)}$ in~\eqref{eq:y_kl_FD} is
the frequency-domain sample received on subcarrier $n$. It can be decomposed as~\cite[Eq. (5)]{petrovic2007effects}}
\rev{
\begin{align}
	{y}_{l,n}^{(\tau)} &=  \sum\nolimits_{k=1}^{K} \Big( \sqrt{p_k}  \sum\nolimits_{\jmath=0}^{N-1}  {J_{k,l,n-\jmath}^{(\tau)}}{{h}}_{k,l,\jmath} {{s}}^{(\tau)}_{t_k,\jmath}\Big) + {{w}}_{l,n}^{(\tau)} \nonumber\\
&= \sum\nolimits_{k=1}^{K}\Big(\sqrt{p_k}   {{s}}^{(\tau)}_{t_k,n}   \underbrace{J_{k,l,0}^{(\tau)}{{h}}_{k,l,n}}_{\triangleq {{h}}_{k,l,n}^{(\tau)}}   \nonumber \\ &\quad \quad \quad \quad + \underbrace{\sqrt{p_k} \sum\nolimits_{\substack{\jmath =0\\ \jmath \neq n}}^{N-1} {{s}}^{(\tau)}_{t_k,\jmath} J_{k,l,n-\jmath}^{(\tau)} {{h}}_{k,l,\jmath} }_{\triangleq\zeta_{k,l,n}^{(\tau)}}\Big) +    {{w}}_{l,n}^{(\tau)},\nonumber \\
&= \sum\limits_{k=1}^{K}\Big(\sqrt{p_k}   {{s}}^{(\tau)}_{t_k,n}   {{h}}_{k,l,n}^{(\tau)}   +  \zeta_{k,l,n}^{(\tau)}\Big) +    {{w}}_{l,n}^{(\tau)},
 \label{eq:FD_recv_tau}
\end{align}
}
where
\begin{align} \label{eq:effective_channel}
    {{h}}_{k,l,n}^{(\tau)} \triangleq J_{k,l,0}^{(\tau)}{{h}}_{k,l,n}
\end{align}
is the time-varying \textit{effective channel} and
\begin{align}
 \zeta_{k,l,n}^{(\tau)} = \sqrt{p_k} \sum\limits_{\substack{\jmath =0, \jmath \neq n}}^{N-1} {{s}}^{(\tau)}_{t_k,\jmath} J_{k,l,n-\jmath}^{(\tau)} {{h}}_{k,l,\jmath}
\end{align}
is the \ac{ici} component over subcarrier $n$ and \ac{ofdm} symbol $\tau$. The variance of the effective channel is $\text{Var}({{h}}_{k,l,n}^{(\tau)})= \mathbb{E}\{|J_{k,l,0}^{(\tau)}|^2\} \mathbb{E}\{|h_{k,l,n}|^2\}= B_{k,l,0,0}^{(0)} \beta_{k,l}$. 
Within an arbitrary coherence block $r$, the effective channels are the same only for subcarriers in the same \ac{ofdm} symbol,
\begin{align}
    {{h}}_{k,l,n_1}^{(\tau_1)} &= {{h}}_{k,l,n_2}^{(\tau_1)}, \text{ for } \{n_1,n_2\} \in \mathcal{R}_{r},  \label{eq:effec_channel_relation_1} \\
    {{h}}_{k,l,n_1}^{(\tau_1)} &\neq {{h}}_{k,l,n_2}^{(\tau_2)}, \text{ for } \{n_1,n_2\} \notin \mathcal{R}_{r}, \text{ or } \tau_1\neq \tau_2.
    \label{eq:effec_channel_relation_2}
\end{align}
Note that the \ac{cpe} coefficient $J_{k,l,0}^{(\tau)}$ of \ac{ofdm} symbol $\tau$ is the same for all subcarriers. 

\subsubsection{Signal Model for All Pilots} \label{sec:sig_model_all_pilots}
Focusing on coherence block $r$, i.e., subcarriers $n\in \mathcal{R}_r$, the received $\tau_p$-length pilot vector at AP $l$ can be collected following~\eqref{eq:FD_recv_tau} as
\begin{align}
    {\boldsymbol{y}}_{l}^{\text{p}} &\triangleq [{y}_{l,n_1}^{(\tau_1)},\cdots,{y}_{l,n_{\tau_p}}^{(\tau_p)}]^{\mathsf{T}} \in \mathbb{C}^{\tau_p}, 
    \\
    &= \sum_{k=1}^{K} (\sqrt{p_k} \boldsymbol{s}_{t_k}^\text{p} \odot  \boldsymbol{h}_{k,l}^{\text{E-p}} + \boldsymbol{\zeta}_{k,l}^{\text{p}}) + \boldsymbol{w}_{l}, \label{eq:y_l_pilot_vector}
\end{align}
where  the effective channel vector between UE $k$ and AP $l$ for $\tau_p$ pilots is defined as
\begin{align}
    \boldsymbol{h}_{k,l}^{\text{E-p}}& \triangleq [{h}_{k,l,n_1}^{(\tau_1)},\cdots,{h}_{k,l,n_{\tau_p}}^{(\tau_p)}]^{\mathsf{T}}  = \boldsymbol{J}_{k,l}^{\text{p}} h_{k,l,n},\label{eq:h_eff_k_l}
    \end{align}
    and the \ac{ici} vector for $\tau_p$ pilots is denoted by $\boldsymbol{\zeta}_{k,l}^{\text{p}} = [{\zeta}_{k,l,n_1}^{(\tau_1)}, \cdots, {\zeta}_{k,l,n_{\tau_p}}^{(\tau_p)}]^{\mathsf{T}} \in \mathbb{C}^{\tau_p}$. The effective channel vector for pilots, $\boldsymbol{h}_{k,l}^{\text{E-p}}$,  consists of the CPE vector for $\tau_p$ OFDM symbols with pilots, $\boldsymbol{J}_{k,l}^{\text{p}} \triangleq [J_{k,l,0}^{(\tau_1)}, \cdots, J_{k,l,0}^{(\tau_p)}]^{\mathsf{T}} \in \mathbb{C}^{\tau_p}$, and the channel $h_{k,l,n}$ with $n$ representing an arbitrary subcarrier $\in \mathcal{R}_r$ due to~\eqref{eq:effec_channel_relation_1}. Additionally, for all $\tau_c$ OFDM symbols in the coherence block $r$, the corresponding \ac{cpe} and effective channel vector is  defined as $\boldsymbol{J}_{k,l}^{\text{c}} \triangleq [{J}_{k,l,0}^{(1)},\cdots,{J}_{k,l,0}^{(\tau_c)}]^{\mathsf{T}} \in \mathbb{C}^{\tau_c}$, and $ \boldsymbol{h}_{k,l}^{\text{E-c}} \triangleq [{h}_{k,l,n}^{(1)},\cdots,{h}_{k,l,n_{}}^{(\tau_c)}]^{\mathsf{T}} \in \mathbb{C}^{\tau_c}$, respectively.

The channel vector from UE $k$ to all $L$ APs over subcarrier $n \in \mathcal{R}_r$ is collected as 
\begin{align}
    \boldsymbol{h}_{k,n} \triangleq [h_{k,1,n},\cdots,h_{k,L,n}]^{\mathsf{T}} \in \mathbb{C}^{L}.
\end{align}
Collecting $\boldsymbol{h}_{k,n}$ for all $K$ UEs, the channels between $K$ UEs and $L$ APs over subcarrier $n \in \mathcal{R}_r$ is
\begin{align}
    \boldsymbol{h}_{n} \triangleq [\boldsymbol{h}_{1,n}^{\mathsf{T}},\cdots,\boldsymbol{h}_{K,n}^{\mathsf{T}}]^{\mathsf{T}} \in \mathbb{C}^{LK}. \label{eq:channel_all}
\end{align}
The \ac{cpe} vector for pilot OFDM symbols between all $L$ APs and UE $k$ is 
\begin{align}
    \boldsymbol{J}_{k}^{\text{p}} \triangleq [\boldsymbol{J}_{k,1}^{\text{p} \mathsf{T}},\cdots,\boldsymbol{J}_{k,L}^{\text{p} \mathsf{T}}]^{\mathsf{T}} \in \mathbb{C}^{\tau_p L}, \label{eq:J_CPE_k_tau_p}
\end{align}
and the ICI vector collected from all $L$ APs is $
\boldsymbol{\zeta}_{k}^{\text{p}} \triangleq [\boldsymbol{\zeta}_{k,1}^{\text{p} \mathsf{T}},\cdots,\boldsymbol{\zeta}_{k,L}^{\text{p} \mathsf{T}}]^{\mathsf{T}} \in \mathbb{C}^{\tau_p L}$.
In the case of a common \ac{lo} shared by all APs, the CPE vector~\eqref{eq:J_CPE_k_tau_p} reduces to $\boldsymbol{J}_{k}^{\text{p}} =  \boldsymbol{1}_{L\times 1} \otimes \boldsymbol{J}_{k,1}^{\text{p}}$. 

Finally, the received pilots at all $L$ \acp{ap} are collected as
\begin{align}
    \boldsymbol{y}^{\text{p}} \triangleq [\boldsymbol{y}_{1}^{\text{p}\mathsf{T}},\cdots,\boldsymbol{y}_{L}^{\text{p} \mathsf{T}}]^{\mathsf{T}} \in \mathbb{C}^{L\tau_p}, \label{eq:y_pilot_vector}
\end{align}
which can be expressed as (following~\eqref{eq:y_l_pilot_vector}) 
\begin{align}
    \boldsymbol{y}^{\text{p}} = \sum_{k=1}^{K} \big(\sqrt{p_k} ( \boldsymbol{h}_{k,n}^{\text{}} \otimes \boldsymbol{s}_{t_k}^{\text{p}}) \odot \boldsymbol{J}_{k}^{\text{p}} + \boldsymbol{\zeta}_{k}^{\text{p}} \big) + \boldsymbol{w}, \label{eq:observ_all} 
\end{align}
where $\otimes$ denotes the Kronecker product.

\subsection{Mismatched Signal Model with Phase Noise}
The works~\cite{bjornson2015massive,papazafeiropoulos2021scalable,zheng2023asynchronous,jin2020spectral} analyze the impact of PN for single-carrier systems under the single-carrier time-domain PN model. Applying these solutions to OFDM systems can lead to inaccurate channel estimators, uplink combiners, and downlink precoders, thereby impacting system performance. Also, this model mismatch results in overly optimistic conclusions, such as achievable rate predictions, due to underestimating the impact of PN. Next, we analyze how this PN mismatch in OFDM systems affects the modeling of channel aging effects within a coherence block.

Specifically, applying the time-domain single-carrier PN model to the frequency-domain OFDM signal model leads to the \rev{\emph{mismatched}}  signal model
\begin{align}
    		{\boldsymbol{y}}_{l}^{(\tau)} = \sum\limits_{k=1}^{K} \sqrt{p_k} \text{diag}(e^{j\check{\boldsymbol{\theta}}_{k,l}^{(\tau)}}) ( {\boldsymbol{h}}_{k,l} \odot {\boldsymbol{s}}^{(\tau)}_{t_k} ) + {\boldsymbol{w}}_{l}^{(\tau)}.\label{eq:y_kl_FD_mismatched}
\end{align}
The corresponding mismatched signal model for sample ${{y}}_{l,n}^{(\tau)}$ received over subcarrier $n$ can be decomposed as
\begin{align}
    {y}_{l,n}^{(\tau)} =  \sum\limits_{k=1}^{K}\Big(  \sqrt{p_k}   \underbrace{e^{j\check{\theta}_{k,l,n}^{(\tau)}} {h}_{k,l,n}}_{ \triangleq \ddot{{h}}_{k,l,n}^{(\tau)}} {{s}}^{(\tau)}_{t_k,n}    \Big) +    {{w}}_{l,n}^{(\tau)},
    \label{eq:FD_recv_tau_mismatch}
\end{align}
where the mismatched effective channel, compared to $h_{k,l,n}^{(\tau)}$ in~\eqref{eq:effective_channel}, is defined as
\begin{align}
        \ddot{{h}}_{k,l,n}^{(\tau)} \triangleq e^{j \check{\theta}_{k,l,n}^{(\tau)}} {{h}}_{k,l,n}.
    \label{eq:effective_channel_mismatched}
\end{align}

One critical issue with the mismatched channel model is that it indicates that $\ddot{{h}}_{k,l,n}^{(\tau)}$ varies across different subcarriers even within the same coherence block and OFDM symbol, i.e., 
\begin{align}
    \ddot{{h}}_{k,l,n_1}^{(\tau_1)} \neq  \ddot{{h}}_{k,l,n_2}^{(\tau_1)} \neq \ddot{{h}}_{k,l,n_1}^{(\tau_2)} \text{, for } \{n_1,n_2\} \in \mathcal{R}_{r}, \tau_1 \neq \tau_2,
    \label{eq:effec_channel_mismatch_relations}
\end{align}
which differs from~\eqref{eq:effec_channel_relation_1} and~\eqref{eq:effec_channel_relation_2}. Another critical issue is the incorrect noise statistics due to ignoring the ICI component. This arises from applying the time-domain \ac{pn} model to the frequency-domain signal model. \rev{Consequently, deriving channel estimation schemes based on the mismatched signal model~\eqref{eq:FD_recv_tau_mismatch} leads to mismatched \ac{pn}-aware channel estimators as in~\cite{bjornson2015massive,papazafeiropoulos2021scalable,zheng2023asynchronous,jin2020spectral,ozdogan2019performance} in OFDM systems.} This subsequently results in mismatched uplink combiners and reciprocity-based downlink precoders.

The aforementioned mismatch disappears in a single carrier system (i.e., $N=1$), where the PN model is equivalent in both time and frequency domains, i.e., $e^{j\theta_{k,l,0}^{(\tau)}}=J_{k,l,0}^{(\tau)}$ and $J_{k,l,n}^{(\tau)} = 0$. In this case, the mismatched effective channel $\ddot{{h}}_{k,l,0}^{(\tau)}={{h}}_{k,l,0}^{(\tau)}$, and the ICI component in~\eqref{eq:FD_recv_tau} disappears, reducing~\eqref{eq:FD_recv_tau} to:
\begin{align}
    	{y}_{l,0}^{(\tau)} &=  \sum\limits_{k=1}^{K}&\Big(  \sqrt{p_k}   \underbrace{J^{(\tau)}_{k,l,0} {h}_{k,l,n}}_{={{h}}_{k,l,0}^{(\tau)}=\ddot{{h}}_{k,l,n}^{(\tau)}} {{s}}^{(\tau)}_{t_k,n}    \Big) +    {{w}}_{l,n}^{(\tau)},
 \label{eq:FD_recv_tau_NoICI}
\end{align}
which is equivalent to~\eqref{eq:FD_recv_tau_mismatch}.

\section{PN-aware Distributed Channel Estimation}
Next, we derive a distributed \acl{pna} joint channel and CPE \ac{lmmse} estimator for scenarios with separate LOs between APs. Additionally, we introduce a distributed \acl{pna} \ac{dl} channel estimator that uses PN-impaired pilots for channel estimation. This nonlinear estimator is used to initialize the centralized estimator, which will be detailed in Section~\ref{sec:Central_Estimator}.

\subsection{Distributed Joint Channel and CPE Estimation}
From~\eqref{eq:FD_recv_tau} and~\eqref{eq:effective_channel}, it is evident that despite using a block fading channel model, \ac{pn} causes the effective channel $h_{k,l,n}^{(\tau)}$ to vary with each \ac{ofdm} symbol $\tau$, a phenomenon known as channel aging. Therefore, within each coherence block, the channel estimation needs to be performed for each OFDM symbol (a total of $\tau_c$ times) rather than for each channel use (a total of $\tau_c N$ times) as suggested by the works for single-carrier systems~\cite{bjornson2015massive,papazafeiropoulos2021scalable}, or just a single time as suggested in systems without PN~\cite{bjornson2020scalable,bjornson2019making}.

We now derive an estimator of the effective channel ${{h}}_{k,l,n}^{(\tau)}$ for any subcarrier $n \in \mathcal{R}_r$ and OFDM symbol $\tau \in \{1,\cdots,\tau_c\}$. The conventional \ac{mmse} estimation is challenging because the received observations~\eqref{eq:observ_all} and channels~\eqref{eq:channel_all} are not jointly Gaussian distributed due to the non-Gaussian distributed \ac{cpe}~\eqref{eq:J_CPE_k_tau_p}. Therefore, we derive an \ac{lmmse} estimator based on~\cite{kay1993fundamentals}.
\begin{lemma}
The \ac{lmmse} estimate of ${h}_{k,l,n}^{(\tau)}$ based on \rev{${\boldsymbol{y}}_{l}^{\text{p}}$ from~\eqref{eq:y_l_pilot_vector}} is
\begin{align}\label{eq:LMMSE_estimator_dist}
    \hat{{h}}_{k,l,n}^{(\tau)} = \sqrt{p_k}\beta_{k,l}{\boldsymbol{s}}_{t_k}^{\mathsf{H}} \Acute{\boldsymbol{B}}_{k,l}^{(\tau)} \boldsymbol{\Psi}_{l}^{-1} {\boldsymbol{y}}_{l}^{\text{p}},
\end{align}
where
\begin{align} 
\Acute{\boldsymbol{B}}_{k,l}^{(\tau)}&=\mathrm{diag}\left(\Big[{B}_{k,l,0,0}^{(\tau-\tau_1)} ,\cdots, {B}_{k,l,0,0}^{(\tau-\tau_p)} \Big]^{\mathsf{T}}\right) \label{eq:LMMSE_weights_matrix}
\end{align}
\begin{align}
\boldsymbol{\Psi}_{l} &= \sum\nolimits_{k^{\prime}=1}^{K}p_{k^{\prime}}\beta_{k^{\prime},l} \boldsymbol{\Phi}_{k^{\prime}} + \boldsymbol{Z}^{\text{ICI}}_{l} + \sigma^2 \mathbf{I}_{\tau_p} \label{eq:LMMSE_inverse_matrix} 
\\
[\boldsymbol{\Phi}_{k^{\prime},l}]_{\tau_1, \tau_2} &= {s}^{(\tau_1)}_{t_{k^{\prime}},n_1} {s}^{*(\tau_2)}_{t_{k^{\prime}},n_2} {B}_{k^{\prime},l,0,0}^{(\tau_1-\tau_2)}. \label{eq:Phi_joint}
\end{align}
Here   ${B}_{k,l,0,0}^{(\tau_1-\tau_2)}$ in~\eqref{eq:LMMSE_weights_matrix} and~\eqref{eq:Phi_joint} are defined in~\eqref{eq:B_matrix_compute} and  $\boldsymbol{Z}^{\text{ICI}}_{l}$  is a diagonal matrix defined in~\eqref{eq:Zeta_ICI_full_expression} in Appendix~\ref{Appendix:LMMSE_Joint}. \rev{In \eqref{eq:LMMSE_estimator_dist}, the estimated effective channel $\hat{{h}}_{k,l,n}^{(\tau)}$ is identical across all subcarriers for a given OFDM symbol $\tau$ and coherence block $r$. }
\end{lemma}

\begin{proof}
See Appendix \ref{Appendix:LMMSE_Joint}.
\end{proof}

In the ideal case of no \ac{pn}, resulting in $\Acute{\boldsymbol{B}}_{0,0}^{(\tau)}=\boldsymbol{I}_{\tau_p}$ and $\boldsymbol{Z}^{\text{ICI}}_{l}=\boldsymbol{0}$, the received samples and channels are jointly Gaussian distributed. Consequently, the \ac{lmmse} estimator~\eqref{eq:LMMSE_estimator_dist} becomes equivalent to the optimal \ac{mmse} estimator as in~\cite{bjornson2020scalable}.  
With the \ac{lmmse} estimation, the channel estimation $\hat{{h}}_{k,l,n}^{(\tau)}$ is of zero mean and variance $ \epsilon_{k,l,n}^{(\tau)} \triangleq p_k \beta_{k,l}^2 {\boldsymbol{s}}_{t_k,n}^{\mathsf{H}} \boldsymbol{B}_{0,0}^{(\tau)} \boldsymbol{\Psi}_{l}^{-1} \boldsymbol{B}_{0,0}^{\mathsf{H},(\tau)} {\boldsymbol{s}}_{t_k,n}^{\mathsf{}} $, and the channel estimation error  $\tilde{{h}}_{k,l,n}^{(\tau)}\triangleq {{h}}_{k,l,n}^{(\tau)} - \hat{{h}}_{k,l,n}^{(\tau)}$ is of zero mean and variance $ c_{k,l,n}^{(\tau)} \triangleq \beta_{k,l}B_{0,0}^{(0)} - \epsilon_{k,l,n}^{(\tau)} $. 

\begin{remark}
    The uncorrelated and zero-mean channels between each AP and UE make the distributed joint channel and CPE LMMSE estimator~\eqref{eq:LMMSE_estimator_dist} equal to its centralized version, as it cannot exploit PN correlation (see Appendix~\ref{Appendix:LMMSE_Joint}). Hence, despite feeding all observations $\boldsymbol{y}^{\text{p}}$ from all APs, the estimator in Lemma 1 still uses only the specific observations $\boldsymbol{y}^{\text{p}}_l$ for each AP $l$. Note that even with separate LOs, PN at APs is correlated for PN from the same UEs. 
\end{remark}
\begin{remark}
\textcolor{black}{
The PN-aware LMMSE estimator in~\eqref{eq:LMMSE_estimator_dist} is different from the one in~\cite{bjornson2015massive,papazafeiropoulos2021scalable,zheng2023asynchronous,jin2020spectral,ozdogan2019performance} as a multi-carrier PN model is considered instead of a single-carrier PN model, which leads to different impacts of PN statistics in~\eqref{eq:LMMSE_weights_matrix} and~\eqref{eq:LMMSE_inverse_matrix}.
}
\end{remark}

\subsection{Distributed Deep Learning Channel Estimation}
The \ac{lmmse} estimator~\eqref{eq:LMMSE_estimator_dist} is a linear estimator, but under \ac{pn}, a \ac{dl}-based nonlinear estimator can outperform it, as demonstrated for power amplifier nonlinearities~\cite{hu2020deep} and joint channel and PN estimation~\cite{mohammadian2021deep}. However, joint estimation requires true channel and PN data for training, which is practically infeasible to obtain separately. While real channels can be measured, obtaining true PN separately is practically infeasible. To address this, we propose a \ac{pn}-aware \ac{dl} estimator that focuses solely on channel estimation, using PN-impaired received pilots, thus eliminating the need for real PN data.

\rev{There are several \ac{nn} architectures in the literature for channel estimation, including \acl{cnn}~\cite{dong2019deep,mashhadi2021pruning,ma2020data} and fully connected NNs~\cite{hu2020deep}. Here, we focus on a simpler approach using a fully connected \ac{nn} with three dense layers containing $M_1$, $M_2$, and $2K$ neurons, respectively.} Additionally, it incorporates residual connections, a technique known to enhance learning speed in various applications~\cite{he2016deep,WuResidualDPD}. The structure of the DL estimator is shown in Fig.~\ref{fig:NN_structure}. Specifically, an arbitrary AP $l$ is equipped with a DL estimator $f_{\boldsymbol{\alpha}}(\cdot)$ that takes the received pilot vector $\boldsymbol{y}_l^{\text{p}}$ as input and produces the channel estimates $\hat{\boldsymbol{h}}_{l,n}$ between $K$ UEs and AP $l$ at subcarrier $n \in \mathcal{R}_r$ as the output
\begin{align}
    \hat{\boldsymbol{h}}_{l,n} = f_{\boldsymbol{\alpha}}(\boldsymbol{y}_l^{\text{p}}) \label{eq:DL_estimator}.
\end{align}
Here, $f_{\boldsymbol{\alpha}}(\cdot)$ represents the underlying function of the DL estimator, which includes the NN layers and activation functions. The learnable parameters, including the weights and biases, are denoted by $\boldsymbol{\alpha}$, and are trained using \rev{a \ac{mse} loss function between the NN output, $\{\hat{h}_{k,l,n}\}$ and the true channels, $\{{h}_{k,l,n}\}$, given by 
\begin{align}
    \mathcal{L} (\hat{\boldsymbol{h}}_{l,n}, {\boldsymbol{h}}_{l,n}) = \mathbb{E}\left\{ \big|\hat{{h}}_{k,l,n} - {{h}}_{k,l,n}\big|^2\right \},
\end{align}
where the expectation is computed by averaging over all $K$ UEs and all samples in a mini-batch.} 

\begin{remark}
    \rev{We train the DL estimators on a diverse set of channel realizations sampled from a specific distribution before deployment. This approach ensures that they do not rely on a single channel condition and can generalize effectively to new, unseen channels from the same channel distribution, even in rapidly changing environments. Re-training is required if the channel distribution changes.} While this requires extensive training data, fine-tuning a pre-trained estimator for similar scenarios can significantly reduce training overhead, as demonstrated in~\cite{hu2020deep}.  This approach alleviates the re-training issue when the scenario changes, such as variations in UE powers and PN variances.
    Training can be performed either distributively at each AP or centrally at the CPU. In the centralized approach, the CPU distributes trained parameters to each AP, enabling distributed inference. 
\end{remark}
\begin{figure}[!t]
	\vspace*{-0.5\baselineskip}
	\centering
	\input{NN_structure}
	\vspace*{-0.5\baselineskip}
	\caption{Structure of the DL channel estimator  $f_{\boldsymbol{\alpha}}(\cdot)$ using a fully-connected NN. The C2R and R2C blocks convert between complex and real signals. }
	\label{fig:NN_structure}
\end{figure}
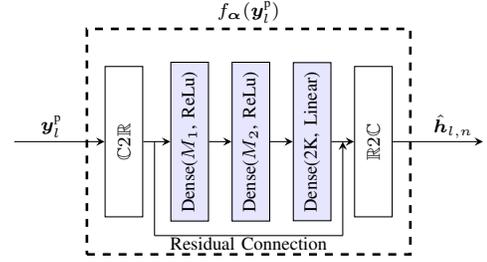

\newcolumntype{Y}{>{\centering\arraybackslash}X}
\newcolumntype{C}[1]{>{\centering\arraybackslash}p{#1}}
\begin{table*}[t] 
\centering
\caption{\rev{Computational complexity per coherence block for different channel and PN estimators for a generic UE $k$ when using the MMSE combining~\eqref{eq:v_k_MMSE}. Only complex multiplications and complex divisions are considered. $\mathcal{M}_k$ is defined in~\eqref{eq:M_k_define}.}}
\label{tab:complexity_estimators}
\rev{
\begin{tabularx}{1\linewidth}{p{0.21\linewidth}|C{0.2\linewidth}|C{0.23\linewidth}|C{0.23\linewidth}}
\hline
\multicolumn{2}{c|}{Estimator} &  Channel estimation & PN estimation \\ \hline
\multicolumn{2}{c|}{PN-unaware distributed channel MMSE~\cite[Eq. (3)]{bjornson2020scalable}}& $(\tau_p+3)K |\mathcal{M}_k|$ & \textemdash \\ \hline
\multicolumn{2}{c|}{Distributed joint channel and PN LMMSE~\cite[Eq.(12)]{papazafeiropoulos2021scalable}}& \multicolumn{2}{c}{$(\tau_p^2+3\tau_p) (\tau_c N_c - \tau_p)K |\mathcal{M}_k| $}   \\ \hline
\multicolumn{2}{c|}{Distributed joint channel and CPE LMMSE~\eqref{eq:LMMSE_estimator_dist}} &   \multicolumn{2}{c}{$(\tau_p^2+3\tau_p) \tau_c K |\mathcal{M}_k|$} \\ \hline 
\multirow{3}{*}{\begin{tabular}[c]{@{}c@{}} Centralized joint channel and CPE \\ estimator in Algorithm~\ref{alg:LMMSE_central} \end{tabular}} &  \multicolumn{1}{c|}{Distributed channel DL~\eqref{eq:DL_estimator}}  & \multicolumn{1}{c|}{$((2\tau_pM_1 + M_1M_2 + 2M_2K)/4) |\mathcal{M}_k| $} & \textemdash    \\ \cline{2-4}
 & \multicolumn{1}{c|}{Distributed channel LMMSE~\eqref{eq:LMMSE_h_knownJ}}   &\multicolumn{1}{c|}{$(\tau_p^2+3\tau_p)K |\mathcal{M}_k| N_{\text{Iter}}$} &  \textemdash \\ \cline{2-4}
 & \multicolumn{1}{c|}{Centralized CPE LMMSE~\eqref{eq:lemma_LMMSE_J_1}} & \textemdash & \multicolumn{1}{c}{$((L^2 \tau_p^2 + L \tau_p) + L \tau_p K \tau_c |\mathcal{M}_k|) N_{\text{Iter}} $}     \\ \cline{1-4}
\end{tabularx} 
}
\end{table*}

\begin{table*}[t] 
\centering
\caption{\rev{Number of complex scalars that a generic AP $l$  needs to exchange with the CPU over the fronthaul in each coherence block for different channel and PN estimators when using the MMSE combining~\eqref{eq:v_k_MMSE}. $\mathcal{D}_l$ is defined in~\eqref{eq:D_l_define}.} }
\label{tab:Overhead_estimators}
\rev{
\begin{tabularx}{1\linewidth}{p{0.21\linewidth}|C{0.2\linewidth}|C{0.23\linewidth}|C{0.23\linewidth}}
\hline
\multicolumn{2}{c|}{Estimator} &  Pilot signals or channel/PN estimates & Data signals  \\ \hline
\multicolumn{2}{c|}{PN-unaware distributed channel MMSE~\cite[Eq. (3)]{bjornson2020scalable}}& $\tau_p$ or $|\mathcal{D}_l|$ & $\tau_c N_c -\tau_p$ \\ \hline
\multicolumn{2}{c|}{Distributed joint channel and PN LMMSE~\cite[Eq.(12)]{papazafeiropoulos2021scalable}}& \multicolumn{1}{c|}{$\tau_p$ or $|\mathcal{D}_l| (\tau_c N_c - \tau_p)$} &  $\tau_c N_c -\tau_p$ \\ \hline
\multicolumn{2}{c|}{Distributed joint channel and CPE LMMSE~\eqref{eq:LMMSE_estimator_dist}} &   \multicolumn{1}{c|}{$\tau_p$ or $|\mathcal{D}_l|\tau_c$} & $\tau_c N_c -\tau_p$ \\ \hline 
\multirow{3}{*}{\begin{tabular}[c]{@{}c@{}} Centralized joint channel and CPE \\ estimator in Algorithm~\ref{alg:LMMSE_central} \end{tabular}} &  \multicolumn{1}{c|}{Distributed channel DL~\eqref{eq:DL_estimator}}  & \multicolumn{1}{c|}{$\tau_p$ or $|\mathcal{D}_l|$} & $\tau_c N_c -\tau_p$    \\ \cline{2-4}
 & \multicolumn{1}{c|}{Distributed channel LMMSE~\eqref{eq:LMMSE_h_knownJ}}   &\multicolumn{1}{c|}{$\tau_p$ or $(|\mathcal{D}_l| +|\mathcal{D}_l|\tau_p)N_{\text{Iter}}$} &  $\tau_c N_c -\tau_p$ \\ \cline{2-4}
 & \multicolumn{1}{c|}{Centralized CPE LMMSE~\eqref{eq:lemma_LMMSE_J_1}} & \multicolumn{1}{c|}{$\tau_p$} & $\tau_c N_c -\tau_p$    \\ \cline{1-4}
\end{tabularx} 
}
\end{table*}

\section{PN-aware Centralized Channel Estimation}\label{sec:Central_Estimator}
We derive a centralized PN-aware channel and CPE estimator that alternates between a distributed channel estimator and a centralized CPE estimator, effectively exploiting phase correlation in shared common LO scenarios.

In cell-free \ac{mmimo} networks, PN can be correlated across APs, due to originating from the same UEs or a shared LO of APs. The correlated PN introduces correlated interference, significantly degrading the performance of centralized combining if the PN correlation is not exploited, as shown in~\cite{bjornson2015massive,papazafeiropoulos2021scalable}.
To address the correlated PN interference issue, we propose a novel centralized estimator that exploits PN correlation between APs. Our approach employs an iterative process that alternates between: 
\begin{enumerate}
\item distributed LMMSE estimation for channel optimization;
\item centralized LMMSE estimation for CPE optimization. 
\end{enumerate}
To enhance convergence speed, we initialize the process using the DL-based channel estimator~\eqref{eq:DL_estimator}. This estimator provides superior initial estimates under PN conditions compared to traditional LMMSE methods, thereby reducing the number of iterations required for 
 the centralized channel and CPE estimation.

\subsection{Centralized LMMSE CPE Estimation given Channels}
Based on the collected pilots  ${\boldsymbol{y}}^{\text{p}}$ at the CPU and the channel estimates $\hat{\boldsymbol{h}}_{n}$ from an arbitrary channel estimator,  we derive a centralized estimator of the CPEs $\{J_{k,l,0}^{(\tau)}\}$ for any UE $k$, AP $l$, subcarrier $n \in \mathcal{R}_r$, and OFDM symbol $\tau \in \{1,\cdots,\tau_c\}$. The conventional \ac{mmse} estimation is difficult to compute since the CPEs are not Gaussian distributed. Moreover, leveraging the fact that the CPE is amplitude constrained ($|J_{k,l,0}^{(\tau)}|< 1$), we propose a \textit{constrained} LMMSE estimator that bounds each estimate~\cite{michaeli2008constrained}. 
\begin{lemma}
        Given arbitrary channel estimates $\hat{\boldsymbol{h}}_{n}$ of~\eqref{eq:channel_all}, the centralized constrained LMMSE estimate of $\{{J}_{k,l,0}^{(\tau)}\}$ in~\eqref{eq:J_CPE_k_tau_p} based on ${\boldsymbol{y}}^{\text{p}}$ in~\eqref{eq:observ_all} is given by 
    \begin{align}
        \hat{{J}}_{k,l,0}^{(\tau)} =  g\Big(\bar{J}_{k,l,0}^{(\tau)} + \boldsymbol{b}_{k,l}^{\rev{\mathsf{T}}(\tau)} \breve{\boldsymbol{\Psi}}^{-1} ({\boldsymbol{y}}^{\text{p}} - \bar{\boldsymbol{y}}^{\text{p}}) \Big) \label{eq:lemma_LMMSE_J_1}
    \end{align}
    where
    \begin{align}
       \bar{\boldsymbol{y}}^{\text{p}} &= \left[ \bar{\boldsymbol{y}}^{\text{p} \mathsf{T}}_1
    ,\cdots, \bar{\boldsymbol{y}}^{\text{p} \mathsf{T}}_L \right]^{\mathsf{T}}, \\
\bar{\boldsymbol{y}}^{\text{p} \mathsf{}}_l &= \sum_{k^{\prime}=1}^{K} \sqrt{p_k^{\prime}} \boldsymbol{s}_{t_k^{\prime}} e^{-\frac{\sigma_{\varphi_k}^2+\sigma_{\phi_l}^2}{2}} \hat{h}_{k^{\prime},l,n}, \\
\big[\boldsymbol{b}_{k,l}^{(\tau)}\big]_{\tau^{\prime} l^{\prime}} &= \sum_{k^{\prime}=0}^{K} 
     \sqrt{p_{k^{\prime}}} \hat{h}_{k^{\prime},l^{\prime},n_{\tau^{\prime}}}^{*}  
    {{s}}^{*(\tau^{\prime})}_{t_{k^{\prime}},n_1} \breve{B}_{k,k^{\prime},l,l^{\prime},0,0}^{(\tau -\tau^{\prime})}. \label{eq:LMMSE_CPE_central_EJy}
    \end{align}
    Here, $\bar{J}_{k,l,0}^{(\tau)}$ and $\breve{\boldsymbol{\Psi}}$ are defined in~\eqref{eq:E_J} and~\eqref{eq:Psi_l_1_l_2}, respectively. The constraint function $g(\cdot)$ with amplitude thresholds $\kappa_{\text{min}}$ and $\kappa_{\text{max}}$ is defined as
    \begin{align}
        g(x) = \frac{x}{|x|} \times \max\left(\kappa_{\text{min}}, \min\left(|x|, \kappa_{\text{max}}\right)\right) \,.\label{eq:constraint_func_def}
    \end{align}
\end{lemma}
\begin{proof}
    See Appendix~\ref{Appendix:LMMSE_CPE_Central}.
\end{proof}
Without constraint \( g(\cdot) \), the CPE estimation $\hat{{J}}_{k,l,0}^{(\tau)}$ is of mean $\bar{{J}}_{k,l,0}^{(\tau)}$ and variance $ \breve{\epsilon}_{k,l}^{(\tau)} \triangleq \boldsymbol{b}_{k,l}^{(\tau)}  \breve{\boldsymbol{\Psi}}_{l}^{-1} \boldsymbol{b}_{k,l}^{\mathsf{H}(\tau)}$, and the channel estimation error  $\tilde{{J}}_{k,l,0}^{(\tau)}\triangleq {{J}}_{k,l,0}^{(\tau)} - \hat{{J}}_{k,l,0}^{(\tau)}$ is of zero mean and variance $ \breve{c}_{k,l,n}^{(\tau)} \triangleq B_{0,0}^{(\tau-\tau)} - \breve{\epsilon}_{k,l}^{(\tau)}$.

\begin{remark}
Without the constraint \( g(\cdot) \), the LMMSE estimate of the CPEs is Gaussian distributed with zero mean, which differs from the actual distribution of the CPEs. The additional prior information beyond the mean and variance of CPEs brought by amplitude limitation leads to improved accuracy via imposing an additional constraint related to the amplitude.
\end{remark}
\begin{remark}
The value of the upper and lower thresholds ($\kappa_{\text{min}}$ and $\kappa_{\text{max}}$) of the constraint function \( g(\cdot) \) depends on the PN variance and the OFDM symbol index. Higher PN variance or later OFDM symbols increase the absolute phase of the PN, resulting in many CPEs with smaller amplitudes.
\end{remark}

\subsection{Distributed LMMSE Channel Estimation Given CPEs}
\begin{lemma}
    Given CPE estimates $\hat{\boldsymbol{J}}_{k,l}^{\text{p}}$, the distributed LMMSE estimate of  ${h}_{k,l,n}^{}$ based on ${\boldsymbol{y}}_l^{\text{p}}$ is
    \begin{align}
    \hat{{h}}_{k,l,n}^{} 
    &=\sqrt{p_k}\beta_{k,l}({\boldsymbol{s}}_{t_k} \odot \hat{\boldsymbol{J}}_{k,l}^{\text{p}})^{\mathsf{H}} \tilde{\boldsymbol{\Psi}}_{l}^{-1} {\boldsymbol{y}}_{l}^{\text{p}}, \label{eq:LMMSE_h_knownJ}
\end{align}
where
\begin{align}\label{eq:LMMSE_inverse_matrix_knownJ}
\tilde{\boldsymbol{\Psi}}_{l} &= \sum\nolimits_{k=1}^{K}p_{k}\beta_{k,l} \tilde{\boldsymbol{\Phi}}_{k} + {\boldsymbol{Z}}^{\text{ICI}}_{l} + \sigma^2 \mathbf{I}_{\tau_p} 
\\
[\tilde{\boldsymbol{\Phi}}_{k}]_{\tau_1, \tau_2} &= {s}^{(\tau_1)}_{t_k,n_1} {s}^{*(\tau_2)}_{t_k,n_2} \hat{J}^{(\tau_1)}_{k,l,0} \hat{J}^{*(\tau_2)}_{k,l,0}, \label{eq:Phi_h}
\end{align}
and ${\boldsymbol{Z}}^{\text{ICI}}_{l}$  is defined in~\eqref{eq:Zeta_ICI_full_expression}.
\end{lemma}
\begin{proof}
See Appendix \ref{Appendix:LMMSE_h}.
\end{proof}
With the \ac{lmmse} estimation, the channel estimation $\hat{{h}}_{k,l,n}^{}$ is of zero mean and variance $ \acute{\epsilon}_{k,l,n}^{} \triangleq p_k \beta_{k,l}^2 {\boldsymbol{s}}_{t_k,n}^{\mathsf{H}} \hat{\boldsymbol{J}}_{k,l}^{p} \tilde{\boldsymbol{\Psi}}_{l}^{-1} \hat{\boldsymbol{J}}_{k,l}^{\mathsf{H} p} {\boldsymbol{s}}_{t_k,n}^{\mathsf{}}$, and the channel estimation error  $\tilde{{h}}_{k,l,0}^{(\tau)}\triangleq {{h}}_{k,l,n}^{} - \hat{{h}}_{k,l,n}^{}$ is of zero mean and variance $ \acute{c}_{k,l,n}^{} \triangleq \beta_{k,l} - \acute{\epsilon}_{k,l}^{(\tau)}$. Compared to the distributed joint LMMSE estimator~\eqref{eq:LMMSE_estimator_dist}, the estimator~\eqref{eq:LMMSE_h_knownJ} estimates the channel using the CPE estimates $\hat{\boldsymbol{J}}_{k,l}^{\text{p}}$. This introduces the new term $\tilde{\boldsymbol{\Psi}}_{l}$ compared to~\eqref{eq:LMMSE_inverse_matrix}.
The proposed centralized estimator that alternates between CPE estimation and channel estimation is in Algorithm~\ref{alg:LMMSE_central}. \rev{Algorithm~\ref{alg:LMMSE_central} converges to a stationary point of~\eqref{eq:mainproblem_joint} (See proof in Appendix~\ref{Appendix:Alg1_convg}).}

\begin{algorithm}[t]
\caption{{Proposed Centralized Alternating Channel and CPE Estimation}}\label{alg:LMMSE_central}
\KwIn{ $\boldsymbol{y}^{\text{p}}$ in \eqref{eq:observ_all}. }
\KwOut{$\{\hat{\boldsymbol{J}}_{k,l}^{\text{c}}\}$, $\{\hat{\boldsymbol{h}}^{\text{}}_{k,l,n}\}$ for all $k$ and $l$.}
\rev{\textbf{Initialization}}: Set $i=0$. Initial the channel estimate $\{\hat{\boldsymbol{h}}^{\text{}\rev{0}}_{k,l,n}\}$ by distributed LMMSE~\eqref{eq:LMMSE_h_knownJ} or DL estimator~\eqref{eq:DL_estimator}; Initial PN to zero for all $k$, $l$, and $n$. \\
\For{i=1:$N_{\text{Iter}}$}{
    Update $\{\hat{\boldsymbol{J}}_{k,l}^{\text{c},\rev{i}}\}$: Centralized LMMSE CPE estimation by~\eqref{eq:lemma_LMMSE_J_1} given $\{\hat{{h}}_{k,l,n}^{\rev{i-1}}\}$. \\
    Update $\{\hat{\boldsymbol{h}}^{\text{}\rev{i}}_{k,l,n}\}$: distributed LMMSE channel estimation by~\eqref{eq:LMMSE_h_knownJ} given $\{\hat{\boldsymbol{J}}_{k,l}^{\text{p},\rev{i}}\}$ \rev{from step 3.} 
}
\Return $\{\hat{\boldsymbol{J}}_{k,l}^{\text{c}} \hat{h}_{k,l,n}\}$\;
\end{algorithm}

\section{\rev{Uplink Data Transmission}}
\rev{In this section, we present a novel achievable \ac{se} expression for uplink transmission under \ac{pn}. We also analyze the computational complexity and fronthaul overhead associated with various channel and PN estimators when employing uplink MMSE combining. The uplink data transmission model follows the scalable cell-free mMIMO framework described in~\cite{bjornson2020scalable,wu2023phasenoise}, with additional details provided in Appendix~\ref{Appendix:uplink_data_transmission}. 
}

\begin{figure*}[!t]
\begin{equation}
\text{SINR}_{k,n}^{(\tau)} = \frac{p_k \left|\mathbb{E}\left\{\boldsymbol{v}_{k,n}^{\mathsf{H},(\tau)} \boldsymbol{D}_k \boldsymbol{h}_{k,n}^{(\tau)}\right\}\right|^2}
{\sum\nolimits_{i=1 }^{K} p_i  \mathbb{E}\Big\{\left|\boldsymbol{v}_{k,n}^{\mathsf{H},(\tau)} \boldsymbol{D}_k \boldsymbol{h}_{i,n}^{(\tau)}\right|^2\Big\} - p_k \left|\mathbb{E}\Big\{\boldsymbol{v}_{k,n}^{\mathsf{H}, (\tau)} \boldsymbol{D}_k \boldsymbol{h}_{k,n}^{(\tau)}\right\}\Big|^2 + \rho^{\text{ICI},(\tau)}_{k,n} + \sigma^2\mathbb{E}\Big\{\left| \boldsymbol{D}_k \boldsymbol{v}_{k,n}^{\mathsf{H},(\tau)} \right|^2\Big\}
}\label{eq:SINR_lo_UatF}
\end{equation}
\hrule
\end{figure*}

\subsection{\rev{Uplink Spectral Efficiency}}
The ergodic capacity for this setup with \ac{pn} is unknown. We use the \ac{uatf} bound in cellular \ac{mmimo}~\cite[Th. 4.4]{bjornson2017massive} and~\cite{demir2021foundations,nayebi2016performance,bashar2019uplink} for cell-free \ac{mmimo} to derive a novel achievable uplink \ac{se} under \ac{pn}. 
\begin{proposition}
An uplink achievable SE of UE $k$ over subcarrier $n \in \mathcal{N}_d$ is
\begin{align}
    \text{SE}_{k,n}^{\text{}} =\frac{\tau_cN_c-\tau_p}{\tau_cN_c} \times\frac{1}{\tau_c}\sum_{\tau=1}^{\tau_c} \log_2(1+\text{SINR}_{k,n}^{ (\tau)}),\label{eq:SE_UatF_PN_aware}
\end{align}
where $\text{SINR}_{k,n}^{ (\tau)}$ is the eﬀective \ac{sinr} of UE $k$ over subcarrier $n$, given in~\eqref{eq:SINR_lo_UatF} with the ICI term $\rho_{k,n}^{\text{ICI},(\tau)}$ being defined in \eqref{eq:rho_ICI}.
\end{proposition}

\begin{proof}
    See Appendix~\ref{sec:Appendix_SINR_lo_UatF}.
\end{proof}
Note that the effective \ac{sinr} in~\eqref{eq:SE_UatF_PN_aware} varies for each OFDM symbol $\tau$ due to the \ac{cpe} introduced by the \ac{pn}. 
The SE expression in~\eqref{eq:SE_UatF_PN_aware} can be computed numerically for any combiner  $\boldsymbol{v}_{k,n}^{(\tau)}$ using Monte Carlo methods. 
\subsection{\rev{Uplink Combining}}
In the context of cell-free \ac{mmimo}, common combiners are \ac{mr}, \ac{lpmmse}, \ac{mmse}, and \ac{pmmse} combinings, given as~\cite[Eqs. (19),(29), (23), (20)]{bjornson2020scalable}. Among the four combiners, the MMSE combiner consistently demonstrates the highest performance, whereas the MR combiner performs significantly worse in comparison. As noted in~\cite{bjornson2019making}, the MR combiner is not recommended for use in cell-free \ac{mmimo} networks even without PN.
Here we only give the MMSE combiner
\begin{align}
    \boldsymbol{v}_{k,n}^{\text{MMSE},(\tau)} 
    &= p_k \Big(\sum\limits_{i=1}^{K} p_i \hat{\boldsymbol{H}}_{i,n}^{D,(\tau)}  + \boldsymbol{Z}_{i,n}^{(\tau)} \Big)^{\dagger} \boldsymbol{D}_k \hat{\boldsymbol{h}}_{i,n}^{(\tau)},
    \label{eq:v_k_MMSE}
\end{align}
where $\hat{\boldsymbol{H}}_{i,n}^{D,(\tau)} =\boldsymbol{D}_k \hat{\boldsymbol{h}}_{i,n}^{(\tau)} \hat{\boldsymbol{h}}_{i,n}^{\mathsf{H},(\tau)} \boldsymbol{D}_k$ 
and $\boldsymbol{Z}_{i,n}^{(\tau)}=\boldsymbol{D}_k\big(\sum_{i=1}^K p_i \boldsymbol{C}_{i,n}^{(\tau)}+\sigma_{\text{}}^2 \boldsymbol{I}_{L}\big) \boldsymbol{D}_k$. Here, the channel estimates, $\hat{\boldsymbol{h}}_{k,n}^{(\tau)}\triangleq [\hat{h}_{k,1,n}^{(\tau)},\cdots, \hat{h}_{k,L,n}^{(\tau)}]^{\mathsf{T}}$ and the estimation error variance $\{c_{k,l,n}^{(\tau)}\}$ are given by the estimators~\eqref{eq:LMMSE_estimator_dist},~\eqref{eq:lemma_LMMSE_J_1}, and~\eqref{eq:LMMSE_h_knownJ} depending on the underlying scenario. We refer to a combining scheme as \acl{pna} or \acl{pnu}, based on the usage of \acl{pna} or \acl{pnu} channel estimators. 
\subsection{\rev{Computational Complexity and Fronthaul Overhead}}
\rev{
The \ac{mmse} combining in~\eqref{eq:v_k_MMSE} is computed at the \ac{cpu}, either by using the channel and PN estimates provided by each AP or by using pilots forwarded from each AP (then the estimation is operated at the CPU). The computational complexity per coherence block for different channel and PN estimators, measured in complex multiplications and divisions, is presented in Table~\ref{tab:complexity_estimators} for a generic UE $k$. Additionally, the required fronthaul overhead per AP $l$, measured in complex scalars, is summarized in Table~\ref{tab:Overhead_estimators}. A detailed complexity analysis is available in Appendix~\ref{Appendix:complexity}. The results in Tables~\ref{tab:complexity_estimators} and~\ref{tab:Overhead_estimators} highlight that while the computational complexities of the channel and PN estimators are not scalable as $K\rightarrow \infty$, their fronthaul overhead remains scalable in this limit. Furthermore, as shown in~\cite[Table II]{bjornson2020scalable}, employing scalable combining schemes such as \ac{pmmse} or \ac{lpmmse} ensures that all estimators listed in Table~\ref{tab:complexity_estimators} achieve scalability as $K\rightarrow \infty$. }

\section{Numerical Results} 

Numerical results are presented to evaluate and compare the performance of the proposed solutions. Both APs and UEs are equipped with single antennas. The simulation configurations are shown in Table~\ref{tab:scen_details}, along with references. These setups allow for a comprehensive comparison with state-of-the-art methods and an investigation of how different LO configurations affect performance. In particular, two scenarios are considered:
\begin{itemize}[noitemsep, topsep=0pt,leftmargin=*]
    \item \textit{Scenario 1}: Each AP operates with an independent LO, resulting in uncorrelated PN across APs. In this case, we evaluate the proposed distributed joint channel and CPE estimator from equation \eqref{eq:LMMSE_estimator_dist}.
    \item \textit{Scenario 2}: APs share a common LO, resulting in correlated PN among APs. We evaluate the proposed centralized joint channel and CPE estimator, as described in Algorithm \ref{alg:LMMSE_central}.
\end{itemize}
Note that the PN processes at the UEs are independent of the PN processes at the APs in both scenarios. Scenario 1 resembles a conventional cell-free \ac{mmimo} network~\cite{bjornson2020scalable,papazafeiropoulos2021scalable} while Scenario 2 resembles a radio stripe network~\cite{interdonato2019ubiquitous,shaik2021mmse}. 

\subsection{Metrics and 
Baselines}\label{sec:benchmarks}
The uplink achievable \ac{se}~\eqref{eq:SE_UatF_PN_aware} is evaluated with the MMSE~\eqref{eq:v_k_MMSE} combiner using five channel estimators. The evaluation includes two pilot pattern, PP1 and PP2, depicted in Fig.~\ref{fig:coherent_block}. The first four estimators are under the OFDM generative signal model with PN~\eqref{eq:observ_all} and the last one is under the single-carrier generative signal model. We detail them in the following:
\begin{enumerate}[noitemsep, topsep=0pt,leftmargin=*]
    \item \textit{Unaware (Distr.)}: The PN-unaware distributed MMSE estimator from~\cite{bjornson2020scalable,shaik2021mmse,bjornson2019making}, derived from a no-PN signal model but applied to the PN-impaired generative model~\eqref{eq:observ_all}.
    \item \textit{Mismatched (Distr.)}: \rev{The mismatched PN-aware LMMSE estimator from~\cite{bjornson2015massive, papazafeiropoulos2021scalable}}, \rev{originally} derived from a single-carrier PN-impaired signal model but applied to the OFDM model~\eqref{eq:observ_all}, \rev{used as a Bayesian estimator benchmark.} 
    \item \textit{Proposed (Distr.)}: The proposed distributed PN-aware LMMSE estimator~\eqref{eq:LMMSE_estimator_dist}, derived from and applied to the OFDM signal model~\eqref{eq:observ_all}.
    \item \textit{Proposed (Centr.) w/ DL or w/ LMMSE}: The proposed centralized alternating channel and CPE estimator from Algorithm~\ref{alg:LMMSE_central}, derived from and applied to the OFDM signal model~\eqref{eq:observ_all}, initialized with the DL channel estimator~\eqref{eq:DL_estimator} or the LMMSE channel estimator~\eqref{eq:LMMSE_h_knownJ}, respectively. 
    \item \textit{Single-carrier (Distr.)}: \rev{The single-carrier PN-aware LMMSE estimator from~\cite{bjornson2015massive, papazafeiropoulos2021scalable}} applied to the single-carrier PN-impared signal model.
\end{enumerate}
Table~\ref{tab:DL_train_par} presents the training details for the \ac{dl} channel estimator, whose dimensions including the number of layers and neurons match those of the channel and PN estimators in~\cite{mohammadian2021deep}.
\begin{table}[t]
\centering
\caption{Parameter Setups for Deep Learning Channel Estimator}
\label{tab:dl_params}
\begin{tabularx}{0.75\linewidth}{c|c}
\hline
\textbf{Parameter} & \textbf{Value} \\ \hline
Number of neurons ($M_1 = M_2$) & 100 \\ \hline
Training realizations & 3000 \\ \hline
Batch size & 128 \\ \hline
Optimizer & Adam~\cite{kingma2014adam} \\ \hline
Initial learning rate & 0.01 \\ \hline
Learning rate drop factor & 0.2 \\ \hline
Drop frequency & Every 50 epochs \\ \hline
\end{tabularx} \label{tab:DL_train_par}
\end{table}

\newcolumntype{Y}{>{\centering\arraybackslash}X}
\begin{table}[t]
\centering
\caption{Parameter setups for two simulation scenarios.}
\label{tab:simulation_params}
\begin{tabularx}{1\linewidth}{>{\hsize=1.55\hsize}Y|>{\hsize=0.55\hsize}Y|>{\hsize=0.9\hsize}Y}
\hline
\textbf{Parameter} & \textbf{Scenario 1} & \textbf{Scenario 2} \\ \hline
AP LOs configuration & Separate & Shared \\ \hline
\multirow{3}{*}{AP distribution} &  \multirow{4}{*}{ \shortstack{Uniform in \\ 1x1 km$^2$~\cite{bjornson2020scalable}}} & Equidistant on 0.5x0.5 km$^2$ radio stripe~\cite{interdonato2019ubiquitous,shaik2021mmse}  \\ \cline{1-1} \cline{3-3}
UE distribution & & Unif.: 0.4x0.4 km$^2$ \\ \hline
Number of APs: $L$ & $200$ & Varied: $5$ to $100$ \\ \hline
Number of UEs: $K$ & $5$ & $2$ \\ \hline
LO coefficients: $\gamma_{\phi_l} = \gamma_{\varphi_k} \forall k,l$ &  \multicolumn{2}{|c}{Varied: $\geq 10^{-17}$~\cite{bjornson2015massive,papazafeiropoulos2021scalable, krishnan2015linear}} \\ \hline
Channel model & \multicolumn{2}{|c}{Spatially correlated~\cite{bjornson2020scalable}}
\\ \hline
UE power and carrier freq.: $p_k$, $f_c$ &  \multicolumn{2}{|c}{$100$ mW~\cite{bjornson2020scalable}, $2$ GHz~\cite{bjornson2015massive,papazafeiropoulos2021scalable}}\\ \hline
OFDM Setup: $N,\Delta_f$ &   \multicolumn{2}{|c}{$667,15 \text{kHz}$} \\ \hline
BW and symbol time: $W$, $T_s$ &  \multicolumn{2}{|c}{$10$ MHz, 0.1 $\mu$s~\cite{bjornson2015massive,papazafeiropoulos2021scalable}} \\ \hline
Coherence block:
$N_c, \tau_c, W_c, T_c$ &  \multicolumn{2}{|c}{12, 20, 180 kHz, 1.33 ms} \\ \hline
Number of pilots: $\tau_p$ &  \multicolumn{2}{|c}{20~\cite{papazafeiropoulos2021scalable,shaik2021mmse}} \\ \hline
\rev{
Number of Monte Carlo simulations} &  \multicolumn{2}{|c}{\rev{100}}\\ \hline
\end{tabularx} \label{tab:scen_details}
\end{table}


\subsection{Scenario 1: Separate LOs across APs}
In Scenario 1, each AP has a separate LO, and the PN processes at different APs are independent of each other. We evaluate the performance of the proposed distributed joint channel and CPE estimator in \eqref{eq:LMMSE_estimator_dist}. 
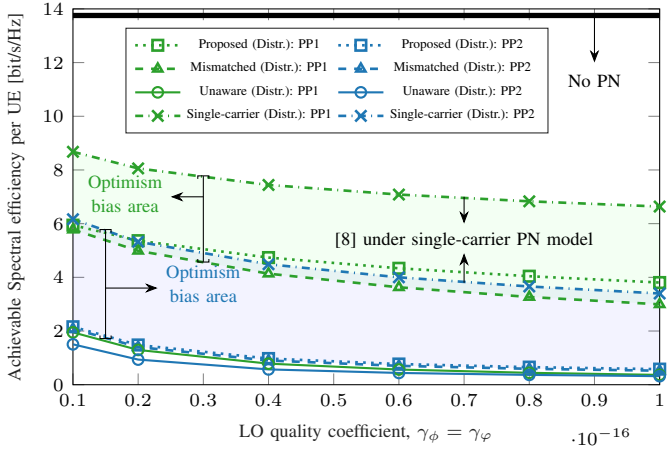
\begin{figure}[t]
    \centering
    \vspace*{-0 \baselineskip}
	\input{Results/SE_vs_LO}
    \caption{Scenario 1 with separate LOs and pilot patterns PP1 and PP2. Uplink achieveable SE per UE versus the AP and UE LO quality coefficients $\gamma_{\phi_l}=\gamma_{\varphi_k}$ for the MMSE combining with the proposed distributed \acl{pna} \ac{lmmse}~\eqref{eq:LMMSE_estimator_dist}, single-carrier \acl{pna} \ac{lmmse}~\cite{papazafeiropoulos2021scalable} under OFDM (mismatched) and single-carrier generative models,  and \acl{pnu} \ac{mmse}~\cite{bjornson2020scalable} estimators. }
    \label{fig:SE_vs_LO}
     \vspace*{-0.5cm}
\end{figure}
\begin{figure}[t]
    \centering
    \vspace*{-0 \baselineskip}
    \input{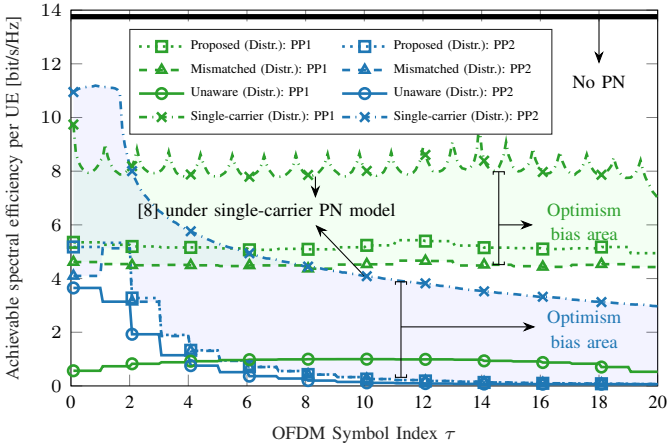}
    \caption{Scenario 1 with separate LOs and pilot patterns PP1 and PP2. Uplink achievable SE per UE versus the OFDM symbol index $\tau$ for the MMSE combining with the proposed distributed \acl{pna} \ac{lmmse}~\eqref{eq:LMMSE_estimator_dist}, single-carrier \acl{pna} \ac{lmmse}~\cite{papazafeiropoulos2021scalable} under OFDM (mismatched) and single-carrier generative signal models, and \acl{pnu} \ac{mmse}~\cite{bjornson2020scalable} estimators. }
    \label{fig:SE_vs_OFDM_Symbol}
\end{figure}

\begin{figure}[t]
    \centering
    \vspace*{-0 \baselineskip}
    \input{Results/SE_vs_UE_AP}
    \caption{Scenario 1 with separate LOs for perfect UE LOs ($\gamma_{\varphi_k}=0$) versus perfect AP LOs ($\gamma_{\phi_l}=0$). Uplink achievable SE per UE versus AP or UE LO quality coefficient for the MMSE combining with the proposed distributed \acl{pna} \ac{lmmse}~\eqref{eq:LMMSE_estimator_dist}, mismatched \acl{pna} \ac{lmmse}~\cite{papazafeiropoulos2021scalable}, and \acl{pnu} \ac{mmse}~\cite{bjornson2020scalable} estimators. Pilot pattern PP1 is used.  Dash-dotted lines with triangle markers show results for mismatched PNA LMMSE estimator under the single-carrier generative signal model.}
    \label{fig:SE_vs_LO_UE_vs_AP}
\end{figure}
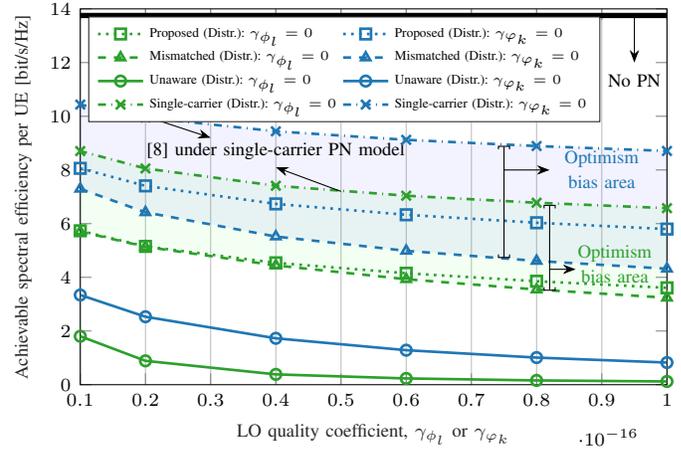%

\subsubsection{Impact of Pilot Patterns and Signal Model Mismatches}
Fig.~\ref{fig:SE_vs_LO} shows the achievable SEs~\eqref{eq:SE_UatF_PN_aware} as a function of the LO quality coefficient (given in~\eqref{eq:LO_coe_define}) for both APs and UEs under two pilot pattern schemes, PP1 and PP2, in the first coherence block. Based on the OFDM generative signal model with PN~\eqref{eq:observ_all}, the achievable SEs are computed using the MMSE combining~\eqref{eq:v_k_MMSE} based on the first three distributed channel estimators outlined in Section~\ref{sec:benchmarks}. Additionally, the single-carrier channel estimator under the single-carrier generative signal model is also shown. The no-\ac{pn} case of the PN-unaware MMSE estimator from~\cite{bjornson2020scalable,shaik2021mmse,bjornson2019making}, under a no-\ac{pn} generative signal model is also shown.

Fig.~\ref{fig:SE_vs_LO} shows the negative impact of LO quality on the achievable SE in cell-free mMIMO, in comparison to the study on cellular \ac{mmimo}~\cite{krishnan2015linear}. The results can also be interpreted as the impact of higher carrier frequencies, with $f_c=4$ GHz and $6$ GHz matching the results of $\gamma_{\phi}=0.2 \times 10^{-16}$ and $0.9 \times 10^{-16}$, respectively. Pilot pattern PP2 yields notably lower achievable SEs compared to PP1, as it lacks measurements for the effective channels of OFDM symbols from $3$ to $20$. The proposed PN-aware estimator, derived from the OFDM signal model, achieves the highest achievable SEs, with its advantage over the mismatched scheme growing as LO quality decreases in PP1. In contrast, the single-carrier estimator~\cite{papazafeiropoulos2021scalable}, now mismatched under the OFDM generative signal model, performs less effectively. \rev{Additionally, when applied to the single-carrier generative signal model, this estimator produces significantly higher achievable SE predictions than under the OFDM generative signal model, showing an \textit{optimism bias} of at least $40\%$ and $300\%$ for PP1 and PP2, respectively, as highlighted by the green and blue areas. The gaps suggest that the single-carrier method overestimates its performance in OFDM systems. In summary, our
findings reveal two key issues with the use of single-carrier methods in OFDM systems: estimators lead to degraded
performance under the correct OFDM PN model, and achievable SE predictions result in overly
optimistic performance predictions under the simplified single-carrier PN model.} 

Extending the results from~Fig.~\ref{fig:SE_vs_LO} with low LO quality ($\gamma_{\phi_l}=\gamma_{\varphi_k}=4\times10^{-17}$) as in~\cite{wu2023phasenoise}, Fig.~\ref{fig:SE_vs_OFDM_Symbol} shows the uplink achievable SE per UE across OFDM symbols in the first coherence block. Under the OFDM generative signal model~\eqref{eq:observ_all}, PP1 with evenly distributed pilots maintains more consistent achievable SEs and resists channel aging. Our distributed estimator outperforms mismatched and unaware estimators across the OFDM symbols. PP2, with pilots in the first two OFDM symbols, achieves the highest initial achievable SE but quickly degrades due to a lack of measurements and channel aging. In PP2, our estimator's advantage is limited to the first two OFDM symbols. Similar to Fig.~\ref{fig:SE_vs_LO}, the mismatched estimator, under a single-carrier model, yields overly optimistic SE predictions, showing an SE optimism bias of at least $60\%$ and $200\%$ for PP1 and PP2, respectively, highlighted by the green and blue areas.

\subsubsection{Different PN Impact from APs and UEs}
Fig.~\ref{fig:SE_vs_LO_UE_vs_AP} shows the impact on the achievable SE of the quality of the LOs in UEs and APs for the pilot pattern PP1. We compare perfect UE LOs ($\gamma_{\varphi_k}=0$) and perfect AP LOs ($\gamma_{\phi_l}=0$). The latter is practical since APs generally have higher-quality LOs than low-cost UEs. The results reveal two key points: (i) The UE LO quality impacts achievable SE far more than the AP LO, leading to lower SEs due to UE PN causing correlated interference between APs which eventually increases the total interference power when using the centralized combining; (ii) The proposed estimator significantly outperforms the mismatched estimator under imperfect AP LOs, as opposed to imperfect UE LOs. This suggests that the proposed algorithm is more effective when LO imperfections occur on the AP side. Similar to Fig.~\ref{fig:SE_vs_LO} and~\ref{fig:SE_vs_OFDM_Symbol}, the single-carrier estimator~\cite{papazafeiropoulos2021scalable} leads to overly optimistic SE predictions under the single-carrier generative signal model, showing an SE optimism bias of at least $40\%$ and $50\%$ for $\gamma_{\phi_l}=0$ and $\gamma_{\varphi_k}=0$, respectively.

\subsubsection{Impact of UE Density}
Fig.~\ref{fig:SE_vs_K} extends the analysis from Fig.~\ref{fig:SE_vs_OFDM_Symbol} by evaluating achievable SEs for pilot pattern PP1 with varying UE numbers, using centralized MMSE combiner. Solid black lines show the case without PN. For schemes with the MMSE combiner, under the OFDM generative signal model~\eqref{eq:observ_all}, SEs decrease dramatically with more UEs due to inter-carrier and inter-user interference. PN-aware MMSE combiner with the proposed distributed estimator shows SE gains for fewer UEs ($K<20$), but this decreases with more UEs due to pilot contamination. Notably, the mismatched estimator's results under the single-carrier signal model show optimistic SE predictions compared to its real performance under the OFDM model for at least $200\%$ due to underestimating the impact of PN. This over-optimism can lead to inaccurate analyses of critical network parameters such as the required number of APs, affordable number of UEs, and required LO quality and bandwidth. 

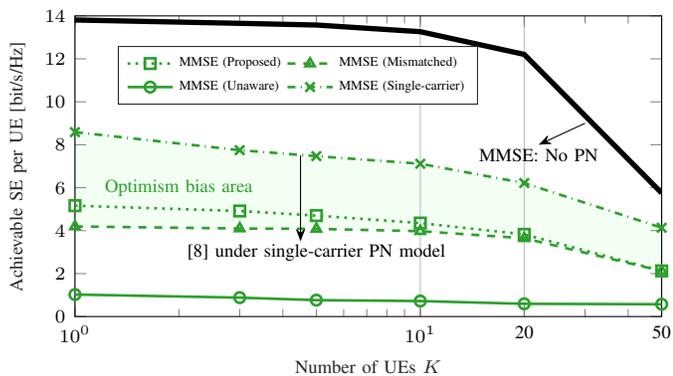
\begin{figure}[t]
    \centering
        \vspace*{-0 \baselineskip}
	\input{Results/SE_vs_K}
    \caption{Scenario 1 with separate LOs between $L=200$ APs with LO quality $\gamma_{\phi_l}=\gamma_{\varphi_k}= 4\times 10^{-17}$, and pilot pattern PP1. Uplink achievable SE per UE versus the number of \acp{ue} $K$ for the MMSE combiner~\eqref{eq:v_k_MMSE} with the proposed distributed \acl{pna} \ac{lmmse}~\eqref{eq:LMMSE_estimator_dist}, single-carrier \acl{pna} \ac{lmmse}~\cite{papazafeiropoulos2021scalable} under the OFDM (mismatched) and single-carrier generative signal models, and \acl{pnu} \ac{mmse}~\cite{bjornson2020scalable} estimators.}
    \label{fig:SE_vs_K}
    \vspace*{-0. cm}
\end{figure}
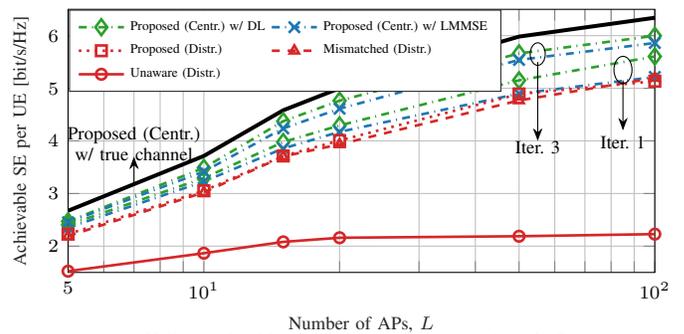
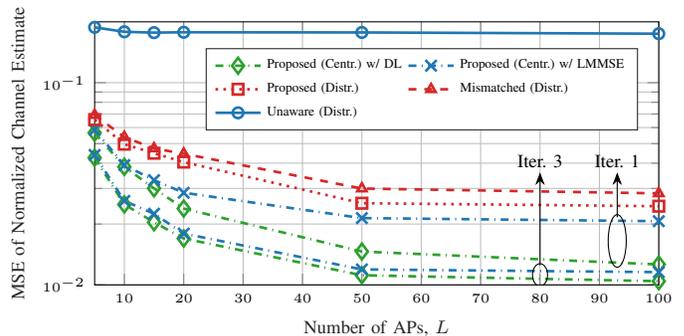
\begin{figure}[t]
    \centering
    \vspace*{-0 \baselineskip}
    \begin{subfigure}{0.5\textwidth}
    \centering
        \input{Results/SE_vs_L_CLO}
            \vspace*{-0.3cm}
        \caption{Uplink achievable SE per UE versus the number of APs}
        \label{fig:SE_vs_AP_shared}
    \end{subfigure}
    \hfill
                \vspace*{-0.5cm}
    \begin{subfigure}{0.5\textwidth}
    \centering
       \input{Results/Channel_vs_L_CLO}
                   \vspace*{-0.3cm}
       \caption{MSE of the channel estimate versus the number of APs}
       \label{fig:Channel_MSE_CLO}
    \end{subfigure}
    \caption{Scenario 2 with a shared common LO between APs. LOs quality coefficients $\gamma_{\phi_l}=\gamma_{\varphi_k}= 10^{-17}$. The comparison includes the MMSE combining with the proposed distributed \acl{pna} \ac{lmmse}~\eqref{eq:LMMSE_estimator_dist}, mismatched \acl{pna} \ac{lmmse}~\cite{papazafeiropoulos2021scalable},  \acl{pnu} \ac{mmse}~\cite{bjornson2020scalable}, and centralized estimators initialized by DL and LMMSE channel estimators.}
    \label{fig:SE&ChannelMSE_vs_AP_shared}
\end{figure}

\subsubsection{Summary of Findings}
The proposed distributed joint channel and CPE estimator for separate LOs among APs is particularly beneficial when using poor-quality LOs (or high carrier frequencies) and a small number of UEs. Performance gains are influenced by the pilot pattern, with distributed pilots across different OFDM symbols significantly outperforming centralized pilots within the same symbol. Single-carrier estimators underestimate the impact of PN in the OFDM model, resulting in overly optimistic SE predictions. 

\subsection{Scenario 2: Common LO shared by APs}
In Scenario 2 where all APs share the same common LO, we evaluate the proposed centralized estimator (detailed in Algorithm~\ref{alg:LMMSE_central}), which can exploit the PN correlation.

\subsubsection{Impact of AP Density}
To evaluate the performance improvement of the proposed centralized estimator exploiting the PN correlation among APs compared to other distributed estimators,
Fig.~\ref{fig:SE_vs_AP_shared} shows the uplink SE per UE as a function of the number of APs in Scenario 2 with a shared LO between APs, using MMSE combining for various estimators: \ac{pnu} MMSE~\cite{bjornson2020scalable} (Unaware), mismatched PN-aware LMMSE~\cite{papazafeiropoulos2021scalable} (Mismatched), the proposed distributed joint LMMSE~\eqref{eq:LMMSE_estimator_dist} (Proposed Distri.), and the proposed centralized alternating estimator initialized with the DL channel estimator~\eqref{eq:DL_estimator} (Proposed Centr. w/ DL) and the LMMSE channel estimator~\eqref{eq:LMMSE_h_knownJ} (Proposed Centr. w/ LMMSE). Additionally, the achievable SE results for the centralized estimator initialized with perfect channels are also depicted by a black solid line.

The results demonstrate that the proposed centralized estimator significantly improves SE, especially as the number of APs increases. This enhancement indicates that centralized CPE estimation effectively exploits the \ac{pn} correlation between APs, with greater gains observed as the number of APs grows. Notably, the centralized estimator initialized with the DL channel estimator outperforms the one initialized with the LMMSE channel estimator~\eqref{eq:LMMSE_h_knownJ}, achieving better performance with the same number of iterations and thereby reducing the overhead. This superior performance is attributed to the DL channel estimator’s ability to provide better channel estimates under the impact of PN, as evidenced by the MSE of normalized channel estimation in Fig.~\ref{fig:Channel_MSE_CLO}. \rev{The centralized estimator with DL-based initialization nearly matches the performance of perfect channel knowledge. This implies that channel estimation isn’t the primary limit on SE, and more complex \ac{nn} structures are unlikely to offer significant gains.}

\subsubsection{Impact of UE Transmit Power}
Based on Scenario 2 with 50 APs sharing a common LO, consistent with the setups in Fig.~\ref{fig:SE_vs_AP_shared} and Fig.~\ref{fig:Channel_MSE_CLO}, we evaluate the impact of the PN under different UE transmit power. To isolate the impact of UE power, the UEs are placed in the center of the radio stripe area, equidistant from all APs. We evaluate the uplink achievable SE per UE by varying the transmit power of the 2 UEs \( p_1=p_2 \). The results, shown in Fig.~\ref{fig:SE_vs_power_CLO}, demonstrate that the centralized estimator significantly improves SE by leveraging PN correlation, particularly at high transmit power levels. The SE gain increases with higher transmit power, compared to distributed estimators. Additionally, the centralized estimator initialized with the DL channel estimator achieves better SEs than when initialized with the LMMSE channel estimator, given the same number of iterations. 
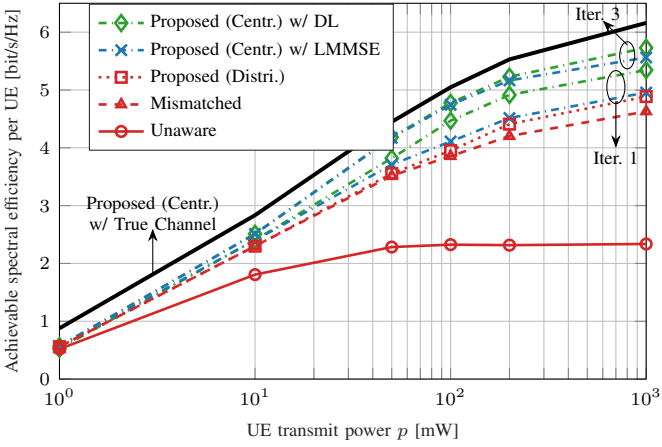
\begin{figure}[t]
    \centering
	\input{Results/SE_vs_power_CLO}
    \caption{Uplink SE versus UE transmit power in scenario 2, with \( L=50 \) APs sharing a common LO ($\gamma_{\phi_l}=10^{-17}$),  \( K=2 \) equidistant UEs (($\gamma_{\varphi_k}=10^{-17}$)), and pilot pattern PP1.  }
    \label{fig:SE_vs_power_CLO}
\end{figure}

\rev{\subsubsection{Impact of CPE Constraint Function}}
\rev{
We now evaluate how the amplitude constraint function~\eqref{eq:constraint_func_def} affects the centralized LMMSE CPE estimator~\eqref{eq:lemma_LMMSE_J_1}. Based on scenario 2 with $20$ APs and $2$ UEs, the evaluation considers CPE estimates at different OFDM symbols under situations of different LO quality.    Fig.~\ref{fig:CPE_threshold_1} and Fig.~\ref{fig:CPE_threshold_2} present two representative snapshots of the CPE estimates, both with and without the amplitude constraint function~\eqref{eq:constraint_func_def}, together with the true CPE values. In Fig.~\ref{fig:CPE_threshold_1}, both the APs and UEs use high-quality LOs ($\gamma_{\varphi_k}=\gamma_{\phi_l}=10^{-17}$) and the CPE results are for the first OFDM symbol. In contrast,  Fig.~\ref{fig:CPE_threshold_2}  focuses on the $20$-th OFDM symbol, with UEs employing low-quality LOs ($\gamma_{\varphi_k}=4\times 10^{-17}$).}

\rev{Fig.~\ref{fig:CPE_threshold_1} and Fig.~\ref{fig:CPE_threshold_2} show that the true CPE points lie near the unit circle with amplitudes below $1$.  When the LMMSE estimator is used without an amplitude constraint, its estimates often exceed $1$ in amplitude. By adding the amplitude constraint function, these estimates are kept closer to the unit circle, and the \ac{mse} of the CPE estimation drops from $0.041$ to $0.022$ in Fig.~\ref{fig:CPE_threshold_1}, and from $0.050$ to $0.031$  Fig.~\ref{fig:CPE_threshold_2}. Specifically, 
in Fig.~\ref{fig:CPE_threshold_1}, the LOs have better quality, so the CPEs are more tightly clustered near the unit circle. This justifies using a higher minimum threshold,  $\kappa_{\min}=0.98$. In contrast, Fig.~\ref{fig:CPE_threshold_2} involves lower-quality LOs, giving a broader distribution of the CPEs around the circle and prompting a lower minimum threshold, $\kappa_{\min}=0.90$. In practice, one can use this strategy to adjust the minimum threshold of the amplitude constraint to improve CPE estimation performance. }

\rev{
\begin{figure}
\begin{subfigure}{0.48\textwidth}
    \centering
\input{Results/CPE_threshold_tau1}
    \caption{\rev{CPEs at OFDM symbol $\tau=1$, $J_{k,l,0}^{(1)}$ with good quality LO at UEs and APs,  $\gamma_{\varphi_k}=\gamma_{\phi_l}=10^{-17}$. The amplitude constraint function improves the estimation accuracy with reduced CPE MSE from $0.041$ to $0.022$.}}
\label{fig:CPE_threshold_1}
\end{subfigure}
\begin{subfigure}{0.48\textwidth}
    \centering
\input{Results/CPE_threshold_tau20_4e17}
    \caption{\rev{CPEs at OFDM symbol $\tau=20$, $J_{k,l,0}^{(20)}$ with low quality LO at UEs,  $\gamma_{\varphi_k}=4 \times 10^{-17}$, and good quality LO at APs,  $\gamma_{\phi_l}=10^{-17}$. The constraint function improves the estimation accuracy with reduced CPE MSE from $0.050$ to $0.031$.}}
\label{fig:CPE_threshold_2}
\end{subfigure}
\caption{\rev{Results of the centralized CPE estimator~\eqref{eq:lemma_LMMSE_J_1} with and without constraint function $g(\cdot)$ in~\eqref{eq:constraint_func_def} in scenario 2 with $20$ APs and $2$ UEs.}  }
\end{figure}
}
\nopagebreak[4]
\subsubsection{Summary of Findings}
The proposed centralized alternating channel and CPE estimator for a shared LO among APs is most effective with a large number of APs and across a wide range of UE transmission powers. The DL channel estimator is more effective than the LMMSE channel estimator in enhancing the performance of the centralized estimator. 

\section{Conclusion}
This paper investigated the impact of \ac{pn} on uplink cell-free \ac{mmimo} \ac{ofdm} networks, considering both separate and shared \acp{lo} between \acp{ap}. We developed an accurate \ac{pn}-impaired cell-free \ac{mmimo} OFDM signal model and derived a novel uplink achievable \ac{se} expression. To address the challenges of uncorrelated and correlated \ac{pn} between APs, we proposed two novel \ac{pn}-aware channel estimators: i) a distributed \ac{lmmse} joint channel and \ac{cpe} estimator for the separate \acp{lo} setup, and ii) a centralized channel and \ac{cpe} estimator for the shared \acp{lo} setup. The latter estimator alternates between distributed channel estimations and centralized \ac{cpe} estimation to exploit \ac{pn} correlation. Additionally, we introduced a \ac{dl}-based channel estimator to improve the initialization of the centralized estimator, enhancing channel estimates and reducing the required number of iterations. \rev{Alternative neural network architectures for the DL estimator in the centralized design remain an area for future exploration.} Key highlights include: distributing pilots across different OFDM symbols rather than one OFDM symbol significantly improves performance. The proposed distributed estimator under OFDM signal model is suitable for cell-free networks with separate, poor-quality AP LOs, high carrier frequencies, and fewer UEs, while single-carrier channel estimation and uplink combining methods tend to underestimate PN effects, leading to overly optimistic performance predictions. The proposed centralized estimator performs better with a large number of APs and high UE power, with its performance further enhanced by the DL channel estimator compared to the LMMSE estimator. \rev{Future work could extend to multi-antenna configurations for both APs and UEs, where shared LOs may be employed for each AP or UE. In this scenario, the proposed alternating channel and CPE estimation algorithm can be used both locally (at each AP/UE) and centrally (at the CPU) to exploit PN correlation.}


\appendices{}
\section{Proof of Lemma 1}
\label{Appendix:LMMSE_Joint}
The general expression for the \ac{lmmse}  estimator is~\cite{kay1993fundamentals}
\begin{align}   
\hat{{h}}_{k,l,n}^{(\tau)}=\mathbb{E}  \{{{h}}_{k,l,n}^{(\tau)} {\boldsymbol{y}}_{}^{\text{p} \mathsf{H}} \} \big(\mathbb{E}\{ {\boldsymbol{y}}_{}^{\text{p}} {\boldsymbol{y}}_{}^{\text{p}\mathsf{H}}\}\big)^{-1} {\boldsymbol{y}}_{}^{\text{p}}, \label{eq:general_LMMSE_dist_1}
\end{align} 
where we have 
\begin{align}
    \mathbb{E}  \{{{h}}_{k,l,n}^{(\tau)}  {\boldsymbol{y}}_{}^{\text{p}\mathsf{H}} \}= \big[0,\cdots, \mathbb{E}  \{{{h}}_{k,l,n}^{(\tau)}  {\boldsymbol{y}}_{l}^{\text{p}\mathsf{H}} \},\cdots,0\big], \label{eq:E_h_y_zero}
    \end{align}
Here, $\mathbb{E}  \{{{h}}_{k,l,n}^{(\tau)} {\boldsymbol{y}}_{l_1}^{\mathsf{H}} \}=0$ for $l_1\neq l$  because $\mathbb{E}  \{{{h}}_{k,l,n}^{} {{h}}_{k,l_1,n}^{*} \}=0$ as assumed in~\eqref{eq:channel_uncorr}. Thus,~\eqref{eq:general_LMMSE_dist_1} is reduced to
\begin{align}
    \hat{{h}}_{k,l,n}^{(\tau)}=\mathbb{E}  \{{{h}}_{k,l,n}^{(\tau)} {\boldsymbol{y}}_{l}^{\text{p} \mathsf{H}} \} \big(\mathbb{E}\{ {\boldsymbol{y}}_{l}^{\text{p}} {\boldsymbol{y}}_{l}^{\text{p}\mathsf{H}}\}\big)^{-1} {\boldsymbol{y}}_{l}^{\text{p}}, \label{eq:general_LMMSE_dist_reduced}
\end{align}
where $\mathbb{E}  \{{{h}}_{k,l,n}^{(\tau)} {\boldsymbol{y}}_{l}^{\text{p} \mathsf{H}} \}$ is given by
    \begin{align}
    \mathbb{E} & \{{{h}}_{k,l,n}^{(\tau)} {\boldsymbol{y}}_{l}^{\text{p} \mathsf{H}} \} 
         =\sqrt{p_k}\beta_{k,l}  \nonumber \\
    &\times \big[{{s}}^{*(\tau_1)}_{t_k,n_1}\mathbb{E}\{J_{k,l,0}^{(\tau)}J_{k,l,0}^{*(\tau_1)}\},\cdots,{{s}}^{*(\tau_p)}_{t_k,n_{\tau_p}}\mathbb{E}\{J_{k,l,0}^{(\tau)}J_{k,l,0}^{*(\tau_p)}\}   \big]\nonumber\\
    &\hspace{17mm} = \sqrt{p_k}\beta_{k,l}{\boldsymbol{s}}_{t_k}^{\mathsf{H}} \Acute{\boldsymbol{B}}_{k,l}^{(\tau)},
\end{align}
where $\Acute{\boldsymbol{B}}_{k,l}^{(\tau)}=\text{diag}[\mathbb{E}\{J_{k,l,0}^{(\tau)}J_{k,l,0}^{*(\tau_1)}\},\cdots, \mathbb{E}\{J_{k,l,0}^{(\tau)}J_{k,l,0}^{*(\tau_{p})}\}]$ and its $\tau^{\prime}$ element ${B}_{k,l,0,0}^{(\tau-\tau^{\prime})} = \mathbb{E}\{J_{k,l,0}^{(\tau)}J_{k,l,0}^{*(\tau^{\prime})}\}$ is given by~\eqref{eq:B_matrix_compute}. 

Furthermore, we have
\begin{align}
    &\mathbb{E}  \left\{{\boldsymbol{y}}_{l} {\boldsymbol{y}}_{l}^{\mathsf{H}}
    \right\}= \sigma^2 \boldsymbol{I}_{\tau_p}  \nonumber\\
&+ \mathbb{E} \big\{ \sum\nolimits_{k^{\prime}=1}^{K}\sqrt{p_{k^{\prime}}} \big[
 {{s}}^{(\tau_1)}_{t_{k^{\prime}},n_1}J_{k^{\prime},l,0}^{(\tau_1)},\cdots, {{s}}^{(\tau_p)}_{t_{k^{\prime}},n_{\tau_p}}J_{k^{\prime},l,0}^{(\tau_p)}
 \big]^{\mathsf{T}}
 {{h}}_{k^{\prime},l,n}  \nonumber\\
&\,\,\,\,\,\,\,\, \times {{h}}_{k^{\prime},l,n}^{*} \sqrt{p_{k}^{\prime}} \big[
 {{s}}^{(\tau_1)}_{t_{k^{\prime}},n_1}J_{k^{\prime},l,0}^{*(\tau_1)},\cdots, {{s}}^{(\tau_p)}_{t_{k^{\prime}},n_{\tau_p}}J_{k^{\prime},l,0}^{*(\tau_p)}
 \big]
 \big\}  
 \end{align}
 \begin{align}
 &+ \mathbb{E} \Big \{ \underbrace{\Big[\sum\nolimits_{k^{\prime}=1}^{K}\zeta_{k^{\prime},l,n_1}^{(\tau_1)},\cdots, \sum\nolimits_{k^{\prime}=1}^{K}\zeta_{k^{\prime},l,n_{\tau_p}}^{(\tau_p)}\Big]^{\mathsf{T}}}_{\triangleq \boldsymbol{\zeta}_{k^{\prime},l}} \boldsymbol{\zeta}_{k^{\prime},l}^{\mathsf{H}} 
 \Big\} \nonumber
 \nonumber\\
 &= \underbrace{\sum\nolimits_{k^{\prime}=1}^{K}p_{k^{\prime}}\beta_{k^{\prime},l} \boldsymbol{\Phi}_{k^{\prime}} + \boldsymbol{Z}^{\text{ICI}}_{l} + \sigma^2 \boldsymbol{I}_{\tau_p}}_{\triangleq \boldsymbol{\Psi}_{l}}, \label{eq:Psi_Joint_l}
\end{align}
where  the $(\tau_1,\tau_2)$ element of $ \boldsymbol{\Phi}_{k^{\prime},l} \in \mathbb{C}^{\tau_p \times \tau_p}$ is given in~\eqref{eq:Phi_joint}.
The $(\tau_1, \tau_2)$ element of the ICI component $\boldsymbol{Z}^{\text{ICI}}_{l} \in \mathbb{C}^{\tau_p \times \tau_p}$ is 
\begin{align}
    &[\boldsymbol{Z}^{\text{ICI}}_{l}]_{\tau_1, \tau_2}  \nonumber \\
    &= \mathbb{E}\Big\{\sum_{k^{\prime}=1}^K \zeta_{k^{\prime},l,n_1}^{(\tau_1)} \sum_{k^{\prime\prime}=1}^K \zeta_{k^{\prime\prime},l,n_{2}}^{*(\tau_2)}\Big\} =
    \sum\limits_{k^{\prime}=1}^K \mathbb{E}\left\{ \zeta_{k^{\prime},l,n_1}^{(\tau_1)} \zeta_{k^{\prime},l,n_2}^{*(\tau_2)}\right\}\nonumber \\
    &=
    \sum_{k^{\prime}=1}^K p_{k^{\prime}} \beta_{k^{\prime},l} \sum_{\substack{\jmath_1 \neq n_1}}^{N-1} \sum_{\substack{\jmath_2 \neq n_2}}^{N-1}\mathbb{E} \{ {{s}}^{(\tau_1)}_{t_{k^{\prime}},\jmath_1} {{s}}^{*(\tau_2)}_{t_{k^{\prime}},\jmath_2} J^{(\tau_1)}_{k^{\prime},l,n_1-\jmath_1} J^{*(\tau_2)}_{k^{\prime},l,n_2-\jmath_2} \}  \nonumber 
        \end{align}
    \begin{align}
    &=\begin{cases} 
0 & \tau_1 \neq \tau_2 \\
\sum\limits_{k^{\prime}=1}^K p_{k^{\prime}} \beta_{k^{\prime},l} \sum\limits_{\substack{\jmath \neq n}}^{N-1}  B_{k^{\prime},l,n-\jmath, n-\jmath}^{(0)}   & \shortstack{ $\tau_1 = \tau_2$, \\ $n_1=n_2=n$}  \\
\sum\limits_{k^{\prime}=1}^K p_{k^{\prime}} \beta_{k^{\prime},l} \sum\limits_{\substack{\jmath \neq n_1,n_2}}^{N-1}  
B_{k^{\prime},l,n_1-\jmath, n_2-\jmath}^{(0)}  \\ \quad \quad  +s_{t_{k^{\prime}},n_2}^{(\tau_1)} s_{t_{k^{\prime}},n_1}^{*(\tau_1)} B_{k^{\prime},l,n_1-n_2, n_1-n_2}^{(0)}  & \shortstack{$\tau_1 = \tau_2$,\\ $n_2 \neq n_1$.}
\end{cases}
\label{eq:Zeta_ICI_full_expression}
\end{align}
where the correlation term ${B}_{k^{\prime},l,n_1-\jmath_1,n_2-\jmath_2}^{(0)}$ is given by \eqref{eq:B_matrix_compute}.

\rev{
\section{Computational Complexity and Fronthaul Overhead Analysis}\label{Appendix:complexity}}
\rev{
Assuming all statistical matrices are precomputed and remain constant during communication, the computational complexity for the considered channel and PN estimators is quantified in terms of complex multiplications and divisions, as these operations dominate the total complexity~\cite[Sec. B.1.1]{bjornson2017massive}. Additionally, the required fronthaul resources are calculated for each estimator, expressed either by the number of pilot and data symbols or the number of channel and PN estimates. The results are summarized in Table~\ref{tab:complexity_estimators} and~\ref{tab:Overhead_estimators}.
Note that both complexity and overhead calculation assume the use of MMSE combining in~\eqref{eq:v_k_MMSE}. Specifically, computational complexity is calculated per coherence block for a generic UE $k$, and fronthaul overhead is calculated per coherence block for a generic AP $l$. 
\subsection{Distributed Joint Channel and CPE LMMSE Estimator~\eqref{eq:LMMSE_estimator_dist}}\label{Appendix:complexity_joint_distrib}
Using Lemma B.2 in~\cite{bjornson2017massive}, computing $\boldsymbol{\Psi}_l^{-1} \boldsymbol{y}_l^p$ in~\eqref{eq:LMMSE_estimator_dist} requires $\tau_p^2$ complex multiplications and $\tau_p$ complex divisions. In addition, multiplying by  $\boldsymbol{s}_{t_k}^{\mathsf{H}}$ and the diagonal matrix $ \boldsymbol{B}_{k,l}^{(\tau)}$ add another  $2\tau_p$ complex multiplications. In total, the computation requires $\tau_p^2+3\tau_p$ complex multiplications and divisions. This process is repeated for each of the $\tau_c$ OFDM symbols in a coherence block. The MMSE combining~\eqref{eq:v_k_MMSE} requires channel estimates for all $K$ UEs and for the APs serving UE $k$, i.e., $l \in \mathcal{M}_k$, yielding the results in Table~\ref{tab:complexity_estimators}.
}

\rev{
Regarding the fronthaul overhead, each AP $l$ must either forward $\tau_p$ pilot symbols to the CPU, which then performs the channel and PN estimation, or compute local estimates for $|\mathcal{D}_l|$ UEs and $\tau_c$ OFDM symbols in a coherence block and forward these estimates to the CPU. Additionally, for the data signals, each AP must send each uplink data symbol to the CPU and subsequently receive each downlink precoding coefficient from the CPU, thus in total $(\tau_cN_c - \tau_p)$ complex scalars.}

\rev{
\subsection{Distributed Channel DL Estimator~\eqref{eq:DL_estimator}} \label{Appendix:complexity_DL}
Each prediction requires $(2\tau_pM_1 + M_1 M_2 + 2 M_2K)$ real-valued multiplications to generate channel predictions for all $K$ UEs, which is equivalent to dividing by $4$ to approximate the cost in complex multiplications. The prediction is performed for APs serving UE $k$, i.e., $l \in \mathcal{M}_k$, for the MMSE combining~\eqref{eq:v_k_MMSE}.}

\rev{Regarding the fronthaul overhead, each AP $l$ must either forward $\tau_p$ pilot symbols to the CPU, which then performs the DL channel estimation, or compute local estimates for $|\mathcal{D}_l|$ UEs and $1$ OFDM symbol and forward these estimates to the CPU. The overhead for data signals is the same as in Appendix~\ref{Appendix:complexity_joint_distrib}.}

\rev{
\subsection{Distributed Channel LMMSE Estimator~\eqref{eq:LMMSE_h_knownJ}} \label{Appendix:complexity_h_knownJ}
The complexity analysis is similar to that in Appendix~\ref{Appendix:complexity_joint_distrib} for~\eqref{eq:LMMSE_estimator_dist}, except that only one channel estimate is computed per coherence block. The prediction is performed for all $K$ UEs and the APs ($l \in \mathcal{M}_k$) for the MMSE combining~\eqref{eq:v_k_MMSE}.
}

\rev{
Regarding the fronthaul overhead for each iteration in Algorithm~\ref{alg:LMMSE_central}, each AP $l$ must either forward $\tau_p$ pilot symbols to the CPU for channel estimation or locally compute channel estimates for $|\mathcal{D}_l|$ UEs over one OFDM symbol via~\eqref{eq:LMMSE_h_knownJ}. Each local estimation requires $\tau_p$ CPE estimates,  $\hat{\boldsymbol{J}}_{k,l}^{\text{p},i}$, to be sent from the CPU, resulting in a total of $|\mathcal{D}_l| \tau_p$ complex scalars transmitted. Finally, each AP $l$ must send its computed $|\mathcal{D}_l|$ channel estimates back to the CPU for the subsequent centralized CPE estimation step in~\eqref{eq:lemma_LMMSE_J_1}. The data signal overhead is the same as in Appendix~\ref{Appendix:complexity_joint_distrib}.
}
\rev{
\subsection{Centralized CPE LMMSE estimator~\eqref{eq:lemma_LMMSE_J_1}} \label{Appendix:complexity_J_central}
In~\eqref{eq:lemma_LMMSE_J_1}, the computation of $\breve{\boldsymbol{\Psi}}^{-1}(\boldsymbol{y}^p-\bar{\boldsymbol{y}}^p)$ is performed once for all $K$ UEs and $L$ APs, requiring $(L^2\tau_p^2+L\tau_p)$ complex multiplication and divisions. When using the MMSE combining~\eqref{eq:v_k_MMSE}, the multiplication by $\boldsymbol{b}_{k,l}^{\mathsf{T}(\tau)}$ must be executed for each of the $\tau_c$ OFDM symbols, for all $K$ UEs and and for the APs serving UE $k$, $l \in \mathcal{M}_k$.
}

\rev{
The centralized CPE estimator~\eqref{eq:lemma_LMMSE_J_1} requires each AP $l$ to forward $\tau_p$ pilot symbols to the CPU. The overhead associated with the required distributed channel estimates in~\eqref{eq:lemma_LMMSE_J_1} is already accounted for, provided in Appendix~\ref{Appendix:complexity_h_knownJ}. The data signal overhead is the same as in Appendix~\ref{Appendix:complexity_joint_distrib}.
}

\section{Proof of Lemma 2}\label{Appendix:LMMSE_CPE_Central}
Leveraging the LMMSE theorem\cite{kay1993fundamentals} and given the channel estimates $\hat{\boldsymbol{h}}_n$, the LMMSE estimate of ${{J}}_{k,l,0}^{(\tau)}$ based on pilots at the CPU $\boldsymbol{y}^{\text{p}}$ is given by 
\begin{align}
    \hat{{J}}_{k,l,0}^{(\tau)} 
    &= \mathbb{E}  \{{{J}}_{k,l,0}^{(\tau)}\} + \mathbb{E}  \{{{J}}_{k,l,0}^{(\tau)} {\boldsymbol{y}}_{}^{\text{p}\mathsf{H}} \} \big(\mathbb{E}\{ {\boldsymbol{y}}^{\text{p}} {\boldsymbol{y}}_{}^{\text{p} \mathsf{H}}\}\big)^{-1} ({\boldsymbol{y}}^{\text{p}} - \mathbb{E}  \{{\boldsymbol{y}}_{}^{\text{p}}\} ), \label{eq:LMMSE_general_CPE}
\end{align}
where $\mathbb{E}  \{{{J}}_{k,l,0}^{(\tau)}\}=\bar{{J}}_{k,l,0}^{(\tau)}$ defined in~\eqref{eq:E_J}, and $\mathbb{E}  \{{\boldsymbol{y}}_{}^{\text{p}}\} \triangleq 
\bar{\boldsymbol{y}}_{}^{\text{p}}$ with its $l\tau_p$-th to $(l+1)\tau_p$-th elements $\mathbb{E} \{\boldsymbol{y}_l^{\text{p}}\} = 
\sum_{k=1}^{K} \sqrt{p_k} \boldsymbol{s}_{t_k}^{\text{p}} \bar{\boldsymbol{J}}_{k,l}^{\text{p}} \hat{h}_{k,l,n}$. Also, we have
\begin{align}
    \mathbb{E}  \{{{J}}_{k,l,0}^{(\tau)} {\boldsymbol{y}}_{}^{\text{p}\mathsf{H}}\} \triangleq \boldsymbol{b}_{k,l}^{\rev{\mathsf{T}}(\tau)} = \big[\mathbb{E}  \{{{J}}_{k,l,0}^{(\tau)} {\boldsymbol{y}}_{1}^{\mathsf{H}} \}, \cdots, \mathbb{E}  \{{{J}}_{k,l,0}^{(\tau)} {\boldsymbol{y}}_{L}^{\mathsf{H}} \}\big], \label{eq:E_Jy_CPE}
\end{align}
with its $\tau^{\prime}l^{\prime}$-th element $\mathbb{E} \{{{J}}_{k,l,0}^{(\tau)} [{\boldsymbol{y}}_{l^{\prime}}^{\mathsf{H}}]_{\tau^{\prime}} \}$ defined in~\eqref{eq:LMMSE_CPE_central_EJy}. 

Furthermore, in~\eqref{eq:LMMSE_general_CPE} we have $\mathbb{E}\{ {\boldsymbol{y}}_{}^{\text{p}} {\boldsymbol{y}}_{}^{{\text{p}}\mathsf{H}}\} \triangleq \breve{\boldsymbol{\Psi}} \in \mathbb{C}^{\tau_pL \times \tau_p L} $, and its submatrix $\breve{\boldsymbol{\Psi}}_{l_1,l_2}$ consisting of the elements from rows $l_1\tau_p$ to $(l_1+1)\tau_p$ and columns $l_1\tau_p$ to $(l_1+1)\tau_p$ can be calculated similar to~\eqref{eq:Psi_Joint_l} as
\begin{align}
    &\breve{\boldsymbol{\Psi}}_{l_1,l_2} = \mathbb{E}  \left\{{\boldsymbol{y}}_{l_1} {\boldsymbol{y}}_{l_2}^{\mathsf{H}}
    \right\} = \sigma^2 \boldsymbol{I}_{\tau_p} +\mathbb{E} \big\{ \sum_{k=1}^{K}\sqrt{p_k} \big[
 {{s}}^{(\tau_1)}_{t_k,n_1} {J}_{k,l_1,0}^{(\tau_1)}, \nonumber\\
 &\cdots, {{s}}^{(\tau_p)}_{t_k,n_{\tau_p}} {J}_{k,l_1,0}^{(\tau_p)}
 \big]^{\mathsf{T}}
 {\hat{h}}_{k,l_1,n}  \times {\hat{h}}_{k,l_2,n}^{*} \sqrt{p_k} \big[
 {{s}}^{(\tau_1)}_{t_k,n_1} {J}_{k,l_2,0}^{*(\tau_1)}, \nonumber \\ 
 &\cdots, {{s}}^{(\tau_p)}_{t_k,n_{\tau_p}} {J}_{k,l_2,0}^{*(\tau_p)}
 \big]
 \big\} + \mathbb{E} \big \{ \boldsymbol{\zeta}_{k,l_1} \boldsymbol{\zeta}_{k,l_2}^{\mathsf{H}} 
 \big\} \nonumber\\
 &={\sum\nolimits_{k=1}^{K}p_{k}\hat{h}_{k,l_1,n}  \hat{h}_{k,l_2,n}^{*}\breve{\boldsymbol{\Phi}}_{k,l_1,l_2} + \breve{\boldsymbol{Z}}^{\text{ICI}}_{l_1,l_2} + \sigma^2 \boldsymbol{I}_{\tau_p}}.\label{eq:Psi_l_1_l_2}
\end{align}
Here, the $(\tau_1,\tau_2)$ element of $ \breve{\boldsymbol{\Phi}}_{k,l_1,L2} \in \mathbb{C}^{\tau_p \times \tau_p}$ is 
\begin{align}
    [\breve{\boldsymbol{\Phi}}_{k,l_1,l_2}]_{\tau_1, \tau_2} = {s}^{(\tau_1)}_{t_k,n_1} {s}^{*(\tau_2)}_{t_k,n_2} \breve{B}_{k,k,l_1,l_2,0,0}^{(\tau_1-\tau_2)},
\end{align} 
where $\breve{B}_{k,k,l_1,l_2,0,0}^{(\tau_1-\tau_2)}$ is calculated following~\eqref{eq:B_matrix_compute_diff}. In~\eqref{eq:Psi_l_1_l_2}, the $(\tau_1, \tau_2)$ element of the ICI component $\breve{\boldsymbol{Z}}^{\text{ICI}}_{l_1,l_2} \in \mathbb{C}^{\tau_p \times \tau_p}$ is given in~\eqref{eq:ICI_CPE_central}. Note that the channel estimates $\{\hat{h}_{k^{},l^{},n}^{}\}$  are used in~\eqref{eq:E_Jy_CPE},~\eqref{eq:Psi_l_1_l_2}, and~\eqref{eq:ICI_CPE_central}.

\begin{figure*}[!t]
\scriptsize
\begin{align}
    [\boldsymbol{Z}^{\text{ICI}}_{l_1,l_2}]_{\tau_1, \tau_2}  &= \mathbb{E}\Big\{\sum_{k=1}^K \zeta_{k,l_1,n_1}^{(\tau_1)} \sum_{k^{\prime}=1}^K \zeta_{k^{\prime},l_2,n_{2}}^{*(\tau_2)}\Big\} =
    \sum\limits_{k=1}^K \mathbb{E}\left\{ \zeta_{k,l_1,n_1}^{(\tau_1)} \zeta_{k^{},l_2,n_2}^{*(\tau_2)}\right\}
    =\sum_{k=1}^K p_k  \sum_{\substack{\jmath_1 \neq n_1}}^{N-1} \sum_{\substack{\jmath_2 \neq n_2}}^{N-1}\mathbb{E} \{ \hat{h}_{k,l_1,\jmath_1} \hat{h}_{k,l_2,\jmath_2}^{*} {{s}}^{(\tau_1)}_{t_k,\jmath_1} {{s}}^{*(\tau_2)}_{t_k,\jmath_2} J^{(\tau_1)}_{k,l_1,n_1-\jmath_1} J^{*(\tau_2)}_{k,l_2,n_2-\jmath_2} \} 
    \nonumber\\
    &= \begin{cases} 
0 & \tau_1 \neq \tau_2 \\
\sum_{k=1}^K p_k  \Big( \sum\limits_{\substack{\jmath_1 \neq n \\ \jmath_1 \in \mathcal{N}_r }}^{N-1}  \hat{h}_{k,l,\jmath_1} \hat{h}_{k,l,\jmath_1}^{*} \mathbb{E}\{|s|^2\} \mathbb{E}\{J_{k,l,n-\jmath_1}^{(\tau)} J_{k,l,n-\jmath_1}^{* (\tau)} \} + \sum\limits_{\substack{\jmath_2 \neq n \\ \jmath_2 \notin \mathcal{N}_r}}^{N-1} \beta_{k,l} \mathbb{E}\{|s|^2\} \mathbb{E}\{J_{k,l,n-\jmath_2}^{(\tau)} J_{k,l,n-\jmath_2}^{* (\tau)} \}    \Big)& \shortstack{ $\tau_1 = \tau_2 = \tau$, \\ $l_1=l_2=l$,\\ $n_1=n_2=n$} \\
\sum_{k=1}^K p_k  \Big( \sum\limits_{\substack{\jmath_1 \neq n_1 \\ \jmath_1 \in \mathcal{N}_p }}^{N-1}   \hat{h}_{k,l_1,\jmath_1} \hat{h}_{k,l_2,\jmath_1}^{*} \mathbb{E}\{|s|^2\} \mathbb{E}\{J_{k,l_1,n-\jmath_1}^{(\tau)} J_{k,l_2,n-\jmath_1}^{* (\tau)} \}   \Big)  & \shortstack{$l_1 \neq l_2$, \\ $\tau_1 = \tau_2=\tau$,\\ $n_1=n_2=n$} 
\end{cases}
\label{eq:ICI_CPE_central}
\end{align}
\hrule
\end{figure*}

\section{Proof of Lemma 3}\label{Appendix:LMMSE_h}
Given the CPE estimates $\{\hat{J}_{k,l,0}^{(\tau)}\}$ and received pilots $\boldsymbol{y}_l^{\text{p}}$, the distributed LMMSE estimates of the channel $\{{h}_{k,l,n}^{}\}$ can be formulated as $\hat{{h}}_{k,l,n}^{}=\mathbb{E}  \{{{h}}_{k,l,n}^{} {\boldsymbol{y}}_{l}^{\text{p}\mathsf{H}} \} \big(\mathbb{E}\{ {\boldsymbol{y}}_{l}^{\text{p}} {\boldsymbol{y}}_{l}^{\text{p}\mathsf{H}}\}\big)^{-1} {\boldsymbol{y}}_{l}^{\text{p}}$, where we have
     $\mathbb{E}  \{{{h}}_{k,l,n}^{} {\boldsymbol{y}}_{l}^{\text{p}\mathsf{H}}\} =  \sqrt{p_k}\beta_{k,l} (\boldsymbol{s}_{t_k}^{\text{p}})^{\mathsf{H}} \boldsymbol{J}_{k,l}^{\text{p}}$ 
and the term $\mathbb{E}\{ {\boldsymbol{y}}_{l}^{\text{p}} {\boldsymbol{y}}_{l}^{\text{p}\mathsf{H}}\} \triangleq \tilde{\boldsymbol{\Psi}}_l$ can be calculated similarly to~\eqref{eq:Psi_Joint_l} by replacing the unknown CPEs with CPE estimates $\hat{\boldsymbol{J}}_{k,l}^{\text{p}}$ for all pilot OFDM symbols. The details are omitted for simplicity. In the end, comparing  $\tilde{\boldsymbol{\Psi}}_l$ to ${\boldsymbol{\Psi}}_l$ in~\eqref{eq:Psi_Joint_l}, we have a different term $\tilde{\boldsymbol{\Phi}}_k$ given in~\eqref{eq:Phi_h}.


\rev{
\section{Convergence Analysis of Algorithm 1} \label{Appendix:Alg1_convg}
}
\rev{
The LMMSE optimization problem of effective channels $\{{h}_{k,l,n}^{(\tau)}\} = \{{{J}}_{k,l,0}^{(\tau)} {h}_{k,l,n}\}$  is formulated as 
\begin{align}
\begin{aligned}
    \hat{h}_{k,l,n}^{(\tau)} &= \arg \underset{{h}_{k,l,n}^{(\tau)}}{\min} \mathbb{E}\left\{\left\|{h}_{k,l,n}^{(\tau)} - \hat{h}_{k,l,n}^{(\tau)} \right\|^2 \right\},  \\
    \text{ s.t.  }  {h}_{k,l,n}^{(\tau)} &= \boldsymbol{A}_{k,l,n}^{(\tau)}\boldsymbol{y}^{p} + b_{k,l,n}^{(\tau)},
    \label{eq:mainproblem_joint},
    \end{aligned}
\end{align}
where the constraint in~\eqref{eq:mainproblem_joint} imposes a linear estimator with parameters $\boldsymbol{A}_{k,l,n}^{(\tau)}$ and $b_{k,l,n}^{(\tau)}$.
}

\rev{
In Algorithm 1, we decompose the problem~\eqref{eq:mainproblem_joint} into two subproblems, each optimizing over CPEs, $\{{{J}}_{k,l,0}^{(\tau)}\}$, and channels $\{{h}_{k,l,n}\}$ while keeping the other variable fixed in an alternating fashion.} 
\rev{
The LMMSE subproblem for CPE estimation at the $i$-th iteration, $\hat{{J}}_{k,l,0}^{(\tau),i}$, can be cast as
\begin{align}
\begin{aligned}
    \hat{{J}}_{k,l,0}^{(\tau),i} &= \arg \underset{{{J}}_{k,l,0}^{(\tau),i}}{\min} \mathbb{E}\left\{\left\|{{J}}_{k,l,0}^{(\tau),i} - \hat{{J}}_{k,l,0}^{(\tau),i} \right\|^2 \right\}, \label{eq:subproblem_CPE}
    \\
    \text{ s.t.  }  {J}_{k,l,n}^{(\tau),i} &= \breve{\boldsymbol{A}}_{k,l,n}^{(\tau),i}\boldsymbol{y}^{p} + \breve{b}_{k,l,n}^{(\tau),i},
    \end{aligned}
\end{align}
which is solved in closed-form by the constraint LMMSE estimator \eqref{eq:lemma_LMMSE_J_1}.}

\rev{
The LMMSE subproblem for channel estimation at the $i$-th iteration, $\hat{{h}}_{k,l,n}^{i}$, can be cast as
\begin{align}
\begin{aligned}
    \hat{{h}}_{k,l,n}^{i} &= \arg \underset{{{h}}_{k,l,n}^{i}}{\min} \mathbb{E}\left\{\left\|{{h}}_{k,l,n}^{i} - \hat{{h}}_{k,l,n}^{i} \right\|^2 \right\}, \label{eq:subproblem_h}
    \\
    \text{ s.t.  }  {h}_{k,l,n}^{i} &= \tilde{\boldsymbol{A}}_{k,l,n}^{i}\boldsymbol{y}^{p} + \tilde{b}_{k,l,n}^{i},
    \end{aligned}
\end{align}
which is solved in closed-form by the LMMSE estimator \eqref{eq:LMMSE_h_knownJ}.}

\rev{
\begin{proposition}
    Algorithm~\ref{alg:LMMSE_central}, which solves~\eqref{eq:subproblem_CPE} and~\eqref{eq:subproblem_h} in an alternating fashion, converges to a stationary point of the optimization problem in~\eqref{eq:mainproblem_joint}. 
\end{proposition}
\begin{proof}
    The closed-form solutions given by~\eqref{eq:lemma_LMMSE_J_1} provide a unique solution of the CPE estimation subproblem in~\eqref{eq:subproblem_CPE}.  The channel estimation subproblem in~\eqref{eq:LMMSE_h_knownJ} is an optimal solution of the channel estimation subproblem in~\eqref{eq:subproblem_h}. It then follows from~\cite[Proposition 1]{aubry2018new} that Algorithm~\ref{alg:LMMSE_central} converges to a stationary point of~\eqref{eq:mainproblem_joint}. Step 1 in Algorithm~\ref{alg:LMMSE_central} uses either the LMMSE~\eqref{eq:LMMSE_h_knownJ} or DL estimator~\eqref{eq:DL_estimator} to compute a reliable initial guess for the iterative process, rather than relying on random initialization.
\end{proof}
}
\rev{
\section{Scalable Uplink Data Transmission} \label{Appendix:uplink_data_transmission}
Following the scalable uplink data transmission protocol in~\cite{bjornson2020scalable},  each UE $k$ is served by a subset of APs, defined as~\cite{bjornson2020scalable,buzzi2017cell}
\begin{align}
    \mathcal{M}_k = \{l: d_{k,l}=1, l\in \{1,\cdots,L\}\}
    \label{eq:M_k_define},
\end{align}
where $d_{k,l} \in \{0,1\}$ indicates if UE $k$ and AP $l$ communicate, based on the \ac{dcc} framework~\cite{bjornson2011optimality}. Similarly, we define the set of UEs that are served by AP $l$ by
\begin{align}
    \mathcal{D}_{l} = \{k : d_{k,l}=1, k \in \{1,\cdots, K\}\}
    \label{eq:D_l_define}.
\end{align}
The uplink data transmission signal model follows~\cite[Eq. (15)]{wu2023phasenoise}, and we omit it for simplicity. At the receiver, the CPU selects $v_{k,l,n}^{(\tau)}$ for UE $k$ and AP $l$, then estimates $s_{k,n}^{(\tau)}$ by
\begin{align}
    \hat{s}_{k,n}^{(\tau)}
    &= \underbrace{\boldsymbol{v}_{k,n}^{\mathsf{H},(\tau)} \boldsymbol{D}_k \boldsymbol{h}_{k,n}^{(\tau)} s_{k,n}^{(\tau)}}_{\text{Desired signal}} + \underbrace{\sum\nolimits_{\substack{i\neq k}}^{K}\boldsymbol{v}_{k,n}^{\mathsf{H},(\tau)} \boldsymbol{D}_k \boldsymbol{h}_{i,n}^{(\tau)} s_{i,n}^{(\tau)}}_{\text{Inter-user interference}}  \nonumber\\
    &+ \underbrace{\sum\nolimits_{i=1}^{K} \boldsymbol{v}_{k,n}^{\mathsf{H},(\tau)} \boldsymbol{D}_k \boldsymbol{\zeta}_{i,n}^{(\tau)}}_{\text{ICI}} 
    + \boldsymbol{v}_{k,n}^{\mathsf{H},(\tau)} \boldsymbol{D}_k \boldsymbol{w}_{n}^{(\tau)},
    \label{eq:s_k_n_level4}
\end{align}
where $\boldsymbol{v}_{k,n}^{(\tau)}=[v_{k,1,n}^{(\tau)} ,\cdots, v_{k,L,n}^{(\tau)}]^{\mathsf{T}}$ denotes the collective combining vector, and $\boldsymbol{D}_k = \text{diag}([d_{k,1},\cdots d_{k,L}]^{\mathsf{T}})$ is a diagonal matrix representing the connections between UE $k$ and each of $L$ APs. Here $\boldsymbol{h}_{k,n}^{(\tau)}=
[h_{k,1,n}^{(\tau)},\cdots, h_{k,L,n}^{(\tau)}]^{\mathsf{T}}$, $\boldsymbol{\zeta}_{k,n}^{(\tau)}=[\zeta_{k,1,n}^{(\tau)},\cdots, \zeta_{k,L,n}^{(\tau)}]^{\mathsf{T}}$, and $\boldsymbol{w}_{n}^{(\tau)}  \sim \mathcal{N}_{\mathbb{C}}(\mathbf{0},\sigma^2\boldsymbol{I}_{L})$.
}

\section{Proof of Proposition 1}\label{sec:Appendix_SINR_lo_UatF}

Since the effective channels vary with each \ac{ofdm} symbol $\tau$, we follow the reference~\cite[Thm. 4.4]{demir2021foundations} using the \ac{uatf} bound to obtain an achievable SE for data subcarrier $n \in \mathcal{N}_d$ at  \ac{ofdm} symbol $\tau \in \{1 ,\cdots, \tau_c\}$. These achievable SEs are averaged over all $\tau_c$ \ac{ofdm} symbols to obtain~\eqref{eq:SE_UatF_PN_aware}. 

Specifically, by adding and subtracting the average effective channel $\mathbb{E}\left\{\boldsymbol{v}_{k,n}^{\mathsf{H},(\tau)} \boldsymbol{D}_k \boldsymbol{h}_{k,n}^{(\tau)}\right\}$, \eqref{eq:s_k_n_level4} is rewritten as
\begin{align}
    \hat{s}_{k,n}^{(\tau)} &= \mathbb{E}\left\{\boldsymbol{v}_{k,n}^{\mathsf{H},(\tau)} \boldsymbol{D}_k \boldsymbol{h}_{k,n}^{(\tau)}\right\}  s_{k,n}^{(\tau)}   + \nu_{k,n}^{(\tau)},
\end{align}
where the interference term is
\begin{align}
    &\nu_{k,n}^{(\tau)} = \Big(\boldsymbol{v}_{k,n}^{\mathsf{H},(\tau)} \boldsymbol{D}_k \boldsymbol{h}_{k,n}^{(\tau)}  - \mathbb{E}\left\{\boldsymbol{v}_{k,n}^{\mathsf{H},(\tau)} \boldsymbol{D}_k \boldsymbol{h}_{k,n}^{(\tau)}\right\}\Big)s_{k,n}^{(\tau)} +  \\& {\sum_{\substack{i=1 \\i\neq k}}^{K}\boldsymbol{v}_{k,n}^{\mathsf{H},(\tau)} \boldsymbol{D}_k \boldsymbol{h}_{i,n}^{(\tau)} s_{i,n}^{(\tau)}} 
    + {\sum_{i=1}^{K} \boldsymbol{v}_{k,n}^{\mathsf{H},(\tau)} \boldsymbol{D}_k \boldsymbol{\zeta}_{i,n}^{(\tau)}}
    + \boldsymbol{v}_{k,n}^{\mathsf{H},(\tau)} \boldsymbol{D}_k \boldsymbol{w}_{n}^{(\tau)}.\nonumber 
\end{align}
This can be viewed as a deterministic channel with a gain $\mathbb{E}\{\boldsymbol{v}_{k,n}^{\mathsf{H},(\tau)} \boldsymbol{D}_k \boldsymbol{h}_{k,n}^{(\tau)}\}$ and additive interference plus noise term $\nu_{k,n}^{(\tau)}$ that has zero mean. Note that $\nu_{k,n}^{(\tau)}$ is uncorrelated with the desired signal $s_{k,n}^{(\tau)}$ due to the independence  between each of the zero-mean symbols $s_{k,n}^{(\tau)}$, i.e., $\mathbb{E}\{s_{i,n}^{(\tau)} s_{j,n}^{(\tau)} \}=\mathbb{E}\{s_{i,n}^{(\tau)} s_{i,\jmath}^{(\tau)} \}=0$ for $n,j\in\mathcal{N}_d$.
The denominator of the SINR is obtained by
\begin{align}
    \mathbb{E}&\{|\nu_{k,n}^{(\tau)}|^2\} 
    = \sum\nolimits_{i=1 }^{K} p_i  \mathbb{E}\Big\{\left|\boldsymbol{v}_{k,n}^{\mathsf{H},(\tau)} \boldsymbol{D}_k \boldsymbol{h}_{i,n}^{(\tau)}\right|^2\Big\}   +{{\rho}^{\text{ICI},(\tau)}_{k,n}}  \nonumber\\
    &\;\;\;  - p_k \left|\mathbb{E}\Big\{\boldsymbol{v}_{k,n}^{\mathsf{H}, (\tau)} \boldsymbol{D}_k \boldsymbol{h}_{k,n}^{(\tau)}\right\}\Big|^2 + \sigma^2\mathbb{E}\Big\{\left|\boldsymbol{v}_{k,n}^{(\tau)} \boldsymbol{D}_k \right|^2\Big\}.
\end{align}
Here, the ICI term ${\rho}^{\text{ICI},(\tau)}_{k,n}$ is computed as
\begin{align}
    \rho^{\text{ICI},(\tau)}_{k,n} \nonumber 
    &= \sum\nolimits_{i=1}^{K} \mathbb{E}\left\{ \boldsymbol{v}_{k,n}^{\mathsf{H},(\tau)} \boldsymbol{D}_k \boldsymbol{\zeta}_{i,n}^{(\tau)}    \boldsymbol{\zeta}_{i,n}^{\mathsf{H},(\tau)} \boldsymbol{D}_k^{} \boldsymbol{v}_{k,n}^{(\tau)} \right\} \nonumber\\
    &= \sum\nolimits_{i=1}^{K} \mathbb{E}\left\{ \boldsymbol{v}_{k,n}^{\mathsf{H},(\tau)} \boldsymbol{D}_k  \text{diag}(\boldsymbol{\lambda}_{i,n}^{(\tau)}) \boldsymbol{D}_k^{} \boldsymbol{v}_{k,n}^{(\tau)} \right\},\label{eq:rho_ICI}
\end{align}
where the $l$-th element of the ICI power $\boldsymbol{\lambda}_{i,n}^{(\tau)}\in \mathbb{C}^{L}$ is given by
\begin{align}
{\lambda}_{i,n,l}^{(\tau)} &= \sum\nolimits_{\substack{ \jmath \neq n}}^{N-1} \mathbb{E}\{|s^{(\tau)}_{i,\jmath}|^2\} \mathbb{E}\{|{h}_{i,l,\jmath}|^2\}  \mathbb{E}\{|{{J}}_{i,l,n-\jmath}^{(\tau)}|^2\}   \nonumber \\
& = p_i \beta_{i,l}  \sum\nolimits_{\substack{ \jmath \neq n}}^{N-1} {B}_{i,l,n-j,n-j}^{(0)}
    =p_i \beta_{i,l} (1-{B}_{i,l,0,0}^{(0)}),\label{eq:define_lamda}
\end{align}
where we view the unknown channels over subcarriers other than $n$ as random variables instead of realizations to reduce computation complexity.  In the ideal case of no \ac{pn}, ${\lambda}_{i,n,l}^{(\tau)}=0$ due to ${B}_{i,l,0,0}^{(0)}=1$  and $\rho_{k,n}^{\text{ICI},(\tau)}=0$, which turns the \ac{sinr} and \ac{se} expressions in~\eqref{eq:SINR_lo_UatF} and~\eqref{eq:SE_UatF_PN_aware} to be the same as in~\cite[Eq. (27), (28)]{bjornson2020scalable}. 

\bibliographystyle{IEEEtran}
\bibliography{reference_list}
\end{document}

%% file: Acronym.tex
\DeclareAcronym{los}{
short=LoS,
long= line-of-sight,
}
\DeclareAcronym{dl}{
short=DL,
long= deep learning,
}
\DeclareAcronym{nn}{
short=NN,
long= neural network,
}
\DeclareAcronym{cnn}{
short=CNN,
long= convolutional neural network,
}
\DeclareAcronym{mmimo}{
short=mMIMO,
long= massive multiple-input  multiple-output,
}
\DeclareAcronym{siso}{
short=SISO,
long= single-input single-output,
}
\DeclareAcronym{ap}{
short=AP,
long= access point,
}
\DeclareAcronym{5g}{
short=5G,
long= fifth generation network,
}
\DeclareAcronym{ofdm}{
short=OFDM,
long= orthogonal frequency division multiplexing,
}
\DeclareAcronym{cpu}{
short=CPU,
long= central processing unit,
}
\DeclareAcronym{pn}{
short=PN,
long= phase noise,
}
\DeclareAcronym{ue}{
short=UE,
long= user equipment,
}
\DeclareAcronym{lo}{
short=LO,
long= local oscillator,
}
\DeclareAcronym{se}{
short=SE,
long= spectral efficiency,
}
\DeclareAcronym{sinr}{
short=SINR,
long= signal-to-interference-plus-noise ratio,
}
\DeclareAcronym{uatf}{
short=UatF,
long= use-and-then-forget,
}
\DeclareAcronym{lmmse}{
short=LMMSE,
long= linear minimum mean square error,
}
\DeclareAcronym{mse}{
short=MSE,
long= mean square error,
}
\DeclareAcronym{mmse}{
short=MMSE,
long= minimum mean square error,
}
\DeclareAcronym{hwi}{
short=HWI,
long= hardware impairment,
}
\DeclareAcronym{mr}{
short=MR,
long= maximum ratio,
}
\DeclareAcronym{pmmse}{
short=P-MMSE,
long= partial MMSE,
}
\DeclareAcronym{ha}{
short=HA,
long= hardware-aware,
}
\DeclareAcronym{fir}{
short=FIR,
long= finite impulse response,
}
\DeclareAcronym{dft}{
short=DFT,
long= discrete Fourier transform,
}
\DeclareAcronym{idft}{
short=IDFT,
long= inverse discrete Fourier transform,
}
\DeclareAcronym{rb}{
short=RB,
long= resource block,
}
\DeclareAcronym{isi}{
short=ISI,
long= intersymbol interference,
}
\DeclareAcronym{ici}{
short=ICI,
long= intercarrier interference,
}
\DeclareAcronym{cp}{
short=CP,
long= cyclic prefix,
}
\DeclareAcronym{cpe}{
short=CPE,
long= common phase error,
}
\DeclareAcronym{hu}{
short=HU,
long= hardware-unaware,
}
\DeclareAcronym{pna}{
short=PN-aware,
long= PN-aware,
}
\DeclareAcronym{pnu}{
short=PN-unaware,
long= PN-unaware,
}
\DeclareAcronym{csi}{
short=CSI,
long= channel state information,
}
\DeclareAcronym{lpmmse}{
short=LP-MMSE,
long= local-partial MMSE,
}
\DeclareAcronym{dcc}{
short=DCC,
long= dynamic cooperation clustering,
}

%% file: Coherence_block_diagram.tex
\begin{tikzpicture}
[font=\scriptsize, draw=black!100, x=0.27cm,y=0.3cm
]
\usetikzlibrary {arrows.meta} 

\definecolor{color1}{rgb}{0.83921568627451,0.152941176470588,0.156862745098039}
\definecolor{color2}{rgb} {0.12156862745098,0.466666666666667,0.705882352941177}
\definecolor{color3}{rgb}{1,0.498039215686275,0.0549019607843137}
\definecolor{color4}{rgb}{0.172549019607843,0.627450980392157,0.172549019607843}

\draw[thick,-Stealth] (-2.5,-5.5) -- (25,-5.5) node[anchor=north ] {Time};
\draw[thick,-Stealth] (-2.5,-5.5) -- (-2.5,9) node[anchor=north east] {Frequency};

\foreach \x in {0,1,...,19.5}
    \foreach \y in {-1,-.5,...,4.5}
        \fill[black!15!white] (\x,\y) rectangle (\x+0.95,\y+0.45);

\draw[xstep=20,ystep=0.5,black, thick] (0,-1) grid (20,5);

\foreach \x in {0,1,...,11}
        \fill[color4] (\x,-\x/2+4.5) rectangle (\x+0.95,-\x/2+4.95);
\foreach \x in {12,13,...,19}
        \fill[color4] (\x,\x/2-6.5) rectangle (\x+0.95,\x/2-6.95);


\foreach \x in {-2,-1,...,9}
\fill[color2] (0,\x/2) rectangle (0+0.95,\x/2+0.45);
\foreach \x in {2,3,...,9}
\fill[color2] (1,\x/2) rectangle (1+0.95,\x/2+0.45);

\fill[color4](0,4.72)rectangle (0+0.95,4.95);
\fill[color4](1,4.22)rectangle (1+0.95,4.45);

 
\fill[color4](5,7.5) rectangle (5+0.95,7.95) node[midway, above] {PP1};
\fill[color2](12,7.5) rectangle (12+0.95,7.95)node[midway, above] {PP2};

  \draw[Stealth-Stealth,thick] (0,5.5)   -- (20,5.5) node[midway, above] {$\tau_c=20$ OFDM Symbols};
    \draw[Stealth-Stealth,thick] (20.5,-1)   -- (20.5,5.) node[midway,right=-0.5] {\begin{tabular}{l}
         Coherence\\
         bandwidth\\
         $W_c$ 
    \end{tabular}};

\foreach \ysub in{4.75,-1.1}
    \draw[-Stealth,thick] (-2,1.5)   -- (0,\ysub) node[midway, left] {};
\draw[-Stealth] (-2,1.5)   -- (0,4.75) node[midway, below left=0. and 0.15] {\begin{tabular}{c}
       $N_c=12$\\
     Subcarriers 
\end{tabular} };



\filldraw[color2,draw=black,thick] (0, 0) rectangle (0.95,-.5);
\filldraw[black!15!white,draw=black,thick] (18., -.5) rectangle (19,-1.);

\draw[-Stealth,thick] (.5,-.5)   -- (.5,-2) node[near end, pos=1.8] {\begin{tabular}{c}
   Pilot sample  \\
    $s_{t_k,10}^{(1)}$
\end{tabular}};
\draw[-Stealth,thick] (18.5,-1)   -- (18.5,-2.2) node[midway, below =0.] {\begin{tabular}{c}
     Data sample \\
     $s_{t_k,9}^{(19)}$
\end{tabular}};

\node[]()at (-0.45,2) {\begin{tabular}{c} .\\ \end{tabular}};
\node[]()at (-0.45,1.7) {\begin{tabular}{c} .\\ \end{tabular}};
\node[]()at (-0.45,1.4) {\begin{tabular}{c} .\\ \end{tabular}};


\node[]()at (10,-2.8) {\begin{tabular}{c} \vdots \end{tabular}};
\node[]()at (10,-5) {\begin{tabular}{c} coherence block $r+1$ \end{tabular}};



\end{tikzpicture}

%% file: NN_structure.tex
\begin{tikzpicture}[every node/.style={minimum size=1cm},font=\scriptsize, >=stealth,nd/.style={draw,fill=blue!10,circle,inner sep=0pt},blk/.style={draw,fill=blue!10,minimum height=2. cm, minimum width=0.5 cm, text centered}, x=0.6cm,y=0.5cm]

\coordinate[](u) at (0,0);
  \node(c2r)[draw, blk,fill=blue!0, right=2 of u] {\rotatebox{90}{$\mathbb{C}2\mathbb{R}$}};
  \coordinate[right=0.25 of c2r](residual_1){};
  \node(dense_1)[draw, blk, right=0.6 of c2r] {\rotatebox{90}{Dense($M_1$, ReLu)}};
    \coordinate[below=2.5 of residual_1](residual_2);

    \node(dense_2)[draw, blk, right=0.5 of dense_1] {\rotatebox{90}{Dense($M_2$, ReLu)}};
   \node(dense_3)[draw, blk, right=0.5 of dense_2] {\rotatebox{90}{Dense(2K, Linear)}};
    \coordinate[right=0.25 of dense_3](residual_3){};
    \coordinate[below=2.5 of residual_3](residual_4){};
   \node(r2c)[draw, blk, fill=blue!0, right=0.5 of dense_3] {\rotatebox{90}{$\mathbb{R}2\mathbb{C}$}};
     \node[blk, fill opacity=0.0, dashed, line width=1, minimum width=4.3cm,minimum height=3cm](arden)at ($(c2r)!0.5!(r2c)$){};
     \node[above=-0.3cm of arden](arden_name){$f_{\boldsymbol{\alpha}}(\boldsymbol{y}_l^{\text{p}})$};
     \coordinate[right=2 of r2c](x_p){};

    \draw [->] (u)--node[above left=-0.3cm and -0.4cm]{$\boldsymbol{y}_l^{\text{p}}$}(c2r);
    \draw [->] (c2r)--(dense_1);
    \draw [->] (dense_1)--(dense_2);
    \draw [->] (dense_2)--(dense_3);
    \draw [->] (dense_3)--(r2c);
    \draw [-] (residual_1)--(residual_2);
    \draw [-] (residual_2)--node[below=-0.4cm]{Residual Connection}(residual_4);
    \draw [->] (residual_4)--(residual_3);
    \draw [->] (r2c)--node[above right=-0.3cm and -0.3cm]{$\hat{\boldsymbol{h}}_{l,n}$}(x_p);

\end{tikzpicture}

%% file: Results/SE_vs_LO.tex
\begin{tikzpicture}[font=\scriptsize]
\definecolor{color2}{rgb}{0.12156862745098,0.466666666666667,0.705882352941177}
\definecolor{color0}{rgb}{1,0.498039215686275,0.0549019607843137}
\definecolor{color1}{rgb}{0.172549019607843,0.627450980392157,0.172549019607843}
\definecolor{color3}{rgb}{0.83921568627451,0.152941176470588,0.156862745098039}
\definecolor{color4}{rgb}{0.580392156862745,0.403921568627451,0.741176470588235}
\definecolor{color5}{rgb}{0.549019607843137,0.337254901960784,0.294117647058824}
\definecolor{color6}{rgb}{0.890196078431372,0.466666666666667,0.76078431372549}
\definecolor{color7}{rgb}{0.737254901960784,0.741176470588235,0.133333333333333}

\begin{axis}[%
width=7.8cm,
height=5 cm,
at={(0,0)},
scale only axis,
xmin=1e-17,
xmax=1e-16,
xlabel style={font=\color{white!15!black}},
xlabel={\scriptsize LO quality coefficient, $\gamma_{\phi}=\gamma_{\varphi}$},
ymin=0,
ymax=14,
xmode=linear,
ymode=linear,
ylabel style={font=\color{white!15!black}},
ylabel={\scriptsize Achievable Spectral efficiency per UE [bit/s/Hz]},
axis background/.style={fill=white},
title style={font=\bfseries},
legend style={at={(0.45,0.95)}, font=\tiny, anchor=north, draw=white!15!black},
legend columns=2,
ytick distance={2},
xtick distance={1e-17},
]

\addplot [color=color1, dotted, line width=1.0pt, mark=square, mark options={solid, color1}]
  table[row sep=crcr]{%
1e-17  5.94383330704542\\
2e-17  5.36512989559428\\
4e-17  4.73617014879611\\
6e-17  4.33696867572506\\
8e-17  4.04170815344665\\
1e-16  3.80596786187626\\
};
\addlegendentry{Proposed (Distr.): PP1}

\addplot [color=color2, dotted, line width=1.0pt, mark=square, mark options={solid, color2}]
  table[row sep=crcr]{%
1e-17  2.15802461886262\\
2e-17  1.47981442998564\\
4e-17  0.975194369483064\\
6e-17  0.770135762628047\\
8e-17  0.654560971692418\\
1e-16  0.579724381398968\\
};
\addlegendentry{Proposed (Distr.): PP2}

\addplot [name path=PP1_low, color=color1, dashed, line width=1.pt, mark=triangle, mark options={solid, color1}]
  table[row sep=crcr]{%
1e-17  5.77626737968841\\
2e-17  4.98812642828819\\
4e-17  4.13656962705507\\
6e-17  3.62865193298017\\
8e-17  3.26773805168507\\
1e-16  2.99409182377315\\
};
\addlegendentry{Mismatched (Distr.): PP1}

\addplot [name path=PP2_low, color=color2, dashed, line width=1.pt, mark=triangle, mark options={solid, color2}]
  table[row sep=crcr]{%
1e-17  2.08264671885203\\
2e-17  1.40255425819434\\
4e-17  0.901046529866739\\
6e-17  0.702009870065457\\
8e-17  0.592306283079457\\
1e-16  0.521937696534418\\
};
\addlegendentry{Mismatched (Distr.): PP2}

\addplot [color=color1, line width=0.8pt, mark=o, mark options={solid, color1}]
  table[row sep=crcr]{%
1e-17  1.94996668781227\\
2e-17  1.29551948552723\\
4e-17  0.784990562742006\\
6e-17  0.565380280067559\\
8e-17  0.443219885602333\\
1e-16  0.36594056210677\\
};
\addlegendentry{Unaware (Distr.): PP1}

\addplot [color=color2, line width=0.8pt, mark=o, mark options={solid, color2}]
  table[row sep=crcr]{%
1e-17  1.50429532543392\\
2e-17  0.936006599238622\\
4e-17  0.569908292947348\\
6e-17  0.438132582433418\\
8e-17  0.365563669662293\\
1e-16  0.318561441634778\\
};
\addlegendentry{Unaware (Distr.): PP2}

\addplot [name path=PP1_upp, color=color1, dashdotted, line width=1.0pt, mark=x, mark size=3, mark options={solid, color1}]
  table[row sep=crcr]{%
1e-17  8.67692964281746\\
2e-17  8.05777386679638\\
4e-17  7.44154726813001\\
6e-17  7.08407532835877\\
8e-17  6.83189321156192\\
1e-16  6.63697763455869\\
};
\addlegendentry{Single-carrier (Distr.): PP1}

\addplot [name path=PP2_upp, color=color2, dashdotted, line width=1.0pt, mark=x, mark size=3, mark options={solid, color2}]
  table[row sep=crcr]{%
1e-17  6.164032459958\\
2e-17  5.32565781068932\\
4e-17  4.48699172016262\\
6e-17  4.00043827885411\\
8e-17  3.65925318841897\\
1e-16  3.39812943442964\\
};
\addlegendentry{Single-carrier (Distr.): PP2}

\addplot [color=black, line width=2.0pt, forget plot]
  table[row sep=crcr]{%
1e-17  13.7547541907785\\
2e-17  13.7547541907785\\
4e-17  13.7547541907785\\
6e-17  13.7547541907785\\
8e-17  13.7547541907785\\
1e-16  13.7547541907785\\
};
\draw[-Stealth] (axis cs: 9.e-17, 13.7)   -- (axis cs: 9e-17, 12) node[below=0cm] {\scriptsize \begin{tabular}{c}
    No PN \end{tabular}};

\draw[-Stealth] (axis cs: 7.e-17, 3.8)   -- (axis cs: 7.e-17, 4.9) node[above= -0.1cm] {\scriptsize \begin{tabular}{c}
    \shortstack{ \cite{papazafeiropoulos2021scalable} under single-carrier PN model}   \end{tabular}};
\draw[-Stealth] (axis cs: 7e-17, 7.)   -- (axis cs: 7e-17, 6.);

\draw[Bracket-Bracket] (axis cs: 1.5e-17, 5.8)   -- (axis cs: 1.5e-17, 1.7) node[above= -0.1cm] {\scriptsize \begin{tabular}{c}
    \shortstack{}   \end{tabular}};
\draw[Bracket-Bracket] (axis cs: 3e-17, 4.55)   -- (axis cs: 3e-17, 7.8) node[above= -0.1cm] {\scriptsize \begin{tabular}{c}
    \shortstack{}   \end{tabular}};
\draw[-Stealth] (axis cs: 1.5e-17, 3.6)   -- (axis cs: 2.3e-17, 3.6) node[above= -0.1cm] {\scriptsize \begin{tabular}{c}
    \shortstack{}   \end{tabular}};
\draw[-Stealth] (axis cs: 3e-17, 7)   -- (axis cs: 2.5e-17, 7) node[above= -0.1cm] {\scriptsize \begin{tabular}{c}
    \shortstack{}   \end{tabular}};

\node at (axis cs: 3e-17, 3.65){\scriptsize \begin{tabular}{c}
    \shortstack{ \textcolor{color2}{Optimism} \\\textcolor{color2}{bias area}}   \end{tabular}};

\node at (axis cs: 1.8e-17, 7.){\scriptsize \begin{tabular}{c}
    \shortstack{ \textcolor{color1}{Optimism} \\\textcolor{color1}{bias area}}   \end{tabular}};

\addplot [
    thick,
    color=green,
    fill=green, 
    fill opacity=0.05
] 
fill between[
    of=PP1_low and PP1_upp,
];

\addplot [
    thick,
    color=blue,
    fill=blue, 
    fill opacity=0.05
] 
fill between[
    of=PP2_low and PP2_upp,
];

\end{axis}

\end{tikzpicture}%

%% file: Results/SE_vs_UE_AP.tex
\begin{tikzpicture}[font=\scriptsize]
\definecolor{color2}{rgb}{0.12156862745098,0.466666666666667,0.705882352941177}
\definecolor{color0}{rgb}{1,0.498039215686275,0.0549019607843137}
\definecolor{color1}{rgb}{0.172549019607843,0.627450980392157,0.172549019607843}
\definecolor{color3}{rgb}{0.83921568627451,0.152941176470588,0.156862745098039}
\definecolor{color4}{rgb}{0.580392156862745,0.403921568627451,0.741176470588235}
\definecolor{color5}{rgb}{0.549019607843137,0.337254901960784,0.294117647058824}
\definecolor{color6}{rgb}{0.890196078431372,0.466666666666667,0.76078431372549}
\definecolor{color7}{rgb}{0.737254901960784,0.741176470588235,0.133333333333333}

\begin{axis}[%
width=7.8 cm,
height=5 cm,
at={(0,0)},
scale only axis,
xmin=1e-17,
xmax=1e-16,
xlabel style={font=\color{white!15!black}},
xlabel={\scriptsize LO quality coefficient, $\gamma_{\phi_{l}}$ or $\gamma_{\varphi_{k}}$},
ymin=0,
ymax=14.,
ylabel style={font=\color{white!15!black}},
ylabel={\scriptsize Achievable spectral efficiency per UE [bit/s/Hz]},
axis background/.style={fill=white},
title style={font=\bfseries},
xmajorgrids,
legend style={at={(0.45,0.98)}, font=\tiny, anchor=north, legend cell align=left, draw=white!15!black},
legend columns=2,
ytick distance={2},
xtick distance={1e-17},
]

\addplot [color=color1, dotted, line width=1pt, mark=square, mark options={solid, color1}]
  table[row sep=crcr]{%
1e-17  5.73444435624321\\
2e-17  5.15190898396721\\
4e-17  4.54095211099974\\
6e-17  4.15253340343854\\
8e-17  3.85483218826421\\
1e-16  3.60889432491248\\
};
\addlegendentry{Proposed (Distr.): $\gamma_{\phi_{l}}=0$}

\addplot [color=color2, dotted, line width=1pt, mark=square, mark options={solid, color2}]
  table[row sep=crcr]{%
1e-17  8.0609986505516\\
2e-17  7.40756156097271\\
4e-17  6.73605943670597\\
6e-17  6.33029549969894\\
8e-17  6.03378592390553\\
1e-16  5.79784075419601\\
};
\addlegendentry{Proposed (Distr.): $\gamma_{\varphi_{k}}=0$}

\addplot [name path=PP1_low,color=color1, dashed, line width=1pt, mark=triangle, mark options={solid, color1}]
  table[row sep=crcr]{%
1e-17  5.70570501498768\\
2e-17  5.13876992183836\\
4e-17  4.44346093740758\\
6e-17  3.92978776366395\\
8e-17  3.54543551868279\\
1e-16  3.23944487057962\\
};
\addlegendentry{Mismatched (Distr.): $\gamma_{\phi_{l}}=0$}

\addplot [name path=PP2_low,color=color2, dashed, line width=1pt, mark=triangle, mark options={solid, color2}]
  table[row sep=crcr]{%
1e-17  7.29808103359052\\
2e-17  6.42203289004168\\
4e-17  5.51988694410637\\
6e-17  4.98807063247679\\
8e-17  4.61106586417215\\
1e-16  4.31950094570669\\
};
\addlegendentry{Mismatched (Distr.): $\gamma_{\varphi_{k}}=0$}

\addplot [color=color1, line width=1pt, mark=o, mark options={solid, color1}]
  table[row sep=crcr]{%
1e-17  1.80485234564352\\
2e-17  0.889389943397042\\
4e-17  0.384795188190813\\
6e-17  0.233877597055431\\
8e-17  0.159194667873629\\
1e-16  0.121039616048602\\
};
\addlegendentry{Unaware (Distr.): $\gamma_{\phi_{l}}=0$}

\addplot [color=color2, line width=1pt, mark=o, mark options={solid, color2}]
  table[row sep=crcr]{%
1e-17  3.34243381176435\\
2e-17  2.53197262641902\\
4e-17  1.72895216232552\\
6e-17  1.28604913199436\\
8e-17  1.0102447048643\\
1e-16  0.825188473302237\\
};
\addlegendentry{Unaware (Distr.): $\gamma_{\varphi_{k}}=0$}

\addplot [name path=PP1_upp,color=color1, dashdotted, line width=1.0pt, mark size=2.5pt, mark=x, mark options={solid, color1}]
  table[row sep=crcr]{%
0  13.232194833224\\
1e-17  8.70322134263176\\
2e-17  8.05566508085203\\
4e-17  7.41247119827609\\
6e-17  7.04042773121918\\
8e-17  6.77869133952681\\
1e-16  6.57696586912615\\
};
\addlegendentry{Single-carrier (Distr.): $\gamma_{\phi_{l}}=0$}

\addplot [name path=PP2_upp,color=color2, dashdotted, line width=1.0pt, mark size=2.5pt, mark=x, mark options={solid, color2}]
  table[row sep=crcr]{%
0  13.232194833224\\
1e-17  10.436747669971\\
2e-17  9.95551740585719\\
4e-17  9.43983276381894\\
6e-17  9.12259494224902\\
8e-17  8.89083395290684\\
1e-16  8.70740295062821\\
};
\addlegendentry{Single-carrier (Distr.): $\gamma_{\varphi_{k}}=0$}

\addplot [name path=ideal,color=black, solid, line width=2.0pt]
  table[row sep=crcr]{%
0  13.7547541907785\\
1e-16  13.7547541907785\\
};

\draw[-Stealth] (axis cs: 9.5e-17, 13.7)   -- (axis cs: 9.5e-17, 12) node[below=0cm] {\scriptsize \begin{tabular}{c}
    No PN \end{tabular}};
    
\node at (axis cs: 4.e-17, 8.7) {\scriptsize 
\begin{tabular}{c}
    \shortstack{ \cite{papazafeiropoulos2021scalable} under single-carrier PN model}   
\end{tabular}};

\draw[-Stealth] (axis cs: 2.e-17, 10)   -- (axis cs: 3.e-17, 9.2);
\draw[-Stealth] (axis cs: 5.e-17, 7.2)   -- (axis cs: 4e-17, 8.1);

\draw[Bracket-Bracket] (axis cs: 7.5e-17, 8.9)   -- (axis cs: 7.5e-17, 4.72);
\draw[Bracket-Bracket] (axis cs: 8.2e-17, 6.7)   -- (axis cs: 8.2e-17, 3.5);

\draw[-Stealth] (axis cs: 7.5e-17, 8.)   -- (axis cs: 8.2e-17, 8) node[above= -0.1cm] {\scriptsize \begin{tabular}{c}
    \shortstack{}   \end{tabular}};
\draw[-Stealth] (axis cs: 8.2e-17, 4.3)   -- (axis cs: 8.6e-17, 4.3) node[above= -0.1cm] {\scriptsize \begin{tabular}{c}
    \shortstack{}   \end{tabular}};

\node at (axis cs: 9e-17, 7.9){\scriptsize \begin{tabular}{c}
    \shortstack{ \textcolor{color2}{Optimism} \\\textcolor{color2}{bias area}}  \end{tabular}};

\node at (axis cs: 9.2e-17, 4.4){\scriptsize \begin{tabular}{c}
    \shortstack{ \textcolor{color1}{Optimism} \\\textcolor{color1}{bias area}}  \end{tabular}};

\addplot [
    thick,
    color=green,
    fill=green, 
    fill opacity=0.05
] 
fill between[
    of=PP1_low and PP1_upp,
];

\addplot [
    thick,
    color=blue,
    fill=blue, 
    fill opacity=0.05
] 
fill between[
    of=PP2_low and PP2_upp,
];

\end{axis}

\end{tikzpicture}

%% file: Results/SE_vs_K.tex
\begin{tikzpicture}[font=\scriptsize]
\definecolor{color2}{rgb}{0.12156862745098,0.466666666666667,0.705882352941177}
\definecolor{color0}{rgb}{1,0.498039215686275,0.0549019607843137}
\definecolor{color1}{rgb}{0.172549019607843,0.627450980392157,0.172549019607843}
\definecolor{color3}{rgb}{0.83921568627451,0.152941176470588,0.156862745098039}
\definecolor{color4}{rgb}{0.580392156862745,0.403921568627451,0.741176470588235}
\definecolor{color5}{rgb}{0.549019607843137,0.337254901960784,0.294117647058824}
\definecolor{color6}{rgb}{0.890196078431372,0.466666666666667,0.76078431372549}
\definecolor{color7}{rgb}{0.737254901960784,0.741176470588235,0.133333333333333}

\begin{axis}[%
width=7.8cm,
height=4. cm,
at={(0,0)},
scale only axis,
xmin=1,
xmax=50,
xlabel style={font=\color{white!15!black}},
xlabel={\scriptsize Number of UEs $K$},
ymin=0,
ymax=14,
ylabel style={font=\color{white!15!black}},
ylabel={\scriptsize Achievable SE per UE [bit/s/Hz]},
axis background/.style={fill=white},
title style={font=\bfseries},
xmajorgrids,
minor xtick={1, 2, 3, 4, 5, 6, 7, 8, 9, 10,20,30,40,50},
legend style={at={(0.38,0.9)}, anchor=north, font=\tiny, legend cell align=left, draw=white!15!black},
legend columns=2,
ytick distance={2},
xtick distance={10},
xmode=log,
xtick={1,10,20,50}, 
xticklabels={$10^0$,$10^1$,20,50} 
]

\addplot [color=color1, dotted, line width=1.0pt, mark=square, mark options={solid, color1}]
  table[row sep=crcr]{%
1  5.1627044128957\\
3  4.91405126701253\\
5  4.70138584420239\\
10  4.35186595505371\\
20  3.82653494068536\\
50  2.12583886931534\\
};
\addlegendentry{MMSE (Proposed)}


\addplot [name path=MMSElow, color=color1, dashed, line width=1.0pt, mark=triangle, mark options={solid, color1}]
  table[row sep=crcr]{%
1  4.19524720689378\\
3  4.09915753400677\\
5  4.08307804207238\\
10  3.97445158977429\\
20  3.64596285281484\\
50  2.12126264517682\\
};
\addlegendentry{MMSE (Mismatched)}


\addplot [color=color1, line width=1.0pt, mark=o, mark options={solid, color1}]
  table[row sep=crcr]{%
1  1.02481828046844\\
3  0.881036824254423\\
5  0.763926788774802\\
10  0.721403676531116\\
20  0.592921494623534\\
50  0.573738144941694\\
};
\addlegendentry{MMSE (Unaware)}


\addplot [name path=MMSEupp, color=color1, dashdotted, line width=1.0pt, mark=x,mark size=2.5, mark options={solid, color1}]
  table[row sep=crcr]{%
1  8.58872032831695\\
3  7.74776053646426\\
5  7.4690749089152\\
10  7.11555379935587\\
20  6.21753998209174\\
50  4.13594723796447\\
};
\addlegendentry{MMSE (Single-carrier)}


\addplot [color=black, line width=2.0pt, forget plot]
  table[row sep=crcr]{%
1  13.8093761002373\\
3  13.654051805808\\
5  13.5745679911089\\
10  13.2649547091478\\
20  12.2079082136578\\
50  5.75270581319395\\
};

\addplot [
    thick,
    color=green,
    fill=green, 
    fill opacity=0.05
] 
fill between[
    of=MMSElow and MMSEupp,
soft clip={domain=1:50}
];


\draw[-Stealth] (axis cs: 30, 9.)   -- (axis cs: 22, 8) node[below= -0.1cm] {\scriptsize \begin{tabular}{c}
    MMSE: No PN \end{tabular}};

\node at (axis cs: 2, 6.){\scriptsize \begin{tabular}{c}
    \shortstack{\textcolor{color1}{Optimism bias area}}   \end{tabular}};

\node at (axis cs: 5, 3.)[] {\scriptsize \begin{tabular}{c}
    \shortstack{\cite{papazafeiropoulos2021scalable} under single-carrier PN model}
\end{tabular}};
\draw[-Stealth] (axis cs: 4.5, 7.5)   -- (axis cs:4.5, 3.5); 

\end{axis}
\end{tikzpicture}%

%% file: Results/SE_vs_L_CLO.tex
\begin{tikzpicture}[font=\scriptsize]
\definecolor{color2}{rgb}{0.12156862745098,0.466666666666667,0.705882352941177}
\definecolor{color0}{rgb}{1,0.498039215686275,0.0549019607843137}
\definecolor{color1}{rgb}{0.172549019607843,0.627450980392157,0.172549019607843}
\definecolor{color3}{rgb}{0.83921568627451,0.152941176470588,0.156862745098039}
\definecolor{color4}{rgb}{0.580392156862745,0.403921568627451,0.741176470588235}
\definecolor{color5}{rgb}{0.549019607843137,0.337254901960784,0.294117647058824}
\definecolor{color6}{rgb}{0.890196078431372,0.466666666666667,0.76078431372549}
\definecolor{color7}{rgb}{0.737254901960784,0.741176470588235,0.133333333333333}

\begin{axis}[%
width=7.8cm,
height=3.5 cm,
at={(0,0)},
scale only axis,
xmin=5,
xmax=100,
xlabel style={font=\color{white!15!black}},
xlabel={\scriptsize Number of APs, $L$},
ymin=1.5,
ymax=6.5,
ylabel style={font=\color{white!15!black}},
ylabel={\scriptsize Achievable SE per UE [bit/s/Hz]},
axis background/.style={fill=white},
title style={font=\bfseries},
xmajorgrids,
ymajorgrids,
legend style={at={(0.37,0.999)}, anchor=north, legend cell align=left, font=\tiny, draw=white!15!black},
legend columns=2,
ytick distance={1},
xtick distance={10},
xmode=log,
xtick={5,6,7,8,9,10,20,30,40,50,60,70,80,90,100}, 
xticklabels={$5$,,,,,$10^1$,,,,,,,,,$10^2$}, 
xminorgrids,
]

\addplot [color=color1,dashdotted, line width=1.0pt, mark=diamond, mark size=3pt, mark options={solid, color1}]
  table[row sep=crcr]{%
5  2.40309940842885\\
10  3.29339231374291\\
15  3.97867850915645\\
20  4.29372529651706\\
50  5.15089302709226\\
100  5.60480124906414\\
};
\addlegendentry{Proposed (Centr.) w/ DL}

\addplot [color=color2, dashdotted, line width=1.0pt, mark=x, mark size=3pt, mark options={solid, color2}]
  table[row sep=crcr]{%
5  2.34949045872774\\
10  3.22796338375705\\
15  3.85647925165609\\
20  4.17210269249384\\
50  4.89827074655187\\
100  5.2145371569613\\
};
\addlegendentry{Proposed (Centr.) w/ LMMSE}

\addplot [color=color1, dashdotted, line width=1.0pt, mark=diamond, mark size=3pt, mark options={solid, color1}, forget plot]
  table[row sep=crcr]{%
5  2.47110825693089\\
10  3.48685517439682\\
15  4.37423502346381\\
20  4.76436176342668\\
50  5.66268848082029\\
100  5.99727142387683\\
};
\addplot [color=color2, dashdotted, line width=1.0pt, mark=x, mark size=3pt,mark options={solid, color2}, forget plot]
  table[row sep=crcr]{%
5  2.45347588163167\\
10  3.41746499280076\\
15  4.23742410925436\\
20  4.61152300711988\\
50  5.53942927360393\\
100  5.86111319581809\\
};
\addplot [color=color3, dotted, line width=1.0pt, mark=square, mark options={solid, color3}]
  table[row sep=crcr]{%
5  2.22833767024064\\
10  3.05140301694621\\
15  3.71342095085952\\
20  4.00679906871948\\
50  4.89105280490922\\
100  5.13465881044454\\
};
\addlegendentry{Proposed (Distr.)}

\addplot [color=color3, dashed, line width=1.0pt, mark=triangle, mark options={solid, color3}]
  table[row sep=crcr]{%
5  2.19030345573139\\
10  3.02104866998941\\
15  3.70456686714686\\
20  3.93445692789672\\
50  4.76876255006783\\
100  5.19503974198274\\
};
\addlegendentry{Mismatched (Distr.)}

\addplot [color=color3, line width=1.0pt, mark=o, mark options={solid, color3}]
  table[row sep=crcr]{%
5  1.52340843804768\\
10  1.86507104923039\\
15  2.08032297451503\\
20  2.15888248130858\\
50  2.1877764359247\\
100  2.22782002605332\\};
\addlegendentry{Unaware (Distr.)}

\addplot [color=black, line width=1.5pt, forget plot]
  table[row sep=crcr]{%
5  2.67205176075763\\
10  3.71344711668013\\
15  4.58191691231127\\
20  5.00054648262609\\
50  5.98192474400887\\
100  6.34018309945403\\
};

\draw (axis cs: 85, 5.35) ellipse (0.12cm and 0.18cm);
\draw[-Stealth] (axis cs: 85.,5.15)   -- (axis cs: 85,4.2) node[pos=1.27] {\scriptsize \begin{tabular}{c}
    Iter. 1
\end{tabular}};
\draw (axis cs: 55, 5.65) ellipse (0.1cm and 0.15cm);
\draw[-Stealth] (axis cs: 55,5.50)   -- (axis cs: 55,4.0) node[pos=1.1] {\scriptsize \begin{tabular}{c}
    Iter. 3
\end{tabular}};
\draw[-Stealth] (axis cs: 7, 3.2)   -- (axis cs: 7, 3.7) node[pos=1.4] {\scriptsize \begin{tabular}{c}
    Proposed (Centr.) \\w/ true channel
\end{tabular}};
\end{axis}

\end{tikzpicture}%

%% file: Results/Channel_vs_L_CLO.tex
\begin{tikzpicture}[font=\scriptsize]
\definecolor{color2}{rgb}{0.12156862745098,0.466666666666667,0.705882352941177}
\definecolor{color0}{rgb}{1,0.498039215686275,0.0549019607843137}
\definecolor{color1}{rgb}{0.172549019607843,0.627450980392157,0.172549019607843}
\definecolor{color3}{rgb}{0.83921568627451,0.152941176470588,0.156862745098039}
\definecolor{color4}{rgb}{0.580392156862745,0.403921568627451,0.741176470588235}
\definecolor{color5}{rgb}{0.549019607843137,0.337254901960784,0.294117647058824}
\definecolor{color6}{rgb}{0.890196078431372,0.466666666666667,0.76078431372549}
\definecolor{color7}{rgb}{0.737254901960784,0.741176470588235,0.133333333333333}

\begin{axis}[%
        width=7.5cm,
        height=3.5cm,
        at={(0,0)},
        scale only axis,
        xmin=5,
        xmax=100,
        xlabel style={font=\color{white!15!black}},
        xlabel={\scriptsize Number of APs, $L$},
        ymode=log,
        ymin=0.01,
        ymax=0.2,
        yminorticks=true,
        ylabel style={font=\color{white!15!black}},
        ylabel={\scriptsize MSE of Normalized Channel Estimate },
        axis background/.style={fill=white},
        title style={font=\bfseries},
        xmajorgrids,
        ymajorgrids,
        yminorgrids,
        legend style={at={(0.58,0.90)}, anchor=north, legend cell align=left, draw=white!15!black},
        legend columns=2, font=\tiny,
        xtick distance={10},
    ]

\addplot [color=color1, dashdotted, line width=1.0pt,mark size=3pt, mark=diamond, mark options={solid, color1}]
  table[row sep=crcr]{%
5  0.0565425604581833\\
10  0.0383755899965763\\
15  0.0300272520184517\\
20  0.0238519106060266\\
50  0.0146247418597341\\
100  0.0126242684200406\\
};
\addlegendentry{Proposed (Centr.) w/ DL}

\addplot [color=color2, dashdotted, line width=1.0pt, mark size=3pt,mark=x, mark options={solid, color2}]
  table[row sep=crcr]{%
5  0.0584789964327692\\
10  0.0391489764622509\\
15  0.0329750233915681\\
20  0.0285639268373347\\
50  0.0214024851238899\\
100  0.0206632204528791\\
};
\addlegendentry{Proposed (Centr.) w/ LMMSE}

\addplot [color=color1, dashdotted, line width=1.0pt,mark size=3pt, mark=diamond, mark options={solid, color1}, forget plot]
  table[row sep=crcr]{%
5  0.0425550937652588\\
10  0.0250141099095345\\
15  0.0204161633998156\\
20  0.016973489895463\\
50  0.0111619746312499\\
100  0.0104318214580417\\
};
\addplot [color=color2, dashdotted, line width=1.0pt,mark size=3pt, mark=x, mark options={solid, color2}, forget plot]
  table[row sep=crcr]{%
5  0.0439064019101151\\
10  0.0260216929920535\\
15  0.0224928149665647\\
20  0.017945656368215\\
50  0.0119305079275725\\
100  0.0115648459388985\\
};
\addplot [color=color3, dotted, line width=1.0pt, mark=square, mark options={solid, color3}]
  table[row sep=crcr]{%
5  0.0656923179633802\\
10  0.0496649628585044\\
15  0.0448060818910021\\
20  0.0405090519858632\\
50  0.025352710952048\\
100  0.0245061868087466\\
};
\addlegendentry{Proposed (Distr.)}

\addplot [color=color3, dashed, line width=1.0pt, mark=triangle, mark options={solid, color3}]
  table[row sep=crcr]{%
5  0.0691992107907005\\
10  0.0538342237359443\\
15  0.0473169926805724\\
20  0.0445399186186131\\
50  0.0299622416140163\\
100  0.0283651775035545\\
};
\addlegendentry{Mismatched (Distr.)}

\addplot [color=color2, line width=1.0pt, mark=o, mark options={solid, color2}]
  table[row sep=crcr]{%
5  0.187915781545628\\
10  0.178540952291551\\
15  0.176898465241096\\
20  0.177625706286982\\
50  0.177220001128913\\
100  0.174726039156272\\
};
\addlegendentry{Unaware (Distr.)}

\draw (axis cs: 93, 0.0165) ellipse (0.12cm and 0.35cm);
\draw[-Stealth] (axis cs: 93.,0.0216)   -- (axis cs: 93, 0.036) node[pos=1.2] {\scriptsize \begin{tabular}{c}
    Iter. 1
\end{tabular}};
\draw (axis cs: 80, 0.0112) ellipse (0.1cm and 0.15cm);
\draw[-Stealth] (axis cs: 80, 0.0125)   -- (axis cs: 80, 0.036) node[pos=1.1] {\scriptsize \begin{tabular}{c}
    Iter. 3
\end{tabular}};
\end{axis}

\end{tikzpicture}%

%% file: Results/SE_vs_power_CLO.tex
\begin{tikzpicture}[font=\scriptsize]
\definecolor{color2}{rgb}{0.12156862745098,0.466666666666667,0.705882352941177}
\definecolor{color0}{rgb}{1,0.498039215686275,0.0549019607843137}
\definecolor{color1}{rgb}{0.172549019607843,0.627450980392157,0.172549019607843}
\definecolor{color3}{rgb}{0.83921568627451,0.152941176470588,0.156862745098039}
\definecolor{color4}{rgb}{0.580392156862745,0.403921568627451,0.741176470588235}
\definecolor{color5}{rgb}{0.549019607843137,0.337254901960784,0.294117647058824}
\definecolor{color6}{rgb}{0.890196078431372,0.466666666666667,0.76078431372549}
\definecolor{color7}{rgb}{0.737254901960784,0.741176470588235,0.133333333333333}

\begin{axis}[%
width=7.8cm,
height=5. cm,
at={(0,0)},
scale only axis,
xmin=1,
xmax=1000,
xlabel style={font=\color{white!15!black}},
xlabel={\scriptsize UE transmit power $p$ [mW]},
ymin=0,
ymax=6.5,
ylabel style={font=\color{white!15!black}},
ylabel={\scriptsize Achievable spectral efficiency per UE [bit/s/Hz]},
axis background/.style={fill=white},
title style={font=\bfseries},
xmajorgrids,
ymajorgrids,
legend style={at={(0.31,0.999)}, anchor=north, legend cell align=left, draw=white!15!black},
legend columns=1,
ytick distance={1},
xtick distance={10},
xmode=log,
xtick={1,10,20,30,40,50,60,70,80,90,100,200,300,400,500,600,700,800,900,1000}, 
xticklabels={$10^0$,$10^1$,,,,,,,,,$10^2$,,,,,,,,,$10^3$}, 
xminorgrids,
]

\addplot [color=color1, dashdotted, line width=1.0pt, mark size=3pt, mark=diamond, mark options={solid, color1}]
  table[row sep=crcr]{%
1  0.525439321842188\\
10  2.40621456379174\\
50  3.81910276802107\\
100  4.46796534838594\\
200  4.91515766143799\\
1000 5.34488\\
};
\addlegendentry{Proposed (Centr.) w/ DL}

\addplot [color=color2, dashdotted, line width=1.0pt, mark size=3pt, mark=x, mark options={solid, color2}]
  table[row sep=crcr]{%
1  0.566294746901324\\
10  2.40241249774746\\
50  3.70849533131813\\
100  4.10915645384868\\
200  4.51466603073865\\
1000 4.95628\\
};
\addlegendentry{Proposed (Centr.) w/ LMMSE}

\addplot [color=color1, dashdotted, line width=1.0pt, mark size=3pt,mark=diamond, mark options={solid, color1}, forget plot]
  table[row sep=crcr]{%
1  0.569356698425103\\
10  2.5128876150636\\
50  4.18503840962646\\
100  4.77665307433125\\
200  5.22820514363688\\
1000   5.7306\\
};
\addplot [color=color2, dashdotted, line width=1.0pt, mark=x, mark size=3pt,mark options={solid, color2}, forget plot]
  table[row sep=crcr]{%
1  0.569650014479115\\
10  2.50990064637007\\
50  4.17186251221869\\
100  4.74422819817652\\
200  5.16550679156001\\
1000 5.5629\\
};
\addplot [color=color3, dotted, line width=1.0pt, mark=square, mark options={solid, color3}]
  table[row sep=crcr]{%
1  0.559482077811965\\
10  2.29906331059284\\
50  3.55443798679154\\
100  3.95136381465999\\
200  4.40979906671559\\
1000 4.8865\\
};
\addlegendentry{Proposed (Distri.)}

\addplot [color=color3, dashed, line width=1.0pt, mark=triangle, mark options={solid, color3}]
  table[row sep=crcr]{%
1  0.55927551051287\\
10  2.2996753546431\\
50  3.51731073912257\\
100  3.85537858671205\\
200  4.20224530612084\\
1000 4.6294\\
};
\addlegendentry{Mismatched}

\addplot [color=color3, line width=1.0pt, mark=o, mark options={solid, color3}]
  table[row sep=crcr]{%
1  0.519758907624441\\
10  1.80674038897441\\
50  2.28351194155783\\
100  2.32514166136668\\
200  2.31629239203417\\
1000 2.3367 \\
};
\addlegendentry{Unaware}

\addplot [color=black, line width=1.5pt, forget plot]
  table[row sep=crcr]{%
1  0.877528442654905\\
10  2.834370044737\\
50  4.44956465727753\\
100  5.04736864903349\\
200  5.52945859514336\\
1000  6.1603\\
};

\draw (axis cs: 700, 5.05) ellipse (0.12cm and 0.22cm);
\draw[-Stealth] (axis cs: 700.,4.8)   -- (axis cs: 700, 4.) node[pos=1.27] {\scriptsize \begin{tabular}{c}
    Iter. 1
\end{tabular}};
\draw (axis cs: 800, 5.6) ellipse (0.1cm and 0.18cm);
\draw[-Stealth] (axis cs: 800, 5.78)   -- (axis cs: 600,6.2) node[pos=1.2] {\scriptsize \begin{tabular}{c}
    Iter. 3
\end{tabular}};


\draw[-Stealth] (axis cs: 3, 1.8)   -- (axis cs: 3, 2.6) node[pos=1.3] {\scriptsize \begin{tabular}{c}
    Proposed (Centr.) \\w/ True Channel
\end{tabular}};

\end{axis}

\end{tikzpicture}%

%% file: Results/CPE_threshold_tau1.tex
\begin{tikzpicture}[font=\scriptsize]
\definecolor{color2}{rgb}{0.12156862745098,0.466666666666667,0.705882352941177}
\definecolor{color0}{rgb}{1,0.498039215686275,0.0549019607843137}
\definecolor{color1}{rgb}{0.172549019607843,0.627450980392157,0.172549019607843}
\definecolor{color3}{rgb}{0.83921568627451,0.152941176470588,0.156862745098039}
\definecolor{color4}{rgb}{0.580392156862745,0.403921568627451,0.741176470588235}
\definecolor{color5}{rgb}{0.549019607843137,0.337254901960784,0.294117647058824}
\definecolor{color6}{rgb}{0.890196078431372,0.466666666666667,0.76078431372549}
\definecolor{color7}{rgb}{0.737254901960784,0.741176470588235,0.133333333333333}

\begin{axis}[%
width=5.82cm,
height=5.82 cm,
at={(0,0)},
scale only axis,
xmin=0.5,
xmax=1.5,
xlabel style={font=\color{white!15!black}},
xlabel={\scriptsize Real },
ymin=-1.,
ymax=1.,
ylabel style={font=\color{white!15!black}},
ylabel={\scriptsize Imaginary},
axis background/.style={fill=white},
title style={font=\bfseries},
xmajorgrids,
ymajorgrids,
legend style={at={(0.5,0.995)}, anchor=north, legend cell align=left, draw=white!15!black},
legend columns=1,
ytick distance={1},
xtick={0.5,1,1.5}, 
]

\addplot [color=color2, only marks, mark size=2pt, mark=o, mark options={line width=1,solid, color2}]
  table[row sep=crcr]{%
0.946302338678669	-0.0148536466488043\\
1.06865900578725	0.130471755012889\\
1.10227758768971	0.135140777071654\\
1.03907340679791	0.153781997071444\\
1.01299332354488	0.0206899978408687\\
1.03349680929085	0.191332077576046\\
1.02523965357463	0.179687319660134\\
0.919946192737317	-0.302243249441261\\
0.968312999524003	0.0566189419539824\\
1.04044063799482	-0.254776716094235\\
1.09742998384992	0.297868001974597\\
0.962601669510442	0.120139487647326\\
0.812131596119682	0.597976064627219\\
1.02332613070996	-0.1567832492967\\
0.741764674090855	0.167735956173331\\
1.07069988064868	0.0102618249792028\\
1.02172738805051	0.0563183621655968\\
0.955914999612786	0.109207995408208\\
0.943804460110128	-0.127235560642119\\
0.939006771754559	0.169958610082973\\
1.01142612778918	-0.035406164772106\\
0.943366274582126	0.274490450317631\\
1.00382354153292	-0.0672655347115839\\
0.996581102127523	-0.192487697276525\\
0.898801982796326	-0.0412880735158899\\
0.969203271564074	-0.167385233133476\\
0.849082897130127	-0.191173074910948\\
1.05266322880949	-0.144437764064739\\
0.7646756055065	-0.0376405389279491\\
0.949423818564866	0.127292388145572\\
0.85771922303863	-0.254382271976068\\
0.925007938297632	-0.00666780585728147\\
1.05369605714903	0.0111225791971183\\
0.813472614537129	0.629537661738005\\
0.924343526991647	-0.174442224297999\\
1.03270356588948	-0.13537139709951\\
0.947027495893776	0.436520366955317\\
0.926797990377349	0.0789033389044219\\
0.894712734521268	-0.126617058576783\\
0.98638250353393	-0.48549197017354\\
0.906015968908111	-0.0306653689763889\\
1.01288146292315	0.0605650842110689\\
0.997180583984344	0.0122646137906585\\
0.989585209502808	-0.142571872191028\\
1.14253858015868	0.0454627899313259\\
1.06804338135693	0.185780951492356\\
1.08290054016934	-0.0693334376607242\\
0.999945626230428	-0.116817254579355\\
0.868583949157677	0.0205713637800924\\
0.908719313516292	0.0667118833696452\\
0.851429395037448	0.0350491807953584\\
0.952070874861739	-0.0610585143901303\\
0.921680100200729	0.120584194556021\\
0.989888227600398	0.0887644999512018\\
0.933871236120927	-0.117038701056021\\
0.886448109623643	0.180730169995718\\
1.06842155360903	0.079776109921609\\
1.08318786623141	0.104703735684053\\
0.992936158499292	-0.209584348529682\\
0.96475593739114	0.216643891700462\\
0.63264504424462	-0.543228988162173\\
0.952460709060775	0.233506103602922\\
0.915203258633331	-0.28370781095619\\
1.09447270959053	0.113165683033547\\
0.974576818190348	0.339110597146651\\
1.14703492897478	0.173733432055242\\
0.952888062569408	0.241308880732814\\
0.918472373924142	0.0101002552740632\\
0.733078352812291	-0.813311207636782\\
0.900460697110521	0.344891571171663\\
0.795672194619401	0.425395281064822\\
0.979709573832779	-0.118931613885144\\
1.09516414432727	0.0718946347787758\\
1.04310011191477	0.227894688357207\\
1.06265080229912	0.0522325252236535\\
0.907308647846518	0.366373313749884\\
1.02343655581495	-0.116969581726539\\
0.799200001569311	0.277149168303825\\
1.02960374040949	0.179284442981374\\
0.956025071559343	0.0751665420552252\\
1.21177942860183	-0.119804771470028\\
0.788399188527234	-0.21620772667089\\
0.929640111132693	0.152084768022766\\
1.11307475694621	0.198627850121243\\
0.984230778721202	-0.238112696178866\\
0.98569525647576	0.0686103591380668\\
0.885653796805132	0.250929814051949\\
1.04157618770665	-0.0284525668865661\\
0.83649610071323	-0.394274558711051\\
1.01956536360817	-0.206135141143913\\
1.0254160681788	0.282176825321646\\
1.10416374681333	0.244861034013889\\
0.888371667843922	0.437250111101853\\
0.883004138889975	0.0544457684866615\\
1.09473400436923	-0.46191052138794\\
0.947929708062071	0.101109757335854\\
1.06518354148039	0.205427684009399\\
1.00451457388342	-0.195800076689844\\
0.912821466247048	-0.101779032558756\\
1.04212647517043	-0.100196986600179\\
1.0241829219727	-0.147349042510071\\
0.87577400078613	-0.0570440857895193\\
1.0727235064642	0.394687853302289\\
1.05087469190663	0.115907806395303\\
0.982419303984151	-0.120462152142651\\
0.930703091330218	0.174836366509397\\
0.854138243474437	0.239936560741636\\
0.967025461286634	-0.141986980732771\\
0.916300754761473	-0.19021329900522\\
0.934544798916396	-0.229920420244337\\
1.36312653642098	-0.0672047569266761\\
1.05393882717998	-0.00997015327399993\\
0.938767149386563	0.339480739542215\\
0.934214998698069	0.234378909470688\\
0.963926191061361	0.090815323627568\\
0.982477439727094	-0.119346446641182\\
1.0290975280765	-0.231352169622776\\
0.910327547479443	0.0680191613682897\\
0.952723097535026	0.358495845777057\\
1.04834069752333	0.15351217253131\\
0.990706986187783	-0.120307306451207\\
0.914636027080902	-0.244427248800067\\
1.02443835150697	0.039411914850121\\
1.04825175213502	-0.162532066270424\\
1.00845321101196	0.0897374688655425\\
0.856088903673571	-0.404316189240266\\
0.998363652613636	-0.199525157599699\\
1.06257856726109	0.0329850999312043\\
0.953731537579647	0.316993164165255\\
0.847586822956928	0.233216039285684\\
1.10930949540639	-0.0612703295761021\\
1.02744337723069	-0.147017224388738\\
0.978879240998642	0.381463031918942\\
1.05683323738138	0.502414438599602\\
1.02028558200701	-0.217501655986607\\
0.867028452953344	-0.262489279310157\\
0.987610123555837	-0.105634061877688\\
1.01604227134855	-0.0600140333604929\\
0.969296883958603	0.0663506140523498\\
1.06936067025367	0.293862345946105\\
0.887791215426431	0.199555071727702\\
0.951036295016323	-0.184347205154439\\
0.947355206051817	0.145081559623623\\
0.956147977702862	0.135444698407344\\
0.985542376193672	-0.236533270068212\\
0.911264917784743	-0.593448532303026\\
0.99962424684082	0.334943654281963\\
0.905029552695827	-0.500560066758567\\
0.918890172928299	0.111183974035546\\
1.0863010108356	0.444802554267436\\
1.10162032437158	0.117280049029338\\
1.00761482960666	-0.111026980587852\\
0.954799191677759	-0.184891796221034\\
1.0771015155223	-0.209977064818796\\
1.09369160434665	0.259538780093883\\
0.885093659473555	-0.403083687804514\\
0.87452567715668	-0.0969277916079594\\
1.00345885277663	0.195758064179182\\
0.859897190334231	-0.0201539976272551\\
0.914489076877378	-0.271267430649672\\
0.753162585698547	-0.131387346508445\\
0.937138745557415	-0.248351142077637\\
1.13579575119409	0.173345024126067\\
1.02457156404315	-0.000261025455240808\\
1.16005508090304	0.0501891828992183\\
0.975682763971771	-0.0134515544654978\\
0.94855749767296	-0.314056778208205\\
0.935151361191195	-0.252386384697167\\
0.843967727026356	0.124662025484945\\
1.07620790562929	0.00961854427693715\\
0.967569140549778	-0.0597679890390657\\
0.982552598560719	-0.0769965555686542\\
0.857482908303283	-0.016718039430846\\
1.00746542921012	0.0954940766156919\\
1.07217418014322	0.0706376593265835\\
1.0279558375295	0.133372229927361\\
0.85704499174421	-0.151473386329096\\
0.980022430159365	-0.13921994776461\\
1.21599757113573	-0.0298851073337935\\
1.02981012662875	0.127166293378236\\
0.955310056166375	0.515605886738551\\
1.11005692376392	-0.245533697207479\\
0.974384752567445	0.0626866407791443\\
1.17420630970445	0.135838720106448\\
0.994326228800661	0.319975315745116\\
1.09680801817119	0.129550406030998\\
0.978996665307218	0.164265554023658\\
0.970376971946675	0.320460324053956\\
0.93804902097841	0.126773949858037\\
0.909109908986433	0.265389565559367\\
1.09458117218602	0.236288878958164\\
0.932121924263184	0.33332632506362\\
1.11246698177474	-0.301124709330414\\
0.931587157509637	-0.0092965061889691\\
1.09178756851811	-0.548679824148092\\
0.929548289521547	0.00277241848339521\\
0.964785783528925	-0.23786141976382\\
0.918407413385166	-0.142146446072957\\
0.956310435366515	0.0613374924983488\\
0.734363972596005	-0.170582709934598\\
};
\addlegendentry{Estimation w/o constraint}

\addplot [color=color3, only marks, mark=triangle, mark size=2pt, mark options={line width=1,solid, color3}]
  table[row sep=crcr]{%
0.996695978448015	0.00462006097133169\\
0.997252163973545	0.0311924319132421\\
0.99295405081429	0.0841033544926554\\
0.975774813778717	0.108252770622138\\
0.994657087808835	-0.0413472715403006\\
0.929601203643061	0.324246935970391\\
0.99796185766787	0.014441824283368\\
0.929971843083342	-0.319820504179812\\
0.996953090480069	-0.00568261096881801\\
0.994033291751355	-0.0377008997694575\\
0.982682104164163	0.144504930504834\\
0.987438678424568	0.136987705354304\\
0.960285235328776	0.225895084495759\\
0.971360689773273	-0.177955004430827\\
0.95939073450623	0.203617109347638\\
0.996729277613285	0.0267501789951424\\
0.994706028983017	-0.0623225069000443\\
0.997213061959737	0.0129351885060437\\
0.964748015021155	-0.240083117087325\\
0.931814058810711	0.336507764917659\\
0.967944424327764	-0.186275404529289\\
0.986545929562936	0.0994737495588932\\
0.988010638150745	-0.0844118568972222\\
0.969096633957688	-0.192491434654842\\
0.99658840756823	-0.07086276879583\\
0.995434272418948	0.0469984956521712\\
0.95373475397617	-0.266593078929017\\
0.995186466466456	-0.0127276590680297\\
0.980389710115498	-0.115167443697403\\
0.989368916396324	0.0881404121671923\\
0.956619579857515	-0.248391176028689\\
0.990230662297243	0.0666471493208783\\
0.975548658904983	-0.0894706897727242\\
0.833755154835233	0.52008503879609\\
0.964919693304543	-0.215089138474991\\
0.997770606268929	0.00998556791533418\\
0.892544537622692	0.385907394526644\\
0.991493056162826	0.0621775976703224\\
0.992778084564409	0.0724816116763763\\
0.9033259489534	-0.407828059769413\\
0.989271005390867	0.0543401177620089\\
0.995904627957677	0.00123006411600491\\
0.987788275118565	-0.0722560001338842\\
0.986626234184965	-0.0340446120830326\\
0.991100712978181	0.0916654305485062\\
0.99116829116957	0.101396719980613\\
0.993915692523918	0.0876202602462648\\
0.994015940082524	-0.0667213437446387\\
0.984511909187142	0.0339919787855388\\
0.991433916633658	0.0663236288087673\\
0.984605407835641	0.12409922281547\\
0.995197640817012	-0.0740024302859042\\
0.991185291250137	0.112655615434629\\
0.994951698592976	0.079068092542327\\
0.968497121527219	-0.209829171980646\\
0.976660091827425	0.193391124028944\\
0.996540820479971	0.0214622650413838\\
0.98307938872274	0.146046354229833\\
0.971507817947994	-0.208397779595599\\
0.978380107581309	0.163878553059852\\
0.875675121496701	-0.461480685089107\\
0.986438855931231	0.119172745432236\\
0.973809126174504	-0.125160806136112\\
0.994592182493816	-0.00338522702834682\\
0.930424678129168	0.321297388141812\\
0.993725813741451	0.0923816448394377\\
0.978629701306015	0.162733712918311\\
0.98477193530262	0.054278169125944\\
0.852435706210773	-0.451099970426723\\
0.957545263967735	0.243570663433458\\
0.972112528057345	0.114368637223632\\
0.986543728516946	-0.144066085116627\\
0.996846106899789	-0.00301397300246876\\
0.954621296220837	0.231408924065355\\
0.995953171584981	0.0527467378118099\\
0.963622880193866	0.227528061154914\\
0.992239592918228	-0.0895833987888691\\
0.98756713461962	0.108709273186496\\
0.98631536372802	0.15152970143385\\
0.995585191160964	-0.0440668388603456\\
0.996712655381041	-0.00650272020198475\\
0.959511596901799	-0.257090419734504\\
0.992743633051795	0.0312665420845051\\
0.976316060402004	0.189452874557156\\
0.985296160964019	-0.153228441850122\\
0.989385113718057	0.137369746194933\\
0.979484218704864	0.120212485444137\\
0.997716787102666	0.0384751933430753\\
0.832617009076595	-0.512192221641405\\
0.986969359423567	-0.118373534160835\\
0.981759764462945	0.108012707048235\\
0.981628732579033	0.162505132193734\\
0.940252103847063	0.278421289918821\\
0.997272601133369	0.000220141942688942\\
0.877537501975625	-0.400283677419851\\
0.988169399167384	-0.00383014401776782\\
0.98558254031209	0.144731920298487\\
0.974247627023178	-0.168026284224477\\
0.993157264183829	-0.0452306989525756\\
0.996955117877353	-0.0272265934737056\\
0.997718472817927	-0.0177813615089516\\
0.994497861411091	-0.0858567340743541\\
0.901271946714986	0.380035140331168\\
0.971053237203096	0.21503265130425\\
0.988860613607701	-0.119202065611618\\
0.994609272101598	0.0578142199438787\\
0.920921556565093	0.300711342177509\\
0.970310693629768	-0.199048902246731\\
0.986916235523837	-0.138522172086372\\
0.962645609660866	-0.241775107593713\\
0.991555363403902	0.0321954401881582\\
0.992710231243851	0.0613353021930516\\
0.958533590589855	0.241326372560993\\
0.961659533896874	0.2518513560588\\
0.973206115411714	0.178874811489058\\
0.982284034456321	-0.170632082683418\\
0.986699084737104	-0.114795388437497\\
0.970843679004307	0.185214987756829\\
0.979305594168954	0.174592317066309\\
0.965794946697115	-0.0209708766761045\\
0.987105326426859	-0.145853919239429\\
0.97364430094396	-0.178664742308612\\
0.975832811311941	0.173668655940844\\
0.945403418474578	-0.300349484322873\\
0.986087531793999	0.148396557291479\\
0.890843078641843	-0.39977223375039\\
0.988405337849185	-0.108206959176391\\
0.998227997814927	-0.0234358445747848\\
0.951929018017616	0.223793640010177\\
0.962663216324598	0.213430466305341\\
0.995696108829021	-0.0397829854356723\\
0.981655744056421	-0.161507223695644\\
0.889741999625106	0.415226713250417\\
0.960063581550523	0.255578198686576\\
0.998035149979777	-0.019481905356619\\
0.973367865260923	-0.183091902623096\\
0.994227555979879	-0.0834819244817922\\
0.997891732069591	-0.029046023158792\\
0.995916557018586	-0.00109894671635649\\
0.968894828416171	0.223468129943355\\
0.977630685960741	0.184536597493883\\
0.96383682504428	-0.220560726668817\\
0.986738207049378	-0.0665680876102328\\
0.98967140493941	0.0245898695819871\\
0.985552014167926	-0.138257468425896\\
0.85811180614031	-0.477900545248269\\
0.955149732434724	0.260328163829954\\
0.871695152966998	-0.442880348214938\\
0.985189002599949	0.116752426626334\\
0.941262459497255	0.285988416821753\\
0.999141400698843	1.91621723345818e-05\\
0.972415596956593	-0.168508746409559\\
0.983601404324713	-0.142711541280966\\
0.926853994115447	-0.31797311759172\\
0.976257268896354	0.199116713498331\\
0.91801184376796	-0.315384690089595\\
0.99547642071229	-0.0690732356226439\\
0.98044222031814	0.138478265852573\\
0.991348688920138	-0.0246147466332326\\
0.984825276153069	-0.151506614773955\\
0.99497832506466	-0.0402015913923065\\
0.984060881726832	-0.0960595025977925\\
0.993667724132192	0.0921378419835096\\
0.99178159118589	-0.0850057346842801\\
0.99293217784457	0.0812002545162276\\
0.995411163261708	0.0154467348674206\\
0.925683683811655	-0.351143677958769\\
0.954482486368532	-0.270319202174642\\
0.991930590681523	0.0939718382527863\\
0.995267528386573	-0.0657702220242486\\
0.99567241554991	0.0539215722505602\\
0.981504846231395	-0.160706951998297\\
0.987632435143879	0.101055186666974\\
0.969023555015089	0.218541785475571\\
0.992316469399383	0.0857060655532028\\
0.99493897996017	0.024604132789004\\
0.980220974372605	-0.0814832503729981\\
0.984287654503959	-0.162003356912714\\
0.995081186136831	-0.038781661134051\\
0.984673152975882	0.139596726130273\\
0.971702322161414	0.212797209973049\\
0.992394724542376	-0.00290847121516016\\
0.993483427816785	0.0281266528341887\\
0.966686582073375	0.213632115930181\\
0.961518747022792	0.235098036479965\\
0.986036169628027	0.0919123537030924\\
0.997124374670342	0.0331114707345094\\
0.944257408092614	0.301909155628183\\
0.983841999427113	0.0937419358807279\\
0.940412419058408	0.313305431634148\\
0.944083554173241	0.288827875323424\\
0.938115085673296	0.321551598247501\\
0.978522184098082	-0.189208580182808\\
0.984529054226024	0.0245708894399603\\
0.883900029200145	-0.384659093602539\\
0.994524504980194	0.0580922201641511\\
0.968207182773099	-0.230311146226567\\
0.990556125093253	-0.0771993553474223\\
0.996008529836848	0.0357386744330737\\
0.954615968146903	-0.224279279575833\\
};
\addlegendentry{Ground Truth}

\addplot [color=color1, only marks, mark size=2pt, mark=x, mark options={line width=1,solid, color1}]
  table[row sep=crcr]{%
0.979879295854003	-0.01538068778253\\
0.9926294067112	0.121189359814178\\
0.992568125965916	0.121690243308611\\
0.989224816613352	0.146404447324115\\
0.99979148280425	0.0204203553367472\\
0.983291580632439	0.182037544087369\\
0.984986326926022	0.17263237172901\\
0.931038552633023	-0.305887583126554\\
0.978329001060769	0.0572045949503792\\
0.971302719294202	-0.237846647005353\\
0.965082575302496	0.261945839530812\\
0.972455393593243	0.121369302014184\\
0.80526256971191	0.592918370284621\\
0.988466090091255	-0.151442361113747\\
0.955865492927114	0.216150779390694\\
0.999954074511579	0.00958378149221589\\
0.998484305390352	0.0550371864210595\\
0.97366654954519	0.111236011690296\\
0.971214264839656	-0.130930713631167\\
0.964331363867096	0.17454231768321\\
0.999387846249232	-0.034984750525365\\
0.943366274582126	0.274490450317631\\
0.997762408092433	-0.0668593822704782\\
0.98185310877677	-0.189643014069571\\
0.978967641336668	-0.0449706261432948\\
0.969203271564074	-0.167385233133476\\
0.956066341588894	-0.215260657066795\\
0.990717338019387	-0.135938059967694\\
0.97881487292173	-0.0481813713712786\\
0.971308927899949	0.130226597060091\\
0.939549542143289	-0.278651498934308\\
0.979974540250365	-0.00706402584122\\
0.999944292455409	0.0105551876274779\\
0.790840206419751	0.612022685780491\\
0.963001247798894	-0.181737714131584\\
0.991517591769364	-0.129972555610334\\
0.908166976684437	0.418608101283109\\
0.976467650682866	0.0831319864425526\\
0.9703317209473	-0.137318430385183\\
0.897210942341281	-0.441602224794069\\
0.979439148582721	-0.0331504784815283\\
0.998217068231299	0.0596882290884039\\
0.997180583984344	0.0122646137906585\\
0.989585209502808	-0.142571872191028\\
0.999209275717975	0.0397595689005876\\
0.985206416990225	0.171371864438951\\
0.997956636601062	-0.0638948469275541\\
0.993245185395156	-0.11603448491437\\
0.979725263733502	0.0232036118378123\\
0.977369781539432	0.071751725648676\\
0.979170716515282	0.0403076657591265\\
0.977990847829014	-0.0627208224012266\\
0.97171896862486	0.127130822441449\\
0.989888227600398	0.0887644999512018\\
0.972393208130981	-0.121866520343935\\
0.960245653879162	0.195776107342495\\
0.997224002321048	0.0744599838489801\\
0.995360658748585	0.0962141310597475\\
0.978441355003212	-0.206524852788906\\
0.96475593739114	0.216643891700462\\
0.743512654708981	-0.638426920083735\\
0.952460709060775	0.233506103602922\\
0.936055683589955	-0.290171944231247\\
0.994696967551527	0.102849126120723\\
0.944458494637132	0.328630722708275\\
0.988723129961206	0.149755040915879\\
0.952888062569408	0.241308880732814\\
0.979940749885016	0.0107762105952171\\
0.669518713470313	-0.742795188671182\\
0.915167978944637	0.350524707138296\\
0.864237844920598	0.462052970347557\\
0.979709573832779	-0.118931613885144\\
0.997852151814503	0.0655063593795683\\
0.976955407607434	0.213443509028954\\
0.998794173505281	0.0490937773236328\\
0.908710838270598	0.366939521460345\\
0.993532080827779	-0.113551769541576\\
0.925906065225431	0.321088707958968\\
0.985175725842564	0.171548212496011\\
0.976984920129394	0.0768144897773934\\
0.995148216061734	-0.0983871336565297\\
0.94510556468312	-0.259182313449434\\
0.967143417201965	0.158220133241337\\
0.984448307535772	0.175674499543769\\
0.971960501476918	-0.235144176131876\\
0.98569525647576	0.0686103591380668\\
0.942885707522682	0.267145171301021\\
0.999627103890643	-0.0273066506039241\\
0.886465212926554	-0.417827029129374\\
0.980167715881647	-0.198169747295985\\
0.964160327749042	0.265320301512832\\
0.976282212045867	0.216501830114273\\
0.888371667843922	0.437250111101853\\
0.978142355293363	0.060311962172961\\
0.921343177396498	-0.388750240469282\\
0.974472314133849	0.103940891792507\\
0.981906355348075	0.189367128422702\\
0.98152781452939	-0.191319495360925\\
0.973964482381235	-0.108596441285402\\
0.995409709967682	-0.0957053253589096\\
0.989808672299556	-0.142403624394151\\
0.977927693481514	-0.0636979302797655\\
0.938492327439699	0.345300088816694\\
0.993972281838987	0.109631669401671\\
0.982419303984151	-0.120462152142651\\
0.963152916019798	0.18093219824719\\
0.943481312735998	0.265033983703897\\
0.969604062441581	-0.142365593093212\\
0.959543310676522	-0.199189946874686\\
0.951623259402941	-0.234122130870457\\
0.998786871386343	-0.0492421115132222\\
0.999955258172592	-0.00945947424462846\\
0.938767149386563	0.339480739542215\\
0.95054177337006	0.238475024014051\\
0.975679367810436	0.0919226371957822\\
0.982477439727094	-0.119346446641182\\
0.975649286341904	-0.219336431220474\\
0.977275734722752	0.0730214921930855\\
0.935933043309692	0.352177992556403\\
0.989448016270791	0.144888312495512\\
0.990706986187783	-0.120307306451207\\
0.946774985323712	-0.253016061081672\\
0.999260783496032	0.0384432902714721\\
0.988192135679069	-0.153219786522631\\
0.996064170625542	0.0886350269140366\\
0.886142895719675	-0.418510177134976\\
0.980608540488912	-0.195976759643091\\
0.999518529402165	0.0310275584236584\\
0.94895676910653	0.315406167293692\\
0.944884319624342	0.259987735341581\\
0.998478147028239	-0.0551487978749736\\
0.9899171490939	-0.141647583565007\\
0.931751145331291	0.363097511935606\\
0.903138802037073	0.429348697744666\\
0.978023977586434	-0.208492444146089\\
0.937957930667691	-0.283962885422696\\
0.987610123555837	-0.105634061877688\\
0.998260127178088	-0.0589637048224382\\
0.977712035027733	0.0669266506104162\\
0.964254212865052	0.264978891559309\\
0.95614319632033	0.21491902691559\\
0.962092190355181	-0.186490260489844\\
0.968706291028024	0.148351345537308\\
0.970312914647112	0.13745125561095\\
0.972386607936091	-0.233375844308151\\
0.837970026002647	-0.54571625916874\\
0.948188260770046	0.317709021177986\\
0.875072774603384	-0.483991362679062\\
0.972903957202795	0.117719539835754\\
0.925425494988523	0.378929615133533\\
0.994380706548521	0.105863168496247\\
0.993984037780795	-0.109525032010896\\
0.962126963116075	-0.186310780270598\\
0.981522965554521	-0.191344370413814\\
0.972979130484764	0.230893074043187\\
0.891867240662129	-0.406168468783238\\
0.974035573706968	-0.107956941200359\\
0.981497722067871	0.191473814333921\\
0.979730941606475	-0.0229626230838253\\
0.939535996452163	-0.278697167855436\\
0.965420285887414	-0.168415176267055\\
0.94729981852866	-0.251043928059549\\
0.988553172448015	0.150872877757947\\
0.999999967547279	-0.000254765460930583\\
0.999065404110125	0.0432240478237734\\
0.979906875949742	-0.0135097915015901\\
0.94855749767296	-0.314056778208205\\
0.946147098504892	-0.255354005237385\\
0.96948095660513	0.14320151807855\\
0.999960063469886	0.00893708371342054\\
0.978135644757681	-0.0604206955808642\\
0.982552598560719	-0.0769965555686542\\
0.979813794701493	-0.0191030812347256\\
0.995537809378786	0.0943634998147499\\
0.997836778843601	0.065740115508167\\
0.991687901988731	0.128666643109974\\
0.965043502951502	-0.170560949255958\\
0.980022430159365	-0.13921994776461\\
0.999698131674922	-0.0245691987185417\\
0.992461806022587	0.122554329121353\\
0.880006106101332	0.474962370324609\\
0.976399979019952	-0.215970092767116\\
0.977978193778712	0.0629178233359096\\
0.993374850406632	0.114919130607581\\
0.951925184361974	0.306330611231431\\
0.993096475293757	0.117300429492464\\
0.978996665307218	0.164265554023658\\
0.949559923502202	0.313585636913576\\
0.971171129636124	0.131250283661772\\
0.940735301369079	0.274621726667844\\
0.977483713738952	0.21101087501621\\
0.932121924263184	0.33332632506362\\
0.965263238214088	-0.261279315967516\\
0.979951207150186	-0.0097791413167535\\
0.893513298776575	-0.449036730022615\\
0.97999564120101	0.00292287991913888\\
0.964785783528925	-0.23786141976382\\
0.96846868110155	-0.149894675440534\\
0.977990385010657	0.0627280386008312\\
0.954585080161115	-0.221737062156505\\
};
\addlegendentry{Estimation w/ constraint: $\kappa_{\min}=0.98,\kappa_{\max}=1$}

\end{axis}

\end{tikzpicture}%

%% file: Results/CPE_threshold_tau20_4e17.tex
\begin{tikzpicture}[font=\scriptsize]
\definecolor{color2}{rgb}{0.12156862745098,0.466666666666667,0.705882352941177}
\definecolor{color0}{rgb}{1,0.498039215686275,0.0549019607843137}
\definecolor{color1}{rgb}{0.172549019607843,0.627450980392157,0.172549019607843}
\definecolor{color3}{rgb}{0.83921568627451,0.152941176470588,0.156862745098039}
\definecolor{color4}{rgb}{0.580392156862745,0.403921568627451,0.741176470588235}
\definecolor{color5}{rgb}{0.549019607843137,0.337254901960784,0.294117647058824}
\definecolor{color6}{rgb}{0.890196078431372,0.466666666666667,0.76078431372549}
\definecolor{color7}{rgb}{0.737254901960784,0.741176470588235,0.133333333333333}

\begin{axis}[%
width=5.82cm,
height=5.82cm,
at={(0,0)},
scale only axis,
xmin=-1.5,
xmax=1.5,
xlabel style={font=\color{white!15!black}},
xlabel={\scriptsize Real },
ymin=-1.5,
ymax=1.5,
ylabel style={font=\color{white!15!black}},
ylabel={\scriptsize Imaginary},
axis background/.style={fill=white},
title style={font=\bfseries},
xmajorgrids,
ymajorgrids,
legend style={at={(0.5,0.995)}, anchor=north, legend cell align=left, draw=white!15!black},
legend columns=1,
ytick distance={1},
]

\addplot [color=color2, only marks, mark size=2pt, mark=o, mark options={line width=1,solid, color2}]
  table[row sep=crcr]{%
0.815270359534969	0.00856927934318671\\
0.468045875895296	-0.726854339282604\\
0.420159991508797	-0.800717250688576\\
0.388235691760989	0.884129123604732\\
0.685245677424843	-0.750325663933955\\
0.7862766373636	-0.00548280922851711\\
1.10923026599664	-0.305669073426212\\
-0.74562374611675	-0.419949492777578\\
0.448090167758297	-0.933057784912561\\
0.858422642939122	0.201731789325186\\
0.539516708611252	-0.901419569400261\\
0.146614127915806	1.02493210188947\\
0.605970815941428	-0.934179836608217\\
0.517257552018666	0.515266723701169\\
-0.0214417788781569	0.667295384788629\\
-0.74469075877945	-0.00132892169047108\\
0.316666551442343	-0.936940420954266\\
0.395845044035477	0.330246484879003\\
-0.505876834092576	-1.08962435196377\\
0.887428295114596	-0.53287755580365\\
-0.89585281186397	-0.176299704765621\\
-0.850125713496821	0.374859889460226\\
0.405689641269629	-0.645014138655452\\
0.551573843669801	-0.449806929404218\\
0.684035156019878	-0.621001915848283\\
0.749610495951393	0.0424918719422545\\
0.0100461017740354	-0.843962124326655\\
0.363771977552961	0.52762855731992\\
-0.260579280877481	-0.991954690146716\\
-0.12038000444853	-0.796866845398118\\
0.290507636154988	-0.495059100848913\\
0.724611065367888	-0.805298782869113\\
-0.118582558189009	-0.808552631701175\\
0.795001409569525	0.301963738338147\\
-0.523785170308366	-0.813028243872278\\
-0.937628895846347	0.255505444852263\\
-0.42764136462652	-0.784513946021773\\
-0.0723369398503911	0.913385894249403\\
-0.0938369311925446	-0.575720952688422\\
-0.961244494030194	-0.023894416227901\\
-0.340463895553553	0.942314193030893\\
0.67131919181593	-0.0726090697043687\\
-0.217801966234671	-0.855230723245395\\
-0.949507454850072	0.278476384003687\\
0.0765268183774307	1.00923738367307\\
-0.120769953505525	-0.827262792730227\\
0.647332351391424	-0.247975399308774\\
1.0334414247756	0.324788892654439\\
0.697926864269139	-0.738462909527644\\
-0.356841021082643	0.699152760420461\\
0.312525223004061	1.06315361651847\\
0.0876691491227026	0.978386227548697\\
0.818414358932177	-0.618843914021161\\
-0.65650001577365	0.366816851580172\\
-0.34747441875844	-0.738036211943391\\
0.111149075672136	0.941655064281268\\
-0.809486763228534	0.490411921084588\\
0.382264722926838	-1.41484984144527\\
0.100190639244213	-0.672822923747048\\
0.694350986815688	0.751004875184281\\
0.696587797033819	-0.340590093607379\\
0.816111794686323	0.916637157512086\\
-0.675611237715138	0.280158704684236\\
-0.732879329619471	-0.125848937671488\\
0.272549885549201	1.01967229344464\\
-1.00115418681229	0.21239878920464\\
-0.558476916577567	-0.907889848608244\\
0.126046848796314	1.05993269259106\\
-0.0229622215330156	1.06647149766356\\
0.625260861605004	0.375262907205333\\
0.593940943921385	-0.797568878326611\\
1.00648971041856	0.179080943404486\\
0.491617666625096	0.649584643109315\\
-0.912645595696644	0.480268211593519\\
0.00257860581836172	-0.902334962130139\\
0.397092065896667	-0.772014743655751\\
-0.891363594688855	0.116995138784878\\
-0.653907323669359	-0.460107554864331\\
0.990257262333153	-0.0936883133964588\\
0.910501284126781	0.256234951191671\\
0.722770770708273	0.568487475268339\\
0.666043763810294	0.540739553361079\\
0.72697518719116	0.488329386676058\\
0.711808532394017	-0.471450511347242\\
0.249024586695558	0.828840679805521\\
-0.981034458107332	-0.398142268554028\\
0.730956089364356	0.588613417169071\\
-0.69115475187217	0.509798716946769\\
-0.0578449856525261	-1.22236801640334\\
-0.566676160702131	0.71461740880681\\
0.518476325114648	-0.742167245050833\\
0.420944009657215	0.159075301267785\\
0.75497657967078	0.49434731238954\\
-0.432859945218366	-0.839715613994997\\
-0.973265461334162	0.390167834244451\\
-0.85413999098632	0.152904674117602\\
-0.304037068606348	0.783599820536883\\
-0.654837422509421	-0.318404623182509\\
0.426697745406802	0.670226105723283\\
-0.59342080229266	-0.739761169225567\\
0.642111551368452	0.709963541220208\\
-0.179637535158282	-0.802552013931309\\
0.0194792489998718	0.852561376443623\\
-0.263435025252087	-0.812639284937538\\
0.591353696627882	-1.03841384395339\\
-0.479243853640748	0.472273697415454\\
-0.280839419342159	-0.981714879405451\\
-0.316277637634224	-0.838281925185318\\
-1.07258306761593	0.146708326486319\\
1.16172225367477	0.0668166742339553\\
0.035462907743344	-0.966885744622382\\
-0.41177659401865	-0.691101225976965\\
-0.484824636061226	-0.656594488829105\\
-0.797977547461616	-0.227794623864134\\
-0.304222450051769	-0.924662723979672\\
0.203756978458025	0.944287177261036\\
-0.94403106663192	0.499059190197302\\
-0.674226373790155	0.617040391830302\\
-0.710635606029976	-0.518557934469627\\
-0.315054688496914	0.829915201962956\\
0.968309440376051	-0.606591268758447\\
0.736525092470297	0.317534085495465\\
-0.104201951124386	-0.935867537645822\\
-0.380091409832702	-1.01846275875408\\
0.698752073505981	-0.859272382668177\\
-0.769583459694462	0.559620306505417\\
-0.65465877336957	-0.76282918278454\\
1.00316461136786	-0.225380736277348\\
-0.118521495652971	-0.905057581826754\\
-0.338141596028944	0.801204698715684\\
-0.199530940092036	0.754771859401887\\
-0.0269589990903376	-0.797244853283391\\
0.461530749876581	0.777849180247771\\
-0.758646575074449	-0.820307978578861\\
-0.813237153508559	0.342003346128911\\
-0.52061596343947	-0.670776084378967\\
-1.04314928289422	0.229893784823162\\
0.983674502546223	0.135376043581096\\
0.288814763835992	0.779230261956127\\
0.4614756899327	-0.274807395820771\\
1.00416154355633	0.0890292569106672\\
-0.000154116672938284	-0.899578855806077\\
0.700952341575257	1.03863081415585\\
-0.210748444485697	0.915062570331546\\
-0.680882712012644	-0.893741135794977\\
-0.69415497872632	0.247893269705716\\
-0.307773839249931	1.09962442099442\\
0.224733947976078	0.927205713283545\\
-0.569033449264283	0.849325865559706\\
0.475467061658771	0.641682088203749\\
0.583824882434961	0.40954259488674\\
-0.740494812399341	-0.0800312994936972\\
-0.473951855130241	0.446054525009637\\
0.661141080128226	0.0406573591443363\\
0.414202272391012	-0.79963409622498\\
-0.0541163318405786	-0.989714186789127\\
0.85742573613365	-0.450612804068275\\
1.25425490235933	0.0880001086819612\\
0.557545309115963	-0.410703349088724\\
-0.507125633058019	-0.896733117857511\\
-0.919517400722154	-0.518892028576246\\
0.849127481013165	0.341800700999478\\
0.723824519824414	0.61332628585675\\
-0.828565501348679	-0.489290122880717\\
0.871004660410338	-0.344435447564813\\
0.886719363697738	-0.627805467103896\\
0.719890909921764	-0.525520358745103\\
-0.740385562360228	-0.691455472066222\\
0.988544131376461	0.471934423238591\\
0.0015940483828643	-0.897901287922665\\
0.783568972804096	0.873362397086955\\
0.276461859238004	-0.743418296642984\\
0.877414070733035	0.247574424835021\\
0.614272699855158	-0.767574807403509\\
-0.831485517873028	-0.103130662078067\\
-0.00970468663430162	0.818699392690496\\
-0.851212233977691	-0.0490594676977651\\
0.746455345122868	0.531976594782405\\
-1.21529364842358	0.226955796992967\\
-0.286406699027152	-0.888675765298342\\
-1.05854956789464	0.0261944850038253\\
-0.644348419545413	0.994531814174059\\
-1.04637058662847	0.242095653552407\\
-0.180049188345336	-1.1809200791949\\
1.03463399214466	-0.334668698500402\\
-1.09778788172248	-0.129621684139936\\
-0.817860001406055	-0.304572142366108\\
0.725454834490288	-0.514716275270116\\
-0.81006481308742	-0.0129535964518234\\
0.170928048148321	0.976192923168408\\
-0.830384927653402	0.153988386444995\\
0.693139496914555	-0.872045672402135\\
-0.169785635596229	0.860089266918999\\
1.06455192340991	-0.0549925673389194\\
-0.462228370402268	-0.465642647247505\\
0.987076215692296	-0.319134591806814\\
0.890674349535539	0.302863550059795\\
0.081467364643391	-1.00291511424026\\
1.00445685059134	0.51687869486825\\
-0.445156614032346	-0.844574024428547\\
};
\addlegendentry{Estimation w/o constraint}

\addplot [color=color3, only marks, mark=triangle, mark size=2pt, mark options={line width=1,solid, color3}]
  table[row sep=crcr]{%
0.933197459262604	-0.169964685131307\\
0.454178856484086	-0.883347110990089\\
0.815847931602525	-0.395356721122835\\
0.733848328516339	0.481731962494054\\
0.692521260638403	-0.689085293995277\\
0.942934269730384	0.194785099189327\\
0.897923305573687	0.0929432728151406\\
-0.793551857248835	-0.5959408496212\\
0.481526135301377	-0.82563943820207\\
0.984996822253374	0.0775756669441214\\
0.40055734607383	-0.87459492246679\\
-0.00438927420071169	0.984148863488636\\
0.616153923564483	-0.770191497275928\\
0.279786558072526	0.864024668197685\\
-0.422618278595614	0.860815087714087\\
-0.967216417409883	-0.0175104018751507\\
0.347928108461348	-0.905422825373046\\
0.859227669932077	0.4389002411435\\
-0.137163107857943	-0.943119667032604\\
0.935276897598289	-0.296855630900452\\
-0.953363459071864	-0.24255773379532\\
-0.793196186222078	0.601398983658043\\
0.523969450318219	-0.80878371382618\\
0.647757224881913	-0.712571918492267\\
0.832246172316039	-0.520182684533524\\
0.977304357730751	0.0490553838647355\\
0.239229018372364	-0.949437839401928\\
0.291599226079828	0.796371127760976\\
-0.222294499632208	-0.938450309422491\\
-0.16683126925393	-0.974424402499796\\
0.615133079823664	-0.678818175228768\\
0.627399046943287	-0.75143829062577\\
0.125527396071104	-0.957366924271242\\
0.895747158743393	0.372548840954968\\
-0.476630806815436	-0.78468405998002\\
-0.930779197309154	0.265319098665202\\
-0.600122144502459	-0.747918833688096\\
-0.0347113821879426	0.983770359253317\\
-0.254680434899521	-0.885365193760374\\
-0.990083833850296	-0.00527888075444934\\
-0.4216350573639	0.885793939723695\\
0.924288543136898	0.157544998222001\\
-0.105704248700482	-0.984691623486412\\
-0.989792180827061	0.108433993182869\\
0.0409393136150903	0.979298069794375\\
-0.287677013600894	-0.934235352418086\\
0.870246550534199	-0.3851112903002\\
0.861699665284271	0.374857841403673\\
0.844484539676128	-0.468848934429706\\
-0.537482510509798	0.769192821826523\\
0.33642337861407	0.904367508646663\\
0.0487288498486623	0.993551242698467\\
0.680650140307806	-0.680515837232687\\
-0.925910660946589	0.321908052033397\\
-0.106507931879133	-0.898699149143203\\
0.135158088177175	0.961981817754602\\
-0.912723920853334	0.328217482923647\\
0.470018491261508	-0.796774624690165\\
0.197075010049407	-0.929002470675194\\
0.573463523855305	0.752924022716485\\
0.831683104054673	-0.491450723992598\\
0.655589961175868	0.651102739119208\\
-0.947969046692269	0.160566585100187\\
-0.956413301817343	-0.0964544088295213\\
0.335613919794682	0.891168908398024\\
-0.922878121911667	0.0586365271632443\\
-0.471953767327105	-0.858687608286493\\
0.300986077998971	0.86726071502423\\
0.0833487121793205	0.955671872542314\\
0.609282462220706	0.72943402843187\\
0.436040077491416	-0.880943595665118\\
0.916302796119245	0.255214096241522\\
0.623604631917142	0.733423496217275\\
-0.734639955529037	0.667723049191469\\
-0.104430887751406	-0.908519464500027\\
-0.148379961378836	-0.949706121984422\\
-0.986329459762247	0.0722697588729704\\
-0.86762201846569	-0.483159218439461\\
0.801075244751202	-0.460213735106192\\
0.956783989352435	0.130754127868796\\
0.618104748667842	0.72763114277662\\
0.739736585935747	0.638152523492199\\
0.872967353767294	0.471477550730549\\
0.903432436078726	-0.322143029927741\\
0.307591709434824	0.938500339990376\\
-0.883288854123852	-0.364988625726501\\
0.779607585800853	0.611701952000475\\
-0.798302125292184	0.571403959252499\\
-0.255487050104505	-0.928529520853162\\
-0.613859475595733	0.770294312214813\\
0.539651761759062	-0.827394491745046\\
0.897542525072938	0.206432315150264\\
0.864150818761972	0.475628711496719\\
-0.177425158678845	-0.970262153850837\\
-0.944445622976267	0.27077679238904\\
-0.96294412847119	0.229835718528568\\
-0.733171451189572	0.660123370675169\\
-0.921889010725558	-0.303455625458992\\
0.74316059491533	0.594283783809673\\
-0.579118912332022	-0.785173561812696\\
0.636344656100302	0.729457479878163\\
-0.505571999810602	-0.807735642259028\\
-0.405010462727357	0.798108860418053\\
-0.424332323379278	-0.87706378618028\\
0.606769434420307	-0.768953772824877\\
-0.691147599051479	0.690022401106576\\
-0.272394227288218	-0.898048537967007\\
-0.513558124514719	-0.849504682351164\\
-0.995194999716679	0.0303696599239457\\
0.963633203898887	0.150868379336606\\
0.011802942495498	-0.982741623829224\\
-0.347044890220617	-0.86251492874871\\
-0.596960122850372	-0.779203738049803\\
-0.93430943841401	-0.257933552205389\\
-0.348186581352515	-0.929487456378175\\
0.239691281035652	0.945273209952026\\
-0.880807172973102	0.462361219827528\\
-0.682896155556754	0.701296020087207\\
-0.790789800154432	-0.575795638542517\\
-0.385640247316986	0.879976425448455\\
0.640224967296744	-0.630043787201862\\
0.946016258689791	-0.010980164150427\\
-0.0490143705317026	-0.901977865045308\\
-0.393077314888606	-0.868523031622294\\
0.720765272785478	-0.657584610622744\\
-0.730238395788691	0.645500497729903\\
-0.652457119260177	-0.741800428744497\\
0.975795567464727	-0.0696158246785913\\
-0.209592294871951	-0.939530752697722\\
-0.399047172282573	0.906583024822322\\
-0.0547399498816541	0.971675926442996\\
0.0658067882830553	-0.972597683037575\\
0.533439186246283	0.777092536484252\\
-0.701115564830009	-0.699406523127741\\
-0.721742061633915	0.581008634839249\\
-0.424671220032463	-0.877324293569801\\
-0.844063078924631	0.429861627492059\\
0.969368697965175	0.21912810506795\\
0.0970489658350413	0.862305784292112\\
0.529260001984615	-0.79198053259179\\
0.923334402736316	0.341825734087027\\
-0.104360780084187	-0.968535181524505\\
0.594597603542231	0.772688144558471\\
-0.381877013595495	0.903948119273467\\
-0.592144730593296	-0.785814509109493\\
-0.939610002358235	0.242985221248585\\
-0.311457077434612	0.935674366387817\\
0.367827696464094	0.907354066680904\\
-0.573988406915061	0.803812603365742\\
0.443945297307245	0.837996147285524\\
0.788055934504217	0.57017489391111\\
-0.887659259911601	-0.381451884269455\\
-0.491702786927812	0.789148477296168\\
0.93918095276964	-0.330053968374668\\
0.426260037869014	-0.896767449387775\\
-0.0953812512747943	-0.976834101848077\\
0.89332661465998	-0.285138555352209\\
0.902166461662528	0.00607089624983987\\
0.199920511791469	-0.875937234563202\\
-0.549716952794731	-0.826246150782072\\
-0.598969666586056	-0.714195494724761\\
0.626874494623458	0.723842137339124\\
0.724696801150651	0.674439278171708\\
-0.847395261744309	-0.491020738732472\\
0.874482454482118	-0.303769275856691\\
0.795861502057392	-0.587590017599032\\
0.675699406051957	-0.709413143303204\\
-0.718900083149748	-0.669666612268108\\
0.890507339728143	0.414796026285598\\
-0.149415160992745	-0.968366463719755\\
0.669529192684692	0.68555305618201\\
0.277236812294329	-0.87063802734519\\
0.982088360652815	0.0207048927466173\\
0.591367706112824	-0.772240562584004\\
-0.759784312268962	-0.537437475552853\\
0.050447825494786	0.976275489217941\\
-0.988636332366643	0.00805296142613141\\
0.872798171643734	0.474660466407404\\
-0.943879552214422	0.0788757900446583\\
-0.284830753542262	-0.863640339435967\\
-0.976562403122365	0.114749367130755\\
-0.261151152774463	0.93936934752029\\
-0.98997445394015	-0.0148416027961686\\
0.24338287529467	-0.853423297119595\\
0.808214809770442	-0.572562883604943\\
-0.96596227907127	-0.0761444563519646\\
-0.920117468091663	-0.361204358475769\\
0.738811543911609	-0.588863491853821\\
-0.916233717193606	0.0641241645651378\\
0.124246694705156	0.962900304741115\\
-0.956417106857819	0.134539328691475\\
0.366548271454876	-0.897541798717398\\
-0.234223453259154	0.953882337596895\\
0.936120055310544	-0.272044855368938\\
-0.613362971610388	-0.747352614967528\\
0.901914370618542	-0.208404612208137\\
0.978397276248611	0.184799206655042\\
0.259941410940741	-0.940388177795171\\
0.844104111615231	0.490490066658971\\
-0.513103442081889	-0.826494399308673\\
};
\addlegendentry{Ground Truth}

\addplot [color=color1, only marks, mark size=2pt, mark=x, mark options={line width=1,solid, color1}]
  table[row sep=crcr]{%
0.89995028793267	0.00945934722399972\\
0.48725800653414	-0.756689920025618\\
0.420159991508797	-0.800717250688576\\
0.388235691760989	0.884129123604732\\
0.674358302212412	-0.738404279671504\\
0.899978119753875	-0.00627566444934086\\
0.964065101391174	-0.265666106757385\\
-0.784176930544177	-0.441663380418066\\
0.432905712804	-0.901439206947236\\
0.876132307787743	0.205893611485456\\
0.513561033953144	-0.858053066194031\\
0.141606168807712	0.989923074261633\\
0.544201425706604	-0.838954592489307\\
0.637621964427017	0.635167875825149\\
-0.0289042107267746	0.89953573948024\\
-0.899998566959952	-0.00160607286034576\\
0.316666551442343	-0.936940420954266\\
0.691076539322448	0.576552874243212\\
-0.421097482803617	-0.907015385744066\\
0.857313620804329	-0.514794478975222\\
-0.89585281186397	-0.176299704765621\\
-0.850125713496821	0.374859889460226\\
0.479167874445334	-0.761838662775486\\
0.697478355887373	-0.568791651722006\\
0.684035156019878	-0.621001915848283\\
0.898557526479434	0.0509349743021593\\
0.0107123896260608	-0.899936244802097\\
0.510855173452268	0.740963556294811\\
-0.254072509085676	-0.967185173649239\\
-0.134434662247019	-0.889902984367695\\
0.455498257205215	-0.776222479501213\\
0.668884230210942	-0.743366589627968\\
-0.130597204331742	-0.890474238942786\\
0.841353285192381	0.319569475225635\\
-0.523785170308366	-0.813028243872278\\
-0.937628895846347	0.255505444852263\\
-0.430753158267388	-0.790222574116097\\
-0.0723369398503911	0.913385894249403\\
-0.144780764818647	-0.888278408010984\\
-0.961244494030194	-0.023894416227901\\
-0.339806715774016	0.94049529287226\\
0.894781495959475	-0.0967784815365668\\
-0.222113707772637	-0.872161396084172\\
-0.949507454850072	0.278476384003687\\
0.0756093293343495	0.997137517756508\\
-0.130010560453702	-0.890560078922536\\
0.840444673686071	-0.321951472232563\\
0.953995560887629	0.299820729447945\\
0.686878208784993	-0.726772540961971\\
-0.40914196096517	0.801625134197759\\
0.28202760208947	0.959406291234148\\
0.0876691491227026	0.978386227548697\\
0.797639257234798	-0.603134823499622\\
-0.785675023958833	0.438992889153443\\
-0.38336479273327	-0.814267422713556\\
0.111149075672136	0.941655064281268\\
-0.809486763228534	0.490411921084588\\
0.260828187873241	-0.965385237307243\\
0.132558129672159	-0.890184442830709\\
0.678868868728918	0.734259531140536\\
0.808529409287473	-0.395322898802043\\
0.664965899170615	0.746873719540469\\
-0.831356052795339	0.344742097053076\\
-0.887017166506146	-0.152317255501165\\
0.258226304235885	0.96608445583224\\
-0.978227629096344	0.207534830022685\\
-0.523944631480335	-0.851752325000018\\
0.118087613079544	0.99300318007405\\
-0.0215260332445427	0.999768288101175\\
0.771685698460149	0.463142723997768\\
0.593940943921385	-0.797568878326611\\
0.984537306217453	0.175175034368718\\
0.543126571337162	0.717644429719578\\
-0.884946825108831	0.465692083602244\\
0.00257860581836172	-0.902334962130139\\
0.411659049804966	-0.800335446368379\\
-0.892346311408417	0.11712412439712\\
-0.736051973662176	-0.517906837247796\\
0.990257262333153	-0.0936883133964588\\
0.910501284126781	0.256234951191671\\
0.722770770708273	0.568487475268339\\
0.698718702303454	0.567267287133127\\
0.747095570178117	0.501844805712119\\
0.75034451112933	-0.496973957686001\\
0.258968352401678	0.861937000281552\\
-0.926599241885484	-0.376050322346952\\
0.730956089364356	0.588613417169071\\
-0.724286530441973	0.534236859286524\\
-0.047269173481611	-0.998882187867201\\
-0.566676160702131	0.71461740880681\\
0.518476325114648	-0.742167245050833\\
0.841890491483277	0.318151536771467\\
0.75497657967078	0.49434731238954\\
-0.432859945218366	-0.839715613994997\\
-0.928193059850196	0.372098970229603\\
-0.885916587292863	0.158593191402934\\
-0.325554002403338	0.839055773783345\\
-0.809391758951453	-0.393554291733011\\
0.483341300979645	0.759197725739023\\
-0.59342080229266	-0.739761169225567\\
0.642111551368452	0.709963541220208\\
-0.196585208943058	-0.878267758502391\\
0.0205577574367665	0.899765179704778\\
-0.277536343933857	-0.856138760830175\\
0.494860547421086	-0.868972403823104\\
-0.641040699669946	0.631717358766297\\
-0.275037518396033	-0.961433494046547\\
-0.317703020919768	-0.842059849712865\\
-0.990774807714731	0.135518561082377\\
0.99835009384114	0.0574202936894905\\
0.035462907743344	-0.966885744622382\\
-0.460671485689607	-0.773163490002296\\
-0.534606514338887	-0.724013725578753\\
-0.865428528054377	-0.247049514935031\\
-0.304222450051769	-0.924662723979672\\
0.203756978458025	0.944287177261036\\
-0.884067336923176	0.467359544446847\\
-0.674226373790155	0.617040391830302\\
-0.727018670250167	-0.530512820870221\\
-0.319418594835863	0.841410578298777\\
0.847448009089647	-0.530878396518444\\
0.826464602517533	0.356309220741669\\
-0.104201951124386	-0.935867537645822\\
-0.349645421226287	-0.936882105398269\\
0.630915231171296	-0.775851771330111\\
-0.769583459694462	0.559620306505417\\
-0.651253005613225	-0.758860674089613\\
0.975678714791908	-0.219205486934088\\
-0.118521495652971	-0.905057581826754\\
-0.349947618261025	0.82917830680345\\
-0.230021452541379	0.87010926404145\\
-0.0304163004953749	-0.899485880191665\\
0.461530749876581	0.777849180247771\\
-0.67897527318691	-0.73416113926083\\
-0.829622188606757	0.348894001340429\\
-0.551820611832281	-0.71098102109483\\
-0.976565636163489	0.21521979059696\\
0.983674502546223	0.135376043581096\\
0.312783822707059	0.843899449136424\\
0.773275583891575	-0.460483301930856\\
0.996092696622154	0.0883138705753852\\
-0.000154188821484263	-0.899999986792115\\
0.55940516576509	0.828894360286841\\
-0.210748444485697	0.915062570331546\\
-0.606007784943235	-0.795458713314647\\
-0.847575154510007	0.302681941082282\\
-0.269531663686695	0.962991527621236\\
0.224733947976078	0.927205713283545\\
-0.556605839186191	0.83077670873938\\
0.53581175228041	0.723122234562178\\
0.736795407102103	0.51684865103166\\
-0.894789208185252	-0.0967071502796468\\
-0.655392284220876	0.616815169871609\\
0.898303034930284	0.0552418087596742\\
0.414202272391012	-0.79963409622498\\
-0.0541163318405786	-0.989714186789127\\
0.85742573613365	-0.450612804068275\\
0.997547748443715	0.0699892104175673\\
0.724624422799551	-0.533778461426103\\
-0.492260639380537	-0.870447851922598\\
-0.870901487566683	-0.491457626814499\\
0.849127481013165	0.341800700999478\\
0.723824519824414	0.61332628585675\\
-0.828565501348679	-0.489290122880717\\
0.871004660410338	-0.344435447564813\\
0.816149239982234	-0.577841170285073\\
0.726918294395569	-0.530650349357312\\
-0.730844026469534	-0.682544510617292\\
0.902435159444972	0.430825699091324\\
0.00159777171844181	-0.899998581735291\\
0.6678069224343	0.744334544642951\\
0.313701956501623	-0.84355858272384\\
0.877414070733035	0.247574424835021\\
0.614272699855158	-0.767574807403509\\
-0.893156112198193	-0.110779778132152\\
-0.0106676571172839	0.899936776163542\\
-0.898508915908937	-0.0517854036592051\\
0.746455345122868	0.531976594782405\\
-0.983005511854599	0.183576043272203\\
-0.286406699027152	-0.888675765298342\\
-0.999693967227588	0.0247380655825501\\
-0.543743598938676	0.839251391785094\\
-0.974263431521464	0.225412435326938\\
-0.150723410180011	-0.988575972610961\\
0.951462246563419	-0.307765484361862\\
-0.993101159708168	-0.117260763200189\\
-0.843414649803993	-0.314088727107497\\
0.734014991081609	-0.520789777998249\\
-0.899884954290983	-0.0143898937006353\\
0.170928048148321	0.976192923168408\\
-0.884913059994173	0.164100201863831\\
0.622230365905307	-0.782834191732417\\
-0.174300532668676	0.88296054516123\\
0.998668392728691	-0.0515891593718307\\
-0.634050070741181	-0.638733518607802\\
0.95150480252456	-0.307633890806421\\
0.890674349535539	0.302863550059795\\
0.0809638911636906	-0.996717035234993\\
0.889179447249115	0.457558641694983\\
-0.445156614032346	-0.844574024428547\\
};
\addlegendentry{Estimation w/ constraint: $\kappa_{\min}=0.90,\kappa_{\max}=1$}

\end{axis}

\end{tikzpicture}%